\documentclass[letterpaper, 10 pt, conference]{ieeeconf}  

\IEEEoverridecommandlockouts                              
\title{\Large \bf
A Martingale Approach and Time-Consistent Sampling-based Algorithms \\for
Risk Management in Stochastic Optimal Control
}

\author{Vu Anh Huynh \and Leonid Kogan \and Emilio Frazzoli
\thanks{Huynh and Frazzoli are affiliated with or members of the Laboratory for Information and Decision Systems, Kogan is with the Sloan School of Management, Massachusetts Institute of Technology, 77 Massachusetts Ave., Cambridge, MA 02139.   
{\tt\small huyn0002@gmail.com,$\lbrace$lkogan2, frazzoli$\rbrace$@mit.edu}}
}

\usepackage{amsmath}
\usepackage{amssymb}
\usepackage[vlined,linesnumbered,boxruled]{algorithm2e}
\usepackage{graphicx}
\usepackage{epsfig}
\usepackage{subfigure}
\usepackage{multicol}
\usepackage{color}
\usepackage{cite}

\renewcommand{\baselinestretch}{.983}

\newcommand{\reals}{\ensuremath{\mathbb{R}}}

\newcommand{\naturals}{\ensuremath{\mathbb{N}}}

\newcommand{\EE}{\ensuremath{\mathbb{E}}}
\newcommand{\Cov}{\ensuremath{\mathrm{Cov}}}
\newcommand{\given}{\ensuremath{\,\vert\,}}

\newcommand{\FirstExit}[2]{{\ensuremath{T_{#1}^{#2}}}}
\newcommand{\CostToGo}[2]{{\ensuremath{J_{#2}(#1)}}}

\newcommand{\OptimalCostToGo}[1]{{\ensuremath{J^*(#1)}}}

\newcommand{\DummyS}{Y}

\newcommand{\dummyrate}{\varsigma}
\newcommand{\dummyasyn}{\theta}

\newcommand{\plim}{\ensuremath{\mathrm{plim}}}

\definecolor{darkgreen}{rgb}{0,0.5,0} 
\definecolor{fullred}{rgb}{0.85,.0,.1} 
\definecolor{brown}{rgb}{0.65,0.16,0.16} 
\newcommand{\todofootnote}[1]{{\let\thefootnote\relax\footnotetext{#1}}}

\usepackage{theorem}
\newtheorem{theorem}{Theorem}
\newtheorem{lemma}[theorem]{Lemma}

\begin{document}

\maketitle

\begin{abstract}
In this paper, we consider a class of stochastic optimal control problems with risk constraints that are expressed as bounded probabilities of failure for particular initial states. 
We present here a martingale approach that diffuses a risk constraint into a martingale to construct time-consistent control policies. The martingale stands for the level of risk tolerance that is contingent on available information over time. By augmenting the system dynamics with the controlled martingale, the original risk-constrained problem is transformed into a stochastic target problem. 
We extend the incremental Markov Decision Process (iMDP) algorithm to approximate arbitrarily well an optimal feedback policy of the original problem by sampling in the augmented state space and computing proper boundary conditions for the reformulated problem.
We show that the algorithm is both probabilistically sound and asymptotically optimal. 
The performance of the proposed algorithm is demonstrated on motion planning and control problems subject to bounded probability of collision in uncertain cluttered environments.

\end{abstract}

\section{Introduction}
Controlling dynamical systems in uncertain environments is a fundamental and essential problem in several fields, ranging from robotics~\cite{Kuwata.Teo.ea:TCST09,thrun2005}, healthcare~\cite{todorov.neural_comp05,alterovitz.simeon.ea.robotics.2008} to management science, economics and finance~\cite{fleming.stein.jour_bank_finance04,sethi.thompson.book06}.
Given a system with dynamics described by a controlled diffusion process, a stochastic optimal control problem is to find an optimal feedback policy to optimize an objective function.
Risk management has always been an important part of stochastic optimal control problems to guarantee safety during the execution of control policies. For instance, in critical applications such as self-driving cars and robotic surgery, regulatory authorities can impose a threshold of failure probability during operation of these systems. Thus, finding control policies that fully respect this type of constraint is important in practice.


There has been intensive literature on stochastic optimal control without risk constraints. Even in this setting, it is well-known that closed-form or exact algorithmic solutions for general continuous-time, continuous-space stochastic optimal control problems are computationally challenging~\cite{Vincent.Tsitsiklis.Complexity.2000}. 
Thus, many approaches have been proposed to investigate approximate solutions of such problems.
Deterministic approaches such as discrete Markov Decision Process approximation~\cite{chow.tsitsiklis.autocontrol.1991,munos.machinelearning.2001} and solving associated Hamilton-Jacobi-Bellman (HJB) PDEs~\cite{Lars.FiniteDifference.HJB.1997,Wang:2003:NSH:859448.859457,Boulbrachene.FiniteEllement.HJB.2004} have been proposed, but the complexities of these approaches scale poorly with the dimension of the state space. 
In \cite{Rust-97,Rust:97b, Vincent.Tsitsiklis.Complexity.2000}, the authors show that randomized algorithms (or sampling-based algorithms) provide a possibility to alleviate the curse of dimensionality by sampling the state space while assuming discrete control inputs.
Recently, in \cite{huynh.karaman.ea:icra12,Huynh2012.ArXiV}, a new computationally-efficient sampling-based algorithm called the incremental Markov Decision Process (iMDP) algorithm has been proposed to provide asymptotically-optimal solutions to problems with continuous control spaces. 

Built upon the approximating Markov chain method \cite{Kushner2000,kushner1977probability}, the iMDP algorithm constructs a sequence of finite-state Markov Decision Processes (MDPs) that consistently approximate the original continuous-time stochastic dynamics. 
Using the rapidly-exploring sampling technique \cite{Lavalle.98} to sample in the state space, iMDP forms the structures of finite-state MDPs randomly over iterations. Control sets for states in these MDPs are constructed or sampled properly in the control space. The finite models serve as incrementally refined models of the original problem. 
Consequently, distributions of approximating trajectories and control processes returned from these finite models approximate arbitrarily well distributions of optimal trajectories and optimal control processes of the original problem. 
The iMDP algorithm also maintains low time complexity per iteration by asynchronously computing Bellman updates in each iteration.
There are two main advantages when using the iMDP algorithm to solve stochastic optimal control problems. First, the iMDP algorithm provides a method to compute optimal control policies without the need to derive and characterize viscosity solutions of associated HJB equations. Second, the algorithm is suitable for various online robotics applications without \textit{a priori} discretization of the state space.

Risk management in stochastic optimal control has also been received extensive attention by researchers in several fields.
In robotics, a common risk management problem is \textit{chance-constrained optimization} \cite{BlackmoreOnoBektassovWilliamsIEEE-TRO10,Banerjee11ACC,MarcoChow2012}. Chance constraints specify that starting from \textit{a given initial state}, the \textit{time}-$\mathit{0}$ probability of success must be above a given threshold where success means reaching goal areas safely. Alternatively, we call these constraints \textit{risk constraints} if we concern more about failure probabilities. Despite intensive work done to solve this problem in last 20 years, designing computationally-efficient algorithms that respect chance constraints for systems with continuous-time dynamics is still an open question. The Lagrangian approach~\cite{KirkOptimalControl,Kosmol1993,Kosmol2001907} is a possible method for solving the mentioned constrained optimization. However, this approach requires numerical procedures to compute Lagrange multipliers before obtaining a policy, which is computationally demanding for high dimensional systems and unsuitable for online robotics applications. 

In another approach (see, e.g.,~\cite{Blackmore06aprobabilistic,OnoW08,Luders10_GNC,Luders13_GNC,luders2010bounds}), most previous works use discrete-time multi-stage formulations to model this problem. In these modified formulations, failure is defined as collision with convex obstacles which can be represented as a set of linear inequalities. Probabilities of safety for states at different time instants as well as for the entire path are pre-specified by users. The proposed algorithms to solve these formulations often involve two main steps. In the first step, these algorithms often use heuristic~\cite{Blackmore06aprobabilistic} or iterative~\cite{OnoW08} \textit{risk allocation} procedures to identify the tightness of different constraints. In the second step, the formulations with identified active constraints can be solved using mixed integer-linear programming with possible assistance of particle sampling~\cite{BlackmoreOnoBektassovWilliamsIEEE-TRO10} and linear programming relaxation~\cite{Banerjee11ACC}. Computing risk allocation fully is computationally intensive. Thus, in more recent works~\cite{Luders10_GNC,Luders13_GNC,luders2010bounds}, the authors make use of the Rapidly-Exploring Random Tree (RRT) and RRT$^*$ algorithms to build tree data structures that also store incremental approximate allocated risks at tree nodes. Based on the RRT$^*$ algorithm, the authors have proposed the Chance-Constrained-RRT$^*$ (CC-RRT$^*$) algorithm that would provide asymptotically-optimal and probabilistically-feasible trajectories for linear Gaussian systems subject to process noise, localization error, and uncertain environmental
constraints. In addition, the authors have also proposed a new objective function that allows users to trade-off between minimizing path duration and risk-averse behavior by adjusting the weights of these additive components in the objective function.

We note that the modified formulations in the above approach do not preserve well the intended guarantees of the original chance constraint formulation. 
In addition, 
the approach requires the direct representation of convex obstacles into the formulations. Therefore, solving the resulting mixed integer-linear programming in the presence of a large number of obstacles is computationally demanding. The proposed algorithms are also over-conservative due to loose union bounds when performing the risk allocation procedures. To counter these conservative bounds, CC-RRT$^*$ constructs more aggressive trajectories by adjusting the weights of the path duration and risk-averse components in the objective function. As a result, it is hard to automate the selection of trajectory patterns. 

Moreover, specifying in advance probabilities of safety for states at different time instants and for the entire path can lead to policies that have irrational behaviors due inconsistent risk preference over time. This phenomenon is known as \textit{time-inconsistency} of control policies. For example, when we execute a control policy returned by one of the proposed algorithms, due to noise, the system can be in an area surrounded by obstacles at some later time $t$, it would be safer if the controller takes into account this situation and increases the required probability of safety at time  $t$ to encourage careful maneuvers. Similarly, if the system enters an obstacle-free area, the controller can reduce the required probability of safety at time $t$ to encourage more aggressive maneuvers. Therefore, to maintain time-consistency of control policies, the controller should adjust safety probabilities so that they are \textit{contingent} on available information along the controlled trajectory.   

In other related works~\cite{chen2004dynamic,piunovskiy2006dynamic,mannor2011mean}, several authors have proposed new formulations in which the objective functions and constraints are evaluated using (different) single-period risk metrics. However, these formulations again lead to potential inconsistent behaviors as risk preferences change in an irrational manner between periods~\cite{huang2011price}. Recently, in~\cite{MarcoChow2012}, the authors used Markov dynamic time-consistent risk measures~\cite{ruszczynski2006optimization,ruszczynski2006conditional,rudloff2011time} to assess the risk of future cost stream in a consistent manner and established a dynamic programming equation for this modified formulation. The resulting dynamic programming equation has functionals over the state space as control variables. When the state space is continuous, the control space has infinite dimensionality, and therefore, solving the dynamic programming equation in this case is computationally challenging. 
%

In mathematical finance, closely-related problems have been studied in the context of hedging with portfolio constraints where constraints on terminal states are enforced almost surely (a.s.), yielding so-called \textit{stochastic target problems}~\cite{bouchard2011weak,bouchard2010optimal,bouchard2009stochastic,touzi2013optimal,bouchard2012weak}. Research in this field focuses on deriving HJB equations for this class of problems. Recent analytical tools such as weak dynamic programming~\cite{bouchard2011weak} and geometric dynamic programming~\cite{SoTo02b, BV10:1} have been developed to achieve this goal. These tools allow us to derive HJB equations and find viscosity solutions for a larger class of problems while avoiding measurability issues.

In this paper, we consider the above risk-constrained problems. That is, we investigate stochastic optimal control problems with risk constraints that are expressed in terms of bounded failure probabilities for \textit{particular initial states}. We present here a martingale approach to solve these problems such that obtained control policies are time-consistent with the initial threshold of failure probability. The martingale represents the level of risk tolerance that is contingent on available information over time. Thus, the martingale approach transforms a risk-constrained problem into a stochastic target problem. By sampling in the augmented state space and computing proper boundary conditions of the reformulated problem, we extend the iMDP algorithm to compute anytime solutions after a small number of iterations. When more computing time is allowed, the proposed algorithm refines the solution quality in an efficient manner.

The main contribution of this paper is twofold. First, we present a novel martingale approach that fully respects the considered risk constraints for systems with continuous-time dynamics in a time-consistent manner. The approach enable us to manage risk in several practical robotics applications without directly deriving HJB equations, which are hard to obtain in many situations. Second, we propose a computationally-efficient algorithm that guarantees probabilistically-sound and asymptotically-optimal solutions to the stochastic optimal control problem in the presence of risk constraints. That is, all constraints are satisfied in a suitable sense, and the objective function is minimized as the number of iterations approaches infinity. We demonstrate the effectiveness of the proposed algorithm on motion planning and control problems subject to bounded collision probability in uncertain cluttered environments.

This paper is organized as follows. A formal problem definition is given in Section~\ref{section:problem}. In Section~\ref{section:martingale}, we discuss the martingale approach and the key transformation. The extended iMDP algorithm is described in Section~\ref{section:algorithm}. The analysis of the proposed algorithm is presented in Section~\ref{section:analysis}. We present simulation examples and experimental results in Section~\ref{section:experiments} and conclude the paper in Section~\ref{section:conclusions}.

\section{Problem Definition} \label{section:problem}

In this section, we present a generic stochastic optimal control formulation with definitions and technical assumptions as discussed in~\cite{huynh.karaman.ea:icra12,Huynh2012.ArXiV,huynh2012probabilistically}. We also explain how to formulate risk constraints.

\paragraph*{Stochastic Dynamics}

Let $d_x$, $d_u$, and $d_w$ be positive integers. 
Let $S$ be a compact subset  of $\reals^{d_x}$, which is the closure of its interior $S^o$ and has a smooth boundary $\partial S$. Let a compact subset ${U}$ of $\reals^{d_u}$ be a control set. 
The state of the system at time $t$ is $x(t) \in S$, which is fully observable at all times.

Suppose that a stochastic process $\{w(t); t \ge 0\}$ is a $d_w$-dimensional Brownian motion
on some probability space. We define $\{\mathcal{F}_t ; t \ge 0\}$ as the augmented filtration generated by the Brownian motion $w(\cdot)$.
Let a control process $\{u(t) ; t \ge 0\}$ be a $U$-valued, measurable random process also defined on the same probability space
such that the pair $(u(\cdot),w(\cdot))$ is admissible~\cite{huynh.karaman.ea:icra12}. Let the set of all such control processes be $\mathcal{U}$. Let $\reals^{d_x \times d_w}$ denote the set of all $d_x$ by $d_w$ real matrices. We consider systems with dynamics described by the controlled diffusion process:
\begin{align}
dx(t) = f(x(t),u(t)) \, dt + F(x(t),u(t)) \, dw(t), \forall t \ge 0 \label{eqn:system}
\end{align}
where $f: S \times {U} \to \reals^{d_x}$ and $F : S \times U \to \reals^{d_x \times d_w}$ are bounded measurable and continuous functions as long as $x(t) \in S^o$. The initial state $x(0)$ is a random vector in $S$. We assume that the matrix $F(\cdot,\cdot)$ has full rank. The continuity requirement of $f$ and $F$ can be relaxed with mild assumptions~\cite{Kushner2000,huynh.karaman.ea:icra12} such that we still have a weak solution to Eq.~\eqref{eqn:system} that is unique in the weak sense~\cite{Karatzas1991}.

\paragraph*{Cost-to-go Function and Risk Constraints}
We define the {\em first exit time} $\FirstExit{u}{z}: \mathcal{U} \times S \rightarrow [0,+\infty]$ under a control process $u(\cdot) \in \mathcal{U}$ starting from $x(0)=z \in S$ as 
$$
\FirstExit{u}{z} = \inf\big\{ t : x(0)=z, \ x(t) \notin {S^o} \mbox{, and Eq.\eqref{eqn:system}} \big\}.
$$
In other words, $\FirstExit{u}{z}$ is the first time that the trajectory of the dynamical system given by Eq.~\eqref{eqn:system} starting from $x(0)=z$ hits the boundary $\partial S$ of $S$. 
The random variable $\FirstExit{u}{z}$ can take value $\infty$ if the trajectory $x(\cdot)$ never exits $S^o$.

The expected cost-to-go function under a control process $u(\cdot)$ is a mapping from $S$ to $\reals$ defined as
\begin{align}
\resizebox{0.9\hsize}{!}{$\CostToGo{z}{u} = {\EE}_0^z\left[\int_0^\FirstExit{u}{z} \alpha^t \, g\big(x(t), u(t)\big)\,dt + \alpha^{\FirstExit{u}{z}}h(x(\FirstExit{u}{z}))\right]$}, \label{eqn:costtogo}
\end{align}
where $\EE_t^z$ denotes the conditional expectation given $x(t)=z$, and $g : S \times U \to \reals$, $h : S \to \reals$ are bounded measurable and continuous functions, called the cost rate function and the terminal cost function, respectively, and $\alpha \in [0,1)$ is the discount rate.
We further assume that $g(x,u)$ is uniformly H\"{o}lder continuous in $x$ with exponent $2\rho \in (0,1]$ for all $u \in U$. 
We note that the discontinuity of $g,h$ can be treated as in~\cite{Kushner2000,huynh.karaman.ea:icra12}.

Let $\Gamma \subset \partial S$ be a set of failure states, and $\eta \in [0,1]$ be a threshold for risk tolerance given as a parameter. We consider a risk constraint that is specified for an initial state $x(0)=z$ under a control process $u(\cdot)$ as follows:
\begin{align*}
P_0^z (x(\FirstExit{u}{z}) \in \Gamma) \leq \eta,
\end{align*}
where $P_t^z$ denotes the conditional probability at time $t$ given $x(t)=z$. That is, controls that drive the system from time $0$ until the first exit time must be consistent with the choice of $\eta$ and initial state $z$ at time $0$. Intuitively, the constraint enforces that starting from \textit{a given state $z$} at time $t=0$, if we execute a control process $u(\cdot)$ for $N$ times, when $N$ is very large, there are at most $N\eta$ executions resulting in failure. Control processes $u(\cdot)$ that satisfy this constraint are called time-consistent. To have time-consistent control processes, the risk tolerance along controlled trajectories must vary consistently with the initial choice of risk tolerance $\eta$ based on available information over time.

Let $\overline{\reals}$ be the extended real number set. The {\em optimal cost-to-go function} $J^*: S \to \overline{\reals}$ is defined as follows \footnote{The semicolon in \OptimalCostToGo{z ; \eta} signifies that $\eta$ is a parameter.} \footnote{Compared to~\cite{huynh2012probabilistically}, we consider a larger set of control processes than the set of Markov control processes here. We will restrict again to Markov control processes in the reformulated problem.}:
\begin{align}
\mathcal{OPT}1: \ \ \ &\OptimalCostToGo{z ; \eta} = \inf_{u(\cdot) \in \mathcal{U}} \CostToGo{z}{u} \label{eqn:opt1}\\
\textit{ s/t } \ \ \ \ \ &P_0^z (x(\FirstExit{u}{z}) \in \Gamma) \leq \eta \ \textit{ and \ Eq.~\eqref{eqn:system}} \label{eqn:opt1a}.
\end{align} 
A control process $u^*(\cdot)$ is called optimal if $ \CostToGo{z}{u^*}=J^*(z; \eta)$. For any $\epsilon >0$, a control process $u(\cdot)$ is called an $\epsilon$-optimal policy if $|\CostToGo{z}{u}-J^*(z;\eta)| \leq \epsilon$.

We call a sampling-based algorithm probabilistically-sound if the probability that a solution returned by the algorithm is feasible approaches one as the number of samples increases. We also call a sampling-based algorithm asymptotically-optimal if the sequence of solutions returned from the algorithm converges to an optimal solution in probability as the number of samples approaches infinity. Solutions returned from algorithms with such properties are called probabilistically-sound and asymptotically-optimal.

In this paper, we consider the problem of computing the optimal cost-to-go function $J^*$ and an optimal control process $u^*$ if obtainable. Our approach, outlined in Section~\ref{section:algorithm}, approximates the optimal cost-to-go function and an optimal policy in an anytime fashion using an incremental sampling-based algorithm that is both probabilistically-sound and asymptotically-optimal.

\section{Martingale Approach} \label{section:martingale}
We now present the martingale approach that transforms the considered risk-constrained problem into an equivalent stochastic target problem. The following lemma to diffuse risk constraints is a key tool for our transformation.
\subsection{Diffusing Risk Constraints}
\begin{lemma}[see \cite{touzi2013optimal,bouchard2012weak}]
From $x(0)=z$, a control process $u(\cdot)$ is feasible for $\mathcal{OPT}1$ if and only if there exists a square-integrable (but possibly unbounded) process $c(\cdot) \in \reals^{d_w}$ and a martingale $q(\cdot)$ satisfying:
\begin{enumerate}
	\item $q(0)=\eta$, and $dq(t)=c^T(t)dw(t)$,
	\item For all $t$, $q(t) \in [0,1]$ a.s.,
	\item $1_\Gamma(x(\FirstExit{u}{z})) \leq q(\FirstExit{u}{z})$ a.s,
\end{enumerate}
where $1_\Gamma(x)=1$ if and only if $x \in \Gamma$ and $0$ otherwise. The martingale $q(t)$ stands for the level of risk tolerance at time $t$. We call $c(\cdot)$ a martingale control process.
\end{lemma}
\begin{proof}
Assuming that there exists $c(\cdot)$ and $q(\cdot)$ as above, due to the martingale property of $q(\cdot)$, we have:
\begin{align*}
P_0^z(x(\FirstExit{u}{z} ) \in \Gamma ) &= \EE \left[ 1_\Gamma(x(\FirstExit{u}{z})) | \mathcal{F}_0 \right] \\
&\leq \EE \left[ q(\FirstExit{u}{z}) | \mathcal{F}_0  \right] = q(0) =\eta.
\end{align*}
Thus, $u(\cdot)$ is feasible.

Now, let $u(\cdot)$ be a feasible control policy. Set $\eta_0 = P_0^z(x(\FirstExit{u}{z} )\in \Gamma)$. We note that $\eta_0 \leq \eta$. We define the martingale
$$
  \overline{q}(t) = \EE[1_\Gamma(x(\FirstExit{u}{z})) | \mathcal{F}_t].
$$
Since $\overline{q}(\FirstExit{u}{z}) \in [0,1]$, we infer that $\overline{q}(t) \in [0,1] $ almost surely. We now set 
$$\widehat{q}(t) = \overline{q}(t) + (\eta - \eta_0),$$
then $\widehat{q}(t)$ is a martingale with $\widehat{q}(0) = \overline{q}(0) + (\eta - \eta_0) = \eta_0 + (\eta - \eta_0) = \eta$ and $\widehat{q}(t) \geq 0$ almost surely.

Now, we define $\tau = \inf\{ t \in [0,\FirstExit{u}{z}] \ | \ \widehat{q}(t) \geq 1 \}$, which is a stopping time. 
Thus, $$q(t) = \widehat{q}(t)1_{t \leq \tau} + 1_{t > \tau},$$ as a stopped process of the martingale $\widehat{q}(t)$ at $\tau$, is a martingale with values in [0,1] a.s. 

If $\tau < \FirstExit{u}{z}$, we have 
$$1_\Gamma(x(\FirstExit{u}{z})) \leq 1 = q(\FirstExit{u}{z}),$$ 
and if $\tau = \FirstExit{u}{z}$, we have 
\begin{align*}
q(\FirstExit{u}{z}) &= \EE[1_\Gamma(x(\FirstExit{u}{z})) | \mathcal{F}_{\FirstExit{u}{z}}] + (\eta - \eta_0) \\
                      &=  1_\Gamma(x(\FirstExit{u}{z})) + (\eta - \eta_0) \geq  1_\Gamma(x(\FirstExit{u}{z})).
\end{align*} 
Hence, $q(\cdot)$ also satisfies that $1_\Gamma(x(\FirstExit{u}{z})) \leq q(\FirstExit{u}{z})$.

The control process $c(\cdot)$ exists due to the martingale representation theorem~\cite{lamberton2008introduction}, which yields $dq(t)=c^T(t)dw(t)$. We however note that $c(t)$ is possibly unbounded. We also emphasize that the risk tolerance $\eta$ becomes the initial value of the martingale $q(\cdot)$.
\end{proof}

\subsection{Stochastic Target Problem} \label{subsection:reformulation}
Using the above lemma, we augment the original system dynamics with the martingale $q(t)$ into the following form:
\begin{align}
\resizebox{.9\hsize}{!} {$d\begin{bmatrix} x(t) \\ q(t) \end{bmatrix}= \begin{bmatrix} f(x(t),u(t)) \\ 0 \end{bmatrix} dt +  \begin{bmatrix} F(x(t),u(t)) \\ c^T(t) \end{bmatrix}dw(t),$} \label{eqn:extendedsystem}
\end{align}
where $(u(\cdot),c(\cdot))$ is the control process of the above dynamics. The initial value of the new state is $(x(0),q(0))=(z,\eta)$. We will refer to the augmented state space $S \times [0,1]$ as $\overline{S}$ and the augmented control space $U \times \reals^{d_w}$ as $\overline{U}$. We also refer to the nominal dynamics and diffusion matrix of Eq.~\eqref{eqn:extendedsystem} as $\overline{f}(x,q,u,c)$ and $\overline{F}(x,q,u,c)$ respectively.

It is well-known that in the following reformulated problem, optimal control processes are Markov controls~\cite{touzi2013optimal,bouchard2012weak,Oksendal:1992}. Thus, let us now focus on the set of Markov controls that depend only on the current state, i.e., $(u(t),c(t))$ is a function only of $(x(t),q(t))$, for all $t \ge 0$. 
A function $\varphi: \overline{S} \rightarrow \overline{U}$ represents a \textit{Markov or feedback control policy}, which is known to be admissible with respect to the process noise $w(\cdot)$. 
Let $\Psi$ be the set of all such policies $\varphi$. Let $\mu: \overline{S} \rightarrow U$ and $\kappa: \overline{S} \rightarrow \reals^{d_w}$ so that $\varphi=(\mu,\kappa)$.
We rename $\FirstExit{u}{z}$ to \FirstExit{\varphi}{z} for the sake of notation clarity.
Using these notations, $\mu(\cdot,1)$ is thus a Markov control policy for the unconstrained problem, i.e. the problem without the risk constraint, that maps from $S$ to $U$. Henceforth, we will \textit{use $\mu(\cdot)$ to refer to $\mu(\cdot,1)$} when it is clear from the context. Let $\Pi$  be the set of all such Markov control policies $\mu(\cdot)$ on $S$. 

Now, let us rewrite cost-to-go function $\CostToGo{z}{u}$ in Eq.~\eqref{eqn:costtogo} for the threshold $\eta$ at time $0$ in a new form:
\begin{align}
J_{\varphi}(z,\eta) &= {\EE} \Bigg [ \int_0^\FirstExit{\varphi}{z} \alpha^t \, g\big(x(t), \mu(x(t),q(t))\big)\,dt \nonumber \\
           &\ \ \ \ \ \ \ \ \ +  \alpha^{\FirstExit{\varphi}{z}}h(x(\FirstExit{\varphi}{z})) \Big | (x,q)(0)=(z,\eta) \Bigg ]. \label{eqn:reformulatedcosttogo}
\end{align}
We therefore transform the risk-constrained problem $\mathcal{OPT}1$ into a stochastic target problem as follows\footnote{The comma in $\OptimalCostToGo{z,\eta}$ signifies that $\eta$ is a state component rather than a parameter, and $\OptimalCostToGo{z,\eta}$ is equal to $\OptimalCostToGo{z;\eta}$ in the previous formulation.}:
\begin{align}
\mathcal{OPT}2: \ \ \ &\OptimalCostToGo{z,\eta} = \inf_{\varphi \in \Psi} J_{\varphi}(z,\eta) \label{eqn:stpcost} \\ 
\text{s/t} \quad \quad &1_\Gamma(x(\FirstExit{\varphi}{z})) \leq q(\FirstExit{\varphi}{z}) \ \ \textit{a.s.} \ \textit{ and \ Eq.~\eqref{eqn:extendedsystem}}. \label{eqn:asconstraint}
\end{align}
The constraint in the above formulation specifies the relationship of random variables at the terminal time as target, and hence the name of this formulation~\cite{touzi2013optimal,bouchard2012weak}. In this formulation, we solve for feedback control policies $\varphi$ for all $(z,\eta) \in \overline{S}$ instead of a particular choice of $\eta$ for $x(0)=z$ at time $t=0$. We note that in this formulation, boundary conditions are not fully specified \textit{a priori}. In the following subsection, we discuss how to remove the constraint in Eq.~\eqref{eqn:asconstraint} by constructing its boundary and computing the boundary conditions.

\subsection{Characterization and Boundary Conditions} \label{subsection:characterization}
The domain of the stochastic target problem in $\mathcal{OPT}2$ is: $$D= \{ (z,\eta) \in \overline{S} \ | \ \exists \varphi \in \Psi \ \textit{s/t} \ 1_\Gamma(x(\FirstExit{\varphi}{z})) \leq q(\FirstExit{\varphi}{z}) \ \textit{a.s.}\}.$$ 
By the definition of the risk-constrained problem $\mathcal{OPT}1$, we can see that if $(z,\eta) \in D$ then $(z,\eta') \in D$ for any $\eta < \eta' \leq 1$. Thus, for each $z\in S$, we define 
\begin{align}
\gamma(z) = \inf \ \{ \eta \in [0,1] \ | \ (z,\eta) \in D \},
\end{align} 
as the infimum of risk tolerance at $z$. Therefore, we also have: 
\begin{align}
\gamma(z)= \inf_{u \in \mathcal{U}} P_0^z \big ( x(\FirstExit{u}{z}) \in \Gamma) = \inf_{u \in \mathcal{U}} \EE_0^z \Big [ 1_\Gamma(x(\FirstExit{u}{z})) \Big]. \label{eqn:minCollisionProb}
\end{align}
Thus, the boundary of $D$ is 
\begin{align}
\partial D = &S \times \{1\} \cup \{(z,\gamma(z)) \ | \ z \in S \}  \nonumber \\
             &\cup \{(z,\eta) \ | \ z \in \partial S, \eta \in [\gamma(z),1] \}.
\end{align} 
For states in $\{(z,\eta) \ | \ z \in \partial S, \eta \in [\gamma(z),1] \}$, the system stops on $\partial S$ and takes terminal values according to $h(\cdot)$.

The domain $D$ is illustrated in Fig.~\ref{figDomain}. In this example, the state space $S$ is a bounded two-dimensional area with boundary $\partial S$ containing a goal region $G$ and an obstacle region $\Gamma=Obs$. The augmented state space $\overline{S}$ augments $S$ with an extra dimension for the martingale state $q$. The infimum probability of reaching into $\Gamma$ from states in $S$ is depicted as $\gamma$. As we can see, $\gamma$ takes value $1$ in $\Gamma$. The volume between $\gamma$ and the hyper-plane $q=1$ is the domain $D$ of $\mathcal{OPT}2$.

\begin{figure*}[t]
\begin{center}
  \subfigure[A domain of $\mathcal{OPT}1$.]{
  \label{figDomain}
  \includegraphics[width=85mm]{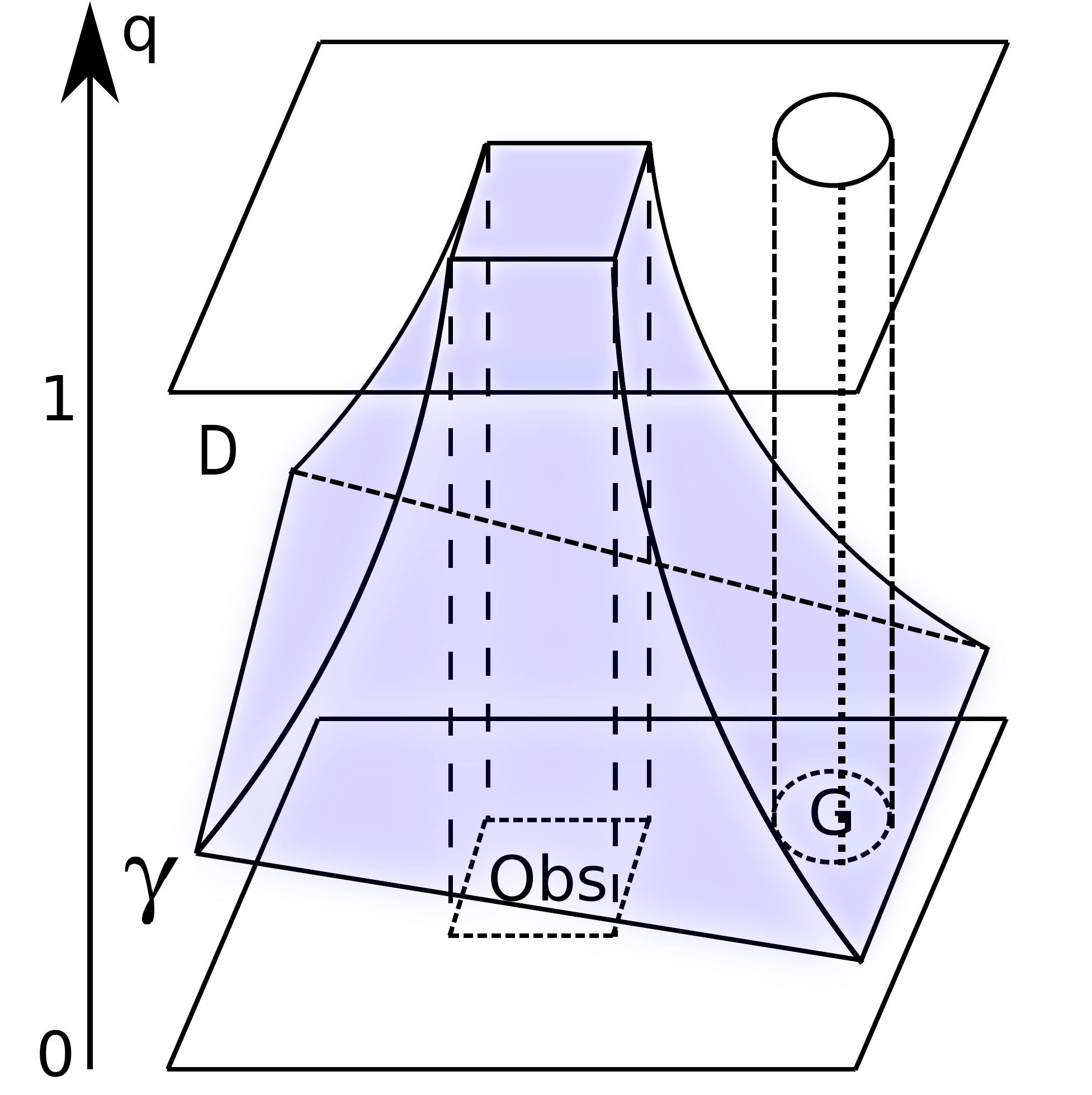}
  }
  \hfill
  \subfigure[Failure probabilities due to optimal policies of the unconstrained problem.]{
  \label{figExtraDomain}
  \includegraphics[width=85mm]{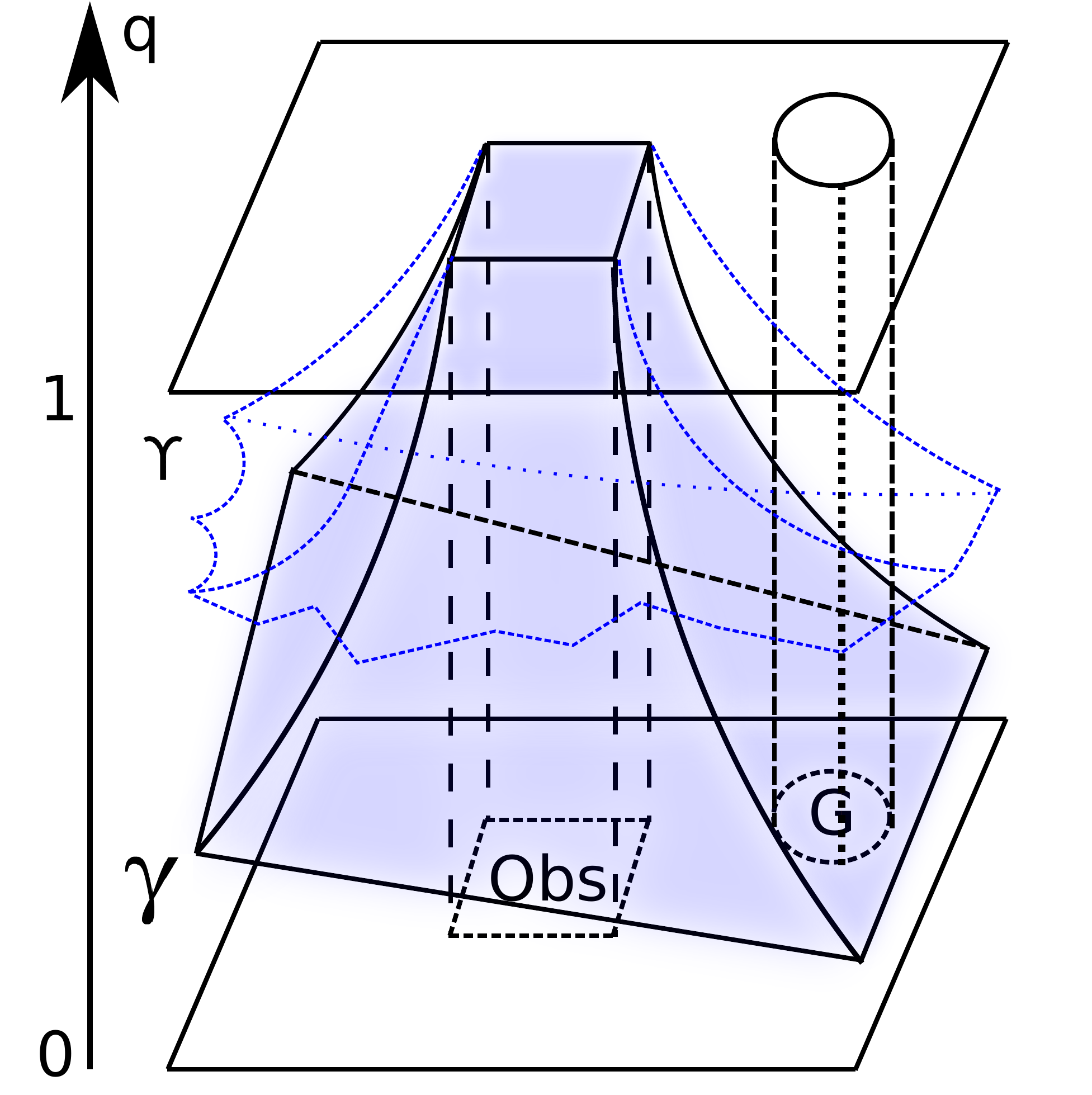}  
  }
  \caption{In Fig.~\ref{figDomain}, we show an example of the domain of $\mathcal{OPT}2$. The state space $S$ is a bounded two-dimensional area with boundary $\partial S$ containing a goal region $G$ and an obstacle region $\Gamma=Obs$. The augmented state space $\overline{S}$ augments $S$  with an extra dimension for the martingale state $q$.  The infimum probability of reaching into $\Gamma$ from states in $S$ is depicted as $\gamma$, which takes value $1$ in $\Gamma$. The volume between $\gamma$ and the hyper-plane $q=1$ is the domain $D$ of $\mathcal{OPT}2$. In Fig.~\ref{figExtraDomain}, we show an illustration of the failure probability function $\Upsilon$ due to an optimal control policy $\mu^*(\cdot,1)$ of the unconstrained problem. We plot $\Upsilon$ for the same two-dimensional example. By the definitions of $\gamma$ and $\Upsilon$, we have $\Upsilon \geq \gamma$.}
\end{center}
\end{figure*}

Now, let $\eta=1$, we notice that $\OptimalCostToGo{z,1}$ is the optimal cost-to-go from $z$ for the stochastic optimal problem without the risk constraint:
$$
\OptimalCostToGo{z,1} = \inf_{u \in \mathcal{U}} \CostToGo{z}{u}. 
$$ 
An optimal control process that solves this optimization problem is given by a Markov policy $\mu^*(\cdot,1) \in \Pi$. We now define the failure probability function $\Upsilon: S \rightarrow [0,1]$ under such an optimal policy $\mu^*(\cdot,1)$ as follows:
\begin{align}
      \Upsilon(z) = \EE \big[ 1_\Gamma(x(T_{\mu^*}^z)) \big ], \ \forall z \in S, \label{eqn:Upsilon}
\end{align}
where $T_{\mu^*}^z$ is the first exit time when the system follows the control policy $\mu^*(\cdot,1)$ from the initial state $z$. By the definitions of $\gamma$ and $\Upsilon$, we can recognize that $\Upsilon(z) \geq \gamma(z)$ for all $z \in S$. Figure~\ref{figExtraDomain} shows an illustration of $\Upsilon$ for the same example in Fig.~\ref{figDomain}.

Since following the policy $\mu^*(\cdot,1)$ from an initial state $z$ yields a failure probability $\Upsilon(z)$, we infer that:
\begin{align}
J^*(z,1)=J^*(z,\Upsilon(z)). \label{eqn:followingMustar}
\end{align}
From the definition of the problem $\mathcal{OPT}1$, we also have: 
\begin{align}
0 \leq \eta < \eta' \leq 1 \Rightarrow J^*(z,\eta) \geq J^*(z,\eta'). \label{eqn:compare}
\end{align}
Thus, for any $\Upsilon(z) < \eta < 1$, we have:
\begin{align}
J^*(z,1) \leq J^*(z,\eta) \leq J^*(z,\Upsilon(z)).  \label{eqn:followingMustar1}
\end{align}
Combining Eq.~\eqref{eqn:followingMustar} and Eq.~\eqref{eqn:followingMustar1}, we have:
\begin{align}
\forall \ \eta \in [\Upsilon(z), 1] \Rightarrow J^*(z,\eta)=J^*(z,1). \label{eqn:followingMustar2}
\end{align}
As a consequence, when we start from an initial state $z$ with a risk threshold $\eta$ that is at least $\Upsilon(z)$, it is optimal to execute an optimal control policy of the corresponding unconstrained problem from the initial state $z$.

It also follows from Eq.~\eqref{eqn:compare} that reducing the risk tolerance from $1.0$ along the controlled process can not reduce the optimal cost-to-go function evaluated at $(x(t),q(t)=1.0)$. Thus, we infer that for augmented states $(x(t),q(t))$ where $q(t)=1.0$, the optimal martingale control $c^*(t)$ is 0.

Now, under all admissible policies $\varphi$, we can not obtain a failure probability for an initial state $z$ that are lower than $\gamma(z)$. Thus, it is clear that $J^*(z,\eta)=+ \infty$ for all $0 \leq \eta < \gamma(z)$. The following lemma characterizes the optimal martingale control $c^*(t)$ for augmented states $(x(t),q(t)=\gamma(x(t)))$.

\begin{lemma}
Given the problem definition as in Eqs.~\eqref{eqn:opt1}-\eqref{eqn:opt1a}, we assume that $\gamma(x)$ is a smooth function\footnote{When $\gamma(x)$ is not smooth, we need the concept of viscosity solutions and weak dynamic programming principle. See \cite{touzi2013optimal,bouchard2012weak} for details.}. When $q(t)=\gamma(x(t))$ and $u(t)$ is chosen, we must have:
\begin{align}
	c(t)^T =\frac{\partial \gamma}{\partial x(t)}^T F(x(t),u(t)). \label{eqn:optimalc}
\end{align}
\end{lemma}
\begin{proof}
Using the geometric dynamic programming principle~\cite{SoTo02b, BV10:1}, we have the following result: for all stopping time $\tau \ge t$, when $q(t) = \gamma(x(t))$, a feasible control policy $\varphi \in \Psi$ satisfies $q(\tau) \geq \gamma(x(\tau))$ almost surely.

Take $\tau = t+$, under a feasible control policy $\varphi$, we have $q(t+) \geq \gamma(x(t+))$ a.s. for all $t$, and hence $dq(t) \geq d\gamma(x(t))$ a.s. By It$\hat{o}$ lemma, we derive the following relationship:
\begin{align*}
c^T(t)dw(t) &\geq \frac{\partial \gamma}{\partial x}^T \Big( f(x(t),u(t))dt + F(x(t),u(t))dw(t) \Big) \\
            + \frac{1}{2} &Tr\Big(F(x(t),u(t)) F(x(t),u(t))^T \frac{\partial^2 \gamma}{(\partial x)^2} \Big) dt \ a.s.
\end{align*}
For the above inequality to hold almost surely, the coefficient of $dw(t)$ must be $0$. This leads to Eq.~\eqref{eqn:optimalc}.

\end{proof}
In addition, if a control process that solves Eq.~\eqref{eqn:minCollisionProb} is obtainable, say $u_{\gamma}$, the cost-to-go due to that control process is $\CostToGo{z}{u_{\gamma}}$. We will conveniently refer to $\CostToGo{z}{u_{\gamma}}$ as $J^{\gamma}(z)$. Under the mild assumption that $u_{\gamma}$ is unique, it follows that $J^{\gamma}(z)=\OptimalCostToGo{z,\gamma(z)}$.

We also emphasize that when $(x(t),q(t))$ is inside the interior $D^{o}$ of $D$, the usual dynamic programming principle holds. The extension of iMDP outlined below is designed to compute the sequence of approximate cost-to-go values on the boundary $\partial D$ and in the interior $D^o$.

\section{Algorithm} \label{section:algorithm}
In this section, we briefly overview how the Markov chain approximation technique is used in both the original and augmented state spaces. We then present the extended iMDP algorithm that incrementally constructs the boundary values and computes solutions to our problem. In particular, we sample in the original state space $S$ to compute $\OptimalCostToGo{\cdot,1}$ and its induced collision probability $\Upsilon(\cdot)$ as in Eq.~\eqref{eqn:Upsilon}, the min-failure probability $\gamma(\cdot)$ as in Eq.~\eqref{eqn:minCollisionProb} and its induced cost-to-go $J^{\gamma}(\cdot)$. Concurrently, we also sample in the augmented state space $\overline{S}$ with appropriate values for samples on the boundary of $D$ and approximate the optimal cost-to-go function $J^*(\cdot,\cdot)$ in the interior $D^o$. As a result, we construct a sequence of anytime control policies to approximate an optimal control policy $\varphi^*=(\mu^*,\kappa^*)$ in an efficient iterative procedure.

\subsection{Markov Chain Approximation} \label{subsection:mc}
A discrete-state Markov decision process (MDP) is a tuple ${\cal M} = (X,A,P,G,H)$ where $X$ is a finite set of states, $A$ is a set of actions that is possibly a continuous space, $P(\cdot \given \cdot,\cdot): X \times X \times A \to \reals_{\ge 0}$ is the transition probability function, $G(\cdot,\cdot) : X \times A \to \reals$ is an immediate cost function, and $H: X \to \reals$ is a terminal cost function. From an initial state $\xi_0$, under a sequence of controls $\{v_i ; i \in \naturals \}$, the induced trajectory $\{\xi_i ; i \in \naturals \}$ is generated by following the transition probability function $P$.

On the state space $S$, we want to approximate $\OptimalCostToGo{z,1}$, $\Upsilon(z)$, $\gamma(z)$ and $J^{\gamma}(z)$ for any state $z \in S$, and  it is suffice to consider optimal Markov controls as shown in~\cite{huynh.karaman.ea:icra12,Huynh2012.ArXiV}. The Markov chain approximation method approximates the continuous dynamics in Eq.~\eqref{eqn:system} using a sequence of MDPs $\{ {\cal M}_n = (S_n,U,P_n,G_n,H_n)\}_{n=0}^{\infty}$ and a sequence of holding times $\{ \Delta t_n\}_{n=0}^{\infty}$ that are locally consistent. In particular, we construct $G_n(z,v)=g(z,v)\Delta t_n(z)$, $H_n(z)=h(z)$ for each $z \in S_n$ and $v \in U$. We also require that $\lim_{n \rightarrow \infty} \sup_{i \in \naturals, \omega \in \Omega_n}||\Delta \xi_{i}^n||_2=0$ where $\Omega_n$ is the sample space of ${\cal M}_n$, $\Delta \xi_{i}^n = \xi^n_{i+1}-\xi^n_{i}$, and
\begin{itemize}
\item For all $z \in S$, $\lim_{n \to \infty} \Delta t_n(z) = 0$,
\item For all $z \in S$ and all $v \in {U}$:
\end{itemize}
\begin{eqnarray*}
\lim_{n \to \infty} \frac{\EE_{P_n}[\Delta \xi_{i}^n \given \xi_i^n = z, u_i^n = v]}{\Delta t_n (z)} &=& f(z,v),\\
\lim_{n \to \infty} \frac{\Cov_{P_n}[\Delta \xi_{i}^n \given \xi_i^n = z, u_i^n = v]}{\Delta t_n (z)}&=& F(z,v)F(z,v)^T.
\end{eqnarray*}

The main idea of the Markov chain approximation approach for solving the original continuous problem is to solve a sequence of control problems defined on $\{\mathcal M_n\}_{n=0}^{\infty}$ as follows.
A Markov or feedback policy $\mu_n$ is a function that maps each state $z \in S_n $ to a control $\mu_n(z) \in U$. The set of all such policies is $\Pi_n$. 
We define $t^n_i= \sum_{0}^{i-1}\Delta t_n(\xi^n_i)$ for $i \ge 1$ and $t^n_0=0$. 
Given a policy $\mu_n$ that approximates a Markov control process $u(\cdot)$ in Eq.~\eqref{eqn:costtogo}, the corresponding cost-to-go due to $\mu_n$ on $\mathcal M_n$ is:
$$
J_{n,\mu_n} (z) = \EE^z_{P_n}\left[ \sum_{i = 0}^{I_n-1} \alpha^{t_i^n} G_n(\xi^n_i, \mu_n(\xi^n_i)) + \alpha^{t^n_{I_n}}H_n(\xi^n_{I_n})\right],
$$
where $\EE^z_{P_n}$ denotes the conditional expectation given $\xi^n_0 = z$ under $P_n$, and $\{\xi^n_i ; i \in \naturals\}$ is the sequence of states of the controlled Markov chain under the policy $\mu_n$, and $I_n$ is termination time defined as $I_n = \min\{i : \xi^n_i \in \partial S_n\}$ where $\partial S_n = \partial S \cap S_n$.

The {\em optimal cost-to-go function} $J_n^*:S \rightarrow \overline{\reals}$ that approximates $\OptimalCostToGo{z,1}$ is denoted as 
\begin{align}
J_n^*(z,1) = \inf_{\mu_n \in \Pi_n} J_{n,\mu_n}(z) \ \forall z \in S_n. \label{eqn:approximateCost}
\end{align}
An {\em optimal policy}, denoted by $\mu_n^*$, satisfies 
$
J_{n,\mu_n^*}(z) =  J_n^*(z)
$
for all $z \in S_n$. For any $\epsilon >0$, $\mu_n$ is an $\epsilon$-optimal policy if $||J_{n,\mu_n}-J^*_n||_{\infty} \leq \epsilon$.

We also define the failure probability function $\Upsilon_n: S_n \rightarrow [0,1]$ due to an optimal policy $\mu^*_n$ as follows:
\begin{align}
\Upsilon_n(z) = \EE_{P_n} \left[ 1_\Gamma(\xi^n_{I_n}) \ \big | \ x(0)= z \ ; \ \mu^*_n \ \right] \ \forall z \in S_n, \label{eqn:approximateUpsilon}
\end{align}
where we denote $\mu_n^*$ after the semicolon (as a parameter) to emphasize the dependence of the Markov chain on this control policy.

In addition, the \textit{min-failure probability} $\gamma_n$ on  $\mathcal M_n$ that approximates $\gamma(z)$ is defined as:
\begin{align}
\gamma_n(z) = \inf_{\mu_n \in \Pi_n} \EE^z_{P_n} \left[ 1_\Gamma(\xi^n_{I_n}) \right] \ \forall z \in S_n. \label{eqn:approximateProb}
\end{align}
We note that the optimization programs in Eq.~\eqref{eqn:approximateCost} and Eq.~\eqref{eqn:approximateProb} may have two different optimal feedback control policies. Let $\nu_n \in \Pi_n$ be a control policy on $\mathcal{M}_n$ that achieves $\gamma_n$, then the cost-to-go function due to $\nu_n$ is $J_{n,\nu_n}$ which approximates $J^{\gamma}$. For this reason, we conveniently refer to $J_{n,\nu_n}$ as $J_n^{\gamma}$.

Similarly, in the augmented state space $\overline{S}$, we use a sequence of MDPs $\{ \overline{{\cal M}}_n = (\overline{S}_n,\overline{U},\overline{P}_n,\overline{G}_n,\overline{H}_n)\}_{n=0}^{\infty}$ and a sequence of holding times $\{ \overline{\Delta t}_n\}_{n=0}^{\infty}$ that are locally consistent with the augmented dynamics in Eq.~\eqref{eqn:extendedsystem}. In particular, $\overline{S}_n$ is a random subset of $D \subset \overline{S}$, $\overline{G}_n$ is identical to $G_n$, and $\overline{H}_n(z,\eta)$ is equal to $H_n(z)$ if $\eta \in [\gamma_n(z),1]$ and $+\infty$ otherwise. Similar to the construction of $P_n$ and $\Delta t_n$, we also construct the transition probabilities $\overline{P}_n$ on $\overline{{\cal M}}_n$ and holding time $\overline{\Delta t}_n$ that satisfy the local consistency conditions for nominal dynamics $\overline{f}(x,q,u,c)$ and diffusion matrix $\overline{F}(x,q,u,c)$. 

A trajectory on $\overline{\mathcal{M}}_n$ is denoted as $\{\overline{\xi}^n_i; i \in \naturals \}$ where $\overline{\xi}^n_i \in \overline{S}_n$. A Markov policy $\varphi_n$ is a function that maps each state $(z,\eta) \in \overline{S}_n $ to a control $(\mu_n(z,\eta),\kappa_n(z,\eta)) \in \overline{U}$. Moreover, admissible $\kappa_n$ at $(z,1) \in \overline{S}_n$ is $0$ and at $(z,\gamma_n(z)) \in \overline{S}_n$ is a function of $\mu(z,\gamma_n(z))$ as shown in Eq.~\eqref{eqn:optimalc}. Admissible $\kappa_n$ for other states in $\overline{S}_n$ is such that the martingale-component process of $\{\overline{\xi}^n_i ; i \in \naturals \}$ belongs to [0,1] almost surely. We can show that equivalently, each control component of $\kappa_n(z,\eta)$ belongs to $[-\frac{\min(\eta,1-\eta)}{\overline{\Delta t}_n d_w},\frac{\min(\eta,1-\eta)}{\overline{\Delta t}_n  d_w}]$. The set of all such policies $\varphi_n$ is $\Psi_n$. 

Under a control policy $\varphi_n$, the cost-to-go on $\overline{\mathcal{M}}_n$ that approximates Eq.~\eqref{eqn:reformulatedcosttogo} is defined as:
\begin{align*}
\resizebox{1.0\hsize}{!}{$J_{n,\varphi_n} (z,\eta) = \EE^{z,\eta}_{\overline{P}_n}\left[ \sum_{i = 0}^{\overline{I}_n-1} \alpha^{\overline{t}_i^n} \overline{G}_n(\overline{\xi}^n_i, \mu_n(\overline{\xi}^n_i)) + \alpha^{\overline{t}^n_{\overline{I}_n}}\overline{H}_n(\overline{\xi}^n_{\overline{I}_n})\right],$}
\end{align*}
where $\overline{t}^n_i= \sum_{0}^{i-1}\overline{\Delta t}_n(\overline{\xi}_i^n)$ for $i \geq 1$ with $\overline{t}^n_0=0$, and $\overline{I}_n$ is index when the $x$-component of $\overline{\xi}_i^n$ first arrives at $\partial S$. The approximating optimal cost $J^*_n:\overline{S}_n \rightarrow \overline{\reals}$ for $J^*$ in Eq.~\eqref{eqn:stpcost} is:
\begin{align}
J_n^*(z,\eta) = \inf_{\varphi_n \in \Psi_n} J_{n,\varphi_n}(z,\eta) \ \forall (z,\eta) \in \overline{S}_n. \label{eqn:approximateCostTotal}
\end{align}
To solve the above optimization, we compute approximate boundary values for states on the boundary of $D$ using the sequence of MDP $\{\mathcal M_n\}_{n=0}^{\infty}$ on $S$ as discussed above. For states $(z,\eta) \in \overline{S}_n \cap D^o$, the normal dynamic programming principle holds.  

The extension of iMDP outlined below is designed to compute the sequence of optimal cost-to-go functions $\{ J^*_n \}_{n=0}^{\infty}$, associated failure probability functions $\{ \Upsilon_n \}_{n=0}^{\infty}$, min-failure probability functions $\{ \gamma_n \}_{n=0}^{\infty}$, min-failure cost functions $\{ J^{\gamma}_n \}_{n=0}^{\infty}$, and the sequence of anytime control policies $\{ \mu_n \}_{n=0}^{\infty}$ and $\{ \kappa_n \}_{n=0}^{\infty}$ in an incremental procedure.

\subsection{Extension of iMDP} \label{subsection:algo_description}
Before presenting the details of the algorithm, we discuss a number of primitive procedures. More details about these procedures can be found in~\cite{huynh.karaman.ea:icra12,Huynh2012.ArXiV}.
\setcounter{paragraph}{0}
\subsubsection{Sampling}
The ${\tt Sample(X)}$ procedure sample states independently and uniformly in $X$.

\subsubsection{Nearest Neighbors}
Given $\zeta \in {X} \subset \reals^{d_X}$ and a set $\DummyS \subseteq {X}$, for any $k \in \naturals$, the procedure ${\tt Nearest}(\zeta,\DummyS,k)$ returns the $k$ nearest states $\zeta' \in \DummyS$ that are closest to $\zeta$ in terms of the $d_X$-dimensional Euclidean norm. 

\subsubsection{Time Intervals}
Given a state $\zeta \in {X}$ and a number $k \in \naturals$, the procedure ${\tt ComputeHoldingTime}(\zeta,k,d)$ returns a holding time computed as follows:
$
{\tt ComputeHoldingTime}(\zeta,k,d) = \chi_t \left(\frac{\log k}{k}\right)^{\dummyasyn \dummyrate \rho/d},
$
where $\chi_t >0 $ is a constant, and $\dummyrate, \dummyasyn$ are constants in $(0,1)$ and $(0,1]$ respectively. The parameter $\rho \in (0,0.5]$ defines the H\"{o}lder continuity of the cost rate function $g(\cdot,\cdot)$ as in Section~\ref{section:problem}.

\subsubsection{Transition Probabilities} \label{para:tranprob}
We are given a state $\zeta \in X$, a subset $Y \in X$, a control $v$ in some control set $V$, a positive number $\tau$ describing a holding time, $k$ is a nominal dynamics, $K$ is a diffusion matrix. The procedure ${\tt ComputeTranProb}(\zeta,v,\tau,Y,k,K)$ returns (i) a finite set $Z_\mathrm{near} \subset X$ of states
such that the state $\zeta + k(\zeta,v)\tau$ belongs to the convex hull of $Z_\mathrm{near}$ and $||z'-z||_2= O( \tau )$ for all $\zeta' \neq \zeta \in Z_\mathrm{near}$, 
and (ii) a function $P$ that maps $Z_\mathrm{near}$ to a non-negative real numbers such that $P(\cdot)$ is a probability distribution over the support $Z_\mathrm{near}$. It is crucial to ensure that these transition probabilities result in a sequence of locally consistent chains that approximate $k$ and $K$ as presented in \cite{Kushner2000,huynh.karaman.ea:icra12,Huynh2012.ArXiV}. 

\subsubsection{Backward Extension}
Given $T >0$ and two states $z,z' \in {S}$, the procedure ${\tt ExtBackwardsS}(z,z',T)$ returns a triple $(x,v,\tau)$ such that (i) $\dot{x}(t) = f(x(t), u(t)) dt$ and $u(t)=v \in U$ for all $t \in [0,\tau]$, (ii) $\tau \leq T$, (iii) $x(t) \in {S}$ for all $t \in [0,\tau]$, (iv) $x(\tau) = z$, and (v) $x(0)$ is close to $z'$. 
If no such trajectory exists, the procedure returns failure. 
We can solve for the triple $(x,v,\tau)$ by sampling several controls $v$ and choose the control resulting in $x(0)$ that is closest to $z'$. 

When $(z,\eta),(z',\eta')$ are in $\overline{S}$, the procedure ${\tt ExtBackwardsSM}((z,\eta),(z',\eta'),T)$ returns $(x,q,v,\tau)$ in which $(x,v,\tau)$ is output of ${\tt ExtBackwardsS}(z,z',T)$ and $q$ is sampled according to a Gaussian distribution $N(\eta',\sigma_q)$ where $\sigma_q$ is a parameter.

\subsubsection{Sampling and Discovering Controls} 
For $z \in S$ and $Y \subseteq S$, the procedure ${\tt ConstructControlsS}(k,z,Y,T)$ returns a set of $k$ controls in $U$. We can uniformly sample $k$ controls in $U$. Alternatively, for each state $z' \in$ ${\tt Nearest}(z,Y,k)$, we solve for a control $v \in U$ such that (i) $\dot{x}(t) = f(x(t), u(t)) dt$ and $u(t)=v \in U$ for all $t \in [0,T]$, (ii) $x(t) \in {S}$ for all $t \in [0,T]$, (iii) $x(0)=z$ and $x(T) = z'$. 

For $(z,\eta) \in \overline{S}$ and $Y \subseteq \overline{S}$, the procedure ${\tt ConstructControlsSM}(k,(z,\eta),Y,T)$ returns a set of $k$ controls in $\overline{U}$ such that the $U$-component of these controls are computed as in ${\tt ConstructControlsS}$, and the martingale-control-components of these controls are sampled in admissible sets.

\begin{algorithm}[t]
\begin{footnotesize}
$(S_0,\overline{S}_0,J_0,\gamma_0,\Upsilon_0,J_0^\gamma,\mu_0,\kappa_0,\Delta t_0,\overline{\Delta t}_0) \leftarrow \boldsymbol{\emptyset}$\; 
\For{$n= 1 \to N$}
{
	${\tt UpdateDataStorage}(n-1,n)$ \;  \label{line:main:updatestorage_2}
	${\tt SampleOnBoundary}(n)$ \;  \label{line:main:sampleboundary_2}
	\vspace{0.03in}
	\tcp{$K_{1,n} \ge 1$ rounds for boundary conditions}
	\For{$i= 1 \to K_{1,n}$}  
	{
		${\tt ConstructBoundary}(S_n,\overline{S}_n,J_n,\gamma_n,\Upsilon_n,J_n^\gamma,\mu_n,\Delta t_n)$ \; \label{line:main:constructboundary_2}
	}
	\vspace{0.03in}
	\tcp{$K_{2,n} \ge 0$ rounds for the interior region}
	\For{$i= 1 \to K_{2,n}$}
	{  
	${\tt ProcessInterior}(S_n,\overline{S}_n,J_n,\gamma_n,\Upsilon_n,J_n^\gamma,\mu_n,\kappa_n,\overline{\Delta t}_n)$\; \label{line:main:processinterior_2}
	}
}	
\caption{\footnotesize Risk Constrained iMDP()} \label{algorithm:main_2}
\end{footnotesize}
\end{algorithm}

\begin{algorithm}[t]
\begin{footnotesize}
	$z_s \leftarrow {\tt Sample}(S)$ \label{line:constructboundary:add_interior_state_start_2}  \; 
	$z_{near} \leftarrow {\tt Nearest}(z_s, S_{n},1)$ \label{line:constructboundary:compute_nearest_2} \; 
	\If{$(x_e, u_e, \tau) \leftarrow {\tt ExtBackwardsS}(z_{near}, z_s,T_0)$}
	{ \label{line:constructboundary:steer_back_2}
		$z_e \leftarrow x_{e}(0)$\;
		$ic= \tau g(z_e,u_e) + \alpha^{\tau}J_{n}(z_{near},1)$\;
		$ic^{\gamma}= \tau g(z_e,u_e) + \alpha^{\tau}J_{n}^\gamma(z_{near})$\;		
		$ (S_{n}, \overline{S}_{n})\leftarrow (S_{n} \cup \{ z_e\}, \overline{S}_{n} \cup \{ (z_e,1) \})$ \; \label{line:constructboundary:add_interior_state_2_a}
		$ (J_{n}(z_{e},1), \gamma_{n}(z_e),\Upsilon_{n}(z_e),J_n^\gamma(z_e),\mu_{n}(z_e,1),\Delta t_{n} (z_e))\leftarrow (ic, \gamma_n(z_{near}),\Upsilon_n(z_{near}), ic^\gamma, u_{e}, \tau)$ \; \label{line:constructboundary:add_interior_state_2}
		\vspace{0.05in}
		\tcp{Perform $L_n \ge 1$ updates}
		\For{$i= 1 \to L_n$}  
		{\label{line:constructboundary:performupdate_start_2}
			\vspace{0.03in}
			\tcp{Choose $\mathcal{K}_n=\Theta \big(|S_n|^{\dummyasyn} \big ) < |S_n|$ states}
			$Z_{update} \leftarrow {\tt Nearest}(z_{e}, S_{n}\backslash \partial S_n, \mathcal{K}_n) \cup \{z_e\}$\;
			\label{line:constructboundary:compute_update_2}
			\For{$z \in Z_{update}$}
			{
				${\tt UpdateS} (z,S_{n},J_{n},\gamma_n,\Upsilon_n,J_n^\gamma,\mu_{n},\Delta t_{n})$  \;\label{line:constructboundary:performupdate_end_2}
			}
		}
	}
\end{footnotesize}
\caption{\footnotesize ConstructBoundary($S_n,\overline{S}_n,J_n,\gamma_n,\Upsilon_n,J_n^\gamma,\mu_n,\Delta t_n$)} \label{algorithm:constructboundary_2}
\end{algorithm}

\begin{algorithm}[t]
\begin{footnotesize}
	$\overline{z}_s=(z_{s},q_s) \leftarrow {\tt Sample}(\overline{S})$\; \label{line:processinterior:sample_2}
	$\overline{z}_{near}=(z_{near},q_{near}) \leftarrow {\tt Nearest}(\overline{z}_{s}, \overline{S}_{n},1)$\;
	\If{$(x_{e}, q_e, u_{e}, \tau) \leftarrow {\tt ExtBackwardsSM}(\overline{z}_{near}, \overline{z}_{s},T_0)$}
	{  \label{line:processinterior:extend_2}
		$\overline{z}_{e} \leftarrow (x_{e}(0),q_e) $\;
		\eIf{$q_e < \gamma_n(z_{near})$}
		{ 	\label{line:main:checkminprob_1}		
			\vspace{0.05in}
			\tcp{$\mathcal{C}$ takes a large value}			
			$ (\overline{S}_{n}, J_{n}(\overline{z}_e), \mu_{n}(\overline{z}_e), \kappa_{n}(\overline{z}_e), \overline{\Delta t}_{n} (\overline{z}_e))  \leftarrow ( \overline{S}_{n} \cup \{ \overline{z}_e \}, \mathcal{C}, u_{e},0, \tau)$ \;  \label{line:main:checkminprob_2}
		}
		{
			$ic= \tau g(z_{e},u_{e}) + \alpha^{\tau}J_{n}(\overline{z}_{near})$\; \label{line:updateextended_1}
			$ (\overline{S}_{n}, J_{n}(\overline{z}_e), \mu_{n}(\overline{z}_e), \kappa_{n}(\overline{z}_e), \overline{\Delta t}_{n} (\overline{z}_e))  \leftarrow ( \overline{S}_{n} \cup \{ \overline{z}_e \}, ic, u_{e},0, \tau)$ \; \label{line:updateextended_2}
			\vspace{0.05in}
			\tcp{Perform $\overline{L}_n \ge 1$ updates}
			\For{$i= 1 \to \overline{L}_n$} 
			{ \label{line:processinterior:updatecost_start_2}
				\vspace{0.03in}
				\tcp{Choose $\overline{\mathcal{K}}_n=\Theta \big(|\overline{S}_n|^{\dummyasyn} \big ) < |\overline{S}_n|$ states}
				$\overline{Z}_\mathrm{update} \leftarrow {\tt Nearest}(\overline{z}_e, \overline{S}_{n}\backslash \partial \overline{S}_n, \overline{\mathcal{K}}_n) \cup \{\overline{z}_e\}$\;
				\label{line:main:compute_update_2}
				\For{$\overline{z}=(z,q) \in \overline{Z}_\mathrm{update}$}
				{
					${\tt UpdateSM} (\overline{z},\overline{S}_n,J_n,\gamma_n,\Upsilon_n,J_n^\gamma,\mu_n,\kappa_n,\overline{\Delta t}_n)$\; \label{line:processinterior:updatecost_end_2}
				}
			}
		}
	}
\end{footnotesize}
\caption{\footnotesize ProcessInterior($S_n,\overline{S}_n,J_n,\gamma_n,\Upsilon_n,J_n^\gamma,\mu_n,\kappa_n,\overline{\Delta t}_n$)} \label{algorithm:processInterior_2}
\end{algorithm}
The extended iMDP algorithm is presented in Algorithms~\ref{algorithm:main_2}-\ref{algorithm:updateSM_2}. 
The algorithm incrementally refines two MDP sequences, namely $\{ \mathcal{M}_n \}_{n=0}^{\infty}$ and $\{ \overline{\mathcal{M}}_n \}_{n=0}^{\infty}$, and two holding time sequences, namely $\{ \Delta t_n \}_{n=0}^{\infty}$ and  $\{ \overline{\Delta t}_n \}_{n=0}^{\infty}$, that consistently approximate the original system in Eq.~\eqref{eqn:system} and the augmented system in Eq.~\eqref{eqn:extendedsystem} respectively. 
We associate with $z \in S_n$ a cost value $J_n(z,1)$, a control $\mu_n(z,1)$, a failure probability $\Upsilon_n(z)$ due to $\mu_n(\cdot,1)$, a min-failure probability $\gamma_n(z)$, a cost-to-go value $J^{\gamma}_n(z)$ induced by the obtained min-failure policy.
Similarly, we associate with $\overline{z} \in \overline{S}_n$ a cost value $J_n(\overline{z})$, a control $(\mu_n(\overline{z}),\kappa_n(\overline{z}))$.

As shown in Algorithm~\ref{algorithm:main_2}, initially, empty MDP models $\mathcal{M}_0$ and $\overline{\mathcal{M}}_0$ are created. The algorithm then executes $N$ iterations in which it samples states on the pre-specified part of the boundary $\partial D$, constructs the un-specified part of $\partial D$ and processes the interior of $D$. More specifically, at Line~\ref{line:main:updatestorage_2}, ${\tt UpdateDataStorage}(n-1,n)$ indicates that refined models in the $n^{th}$ iteration are constructed from models in the $(n-1)^{th}$ iteration, which can be implemented by simply sharing memory among iterations. Using rejection sampling, the procedure ${\tt SampleOnBoundary}$ at Line~\ref{line:main:sampleboundary_2} sample states in $\partial S$ and $\partial S \times [0,1]$ to add to $S_n$ and $\overline{S}_n$ respectively. We also initialize appropriate cost values for these sampled states. 

We conduct $K_{1,n}$ rounds to refine the MDP sequence $\{ \mathcal{M}_n \}_{n=0}^{\infty}$ as done in the original iMDP algorithm using the procedure ${\tt ConstructBoundary}$ (Line~\ref{line:main:constructboundary_2}). Thus, we can compute the cost function $J_n$ and the associated failure probability function $\Upsilon_n$ on $S_n \times \{1\}$. In the same procedure, we compute the min-failure probability function $\gamma_n$ as well as the min-failure cost function $J_n^\gamma$ on $S_n$. In other words, the algorithm effectively constructs approximate boundaries for $D$ and approximate cost-to-go functions $J_n$ on these approximate boundaries over iterations. To compute cost values for the interior $D^{o}$ of $D$, we conduct $K_{2,n}$ rounds of the procedure ${\tt ProcessInterior}$ (Line~\ref{line:main:processinterior_2}) that similarly refines the MDP sequence $\{ \overline{\mathcal{M}}_n \}_{n=0}^{\infty}$ in the augmented state space. We can choose the values of $K_{1,n}$ and $K_{2,n}$ so that we perform a large number of iterations to obtain stable boundary values before processing the interior domain when $n$ is small. In the following discussion, we will present in detail the implementations of these procedures.

In Algorithm~\ref{algorithm:constructboundary_2}, we show the implementation of the procedure ${\tt ConstructBoundary}$. We construct a finer MDP model $\mathcal{M}_{n}$ based on the previous model as follows. A state $z_\mathrm{s}$, is sampled from the interior of the state space $S$ (Line~\ref{line:constructboundary:add_interior_state_start_2}). 
The nearest state $z_\mathrm{near}$ to $z_\mathrm{s}$ (Line~\ref{line:constructboundary:compute_nearest_2}) in the previous model is used to construct an extended state $z_\mathrm{e}$ by using the procedure ${\tt ExtendBackwardsS}$ at Line~\ref{line:constructboundary:steer_back_2}. 
The extended states $z_\mathrm{e}$ and $(z_\mathrm{e},1)$ are added into $S_n$ and $\overline{S}_n$ respectively. The associated cost value $J_{n}(z_\mathrm{e},1)$, failure probability $\Upsilon_n(z_\mathrm{e})$, min-failure probability $\gamma_{n}(z_\mathrm{e})$, min-failure cost value $J_{n}^{\gamma}(z_\mathrm{e})$ and control $\mu_{n}(z_\mathrm{e})$ are initialized at Line~\ref{line:constructboundary:add_interior_state_2}. 

\begin{algorithm}[t]
\begin{footnotesize}
$\tau \leftarrow {\tt ComputeHoldingTime}(z,|S_{n}|,d_x)$\; \label{line:holdingtime_2}
\vspace{0.03in}
\tcp{Sample or discover $M_n=\Theta(\log(|S_n|))$ controls}
$U_n \leftarrow {\tt ConstructControlsS}(M_n,z,S_n,\tau)$\;
\label{line:Controls_2}
\label{line:improve:begin_2}
\For{$v \in U_n$}
{
        \label{line:discover_2}
	$(Z_\mathrm{near},P_n) \leftarrow {\tt ComputeTranProb}(z,v,\tau,S_n,f,F)$\; \label{line:pn_2}
	\vspace{0.05in}
	\tcp{Update cost}	
	$J \leftarrow \tau g(z,v) + \alpha^{\tau}\sum_{y \in Z_{\mathrm{near}}}P_n(y)J_n(y,1)$\; \label{line:update_J_2}
	\If{$J < J_n(z,1)$}
	{ 	
	  $p \leftarrow \sum_{y \in Z_{\mathrm{near}}}P_n(y)\Upsilon_n(y)$\; \label{line:update_upsilon_2}	
	  $(J_n(z,1), \Upsilon_n(z),\mu_n(z,1),\Delta t_n(z)) \leftarrow (J,p,v,\tau)$\; 
	  \label{line:improve_cost:end_2}
 	}
	\vspace{0.05in}
	\tcp{Update min-failure probability}
	$b \leftarrow \sum_{y \in Z_{\mathrm{near}}}P_n(y)\gamma_n(y)$\; \label{line:update_gamma_2}
	\If{$b < \gamma_n(z)$}
	{ 	
	  $J \leftarrow \tau g(z,v) + \alpha^{\tau}\sum_{y \in Z_{\mathrm{near}}}P_n(y)J^{\gamma}_n(y)$\;
	  $(\gamma_n(z),J_n^{\gamma}(z) ) \leftarrow (b,J)$\; 
	  \label{line:update_gamma_end_2}
 	}
}
\end{footnotesize}
\caption{\footnotesize ${\tt UpdateS}( z, S_n, J _{n},\gamma_n,\Upsilon_n,J_n^\gamma,\mu_{n},\Delta t_{n})$}
\label{algorithm:updateS_2}
\end{algorithm}

\begin{algorithm}[t]
\begin{footnotesize}
$\overline{\tau} \leftarrow {\tt ComputeHoldingTime}(\overline{z},|\overline{S}_{n}|,d_x+1)$\; \label{line:holdingtime_2}
\vspace{0.03in}
\tcp{Sample or discover $\overline{M}_n=\Theta(\log(|\overline{S}_n|))$ controls}
$\overline{U}_n \leftarrow {\tt ConstructControlsSM}(\overline{M}_n,\overline{z},\overline{S}_n,\overline{\tau})$\;
\label{line:ControlsSM_2}
\label{line:improve:beginSM_2}
\For{$\overline{v}=(v,c) \in \overline{U}_n$}
{
        \label{line:discoverSM}
	$(\overline{Z}_\mathrm{near},\overline{P}_n) \leftarrow {\tt ComputeTranProb}(\overline{z},\overline{v},\overline{\tau},\overline{S}_n,\overline{f},\overline{F})$\; \label{line:pn_2}
	$J \leftarrow \overline{\tau} g(z,v) + \alpha^{\overline{\tau}}\sum_{\overline{y}=(y,s) \in \overline{Z}_{\mathrm{near}}}\overline{P}_n(\overline{y})\big [1_{s=\gamma_n(y)}J_n^{\gamma}(y) +  \text{\ \ \ \ \ \ \ \ \ \ \ \ \ \ \ \ \ \ \ \ \ } 1_{\gamma_n(y) < s < \Upsilon_n(y)}J_n(\overline{y}) + 1_{s \geq \Upsilon_n(y)}J_n(y,1) \big ]$\; \label{line:update_J_SM_2}
	\vspace{0.03in}
	\tcp{Improved cost}
	\If{$J < J_n(\overline{z})$}
	{ 	
	  $(J_n(\overline{z}),\mu_n(\overline{z}), \kappa_n(\overline{z}), \overline{\Delta t}_n(\overline{z})) \leftarrow (J,v,c,\tau)$\; 
	  \label{line:improve:end_2}
 	}
}
\end{footnotesize}
\caption{\footnotesize ${\tt UpdateSM}( \overline{z},\overline{S}_n,J_n,\gamma_n,\Upsilon_n,J_n^\gamma,\mu_n,\kappa_n,\overline{\Delta t}_n)$}
\label{algorithm:updateSM_2}
\end{algorithm}

We then perform $L_n \ge 1$ updating rounds in each iteration (Lines~\ref{line:constructboundary:performupdate_start_2}-\ref{line:constructboundary:performupdate_end_2}). 
In particular, we construct the update-set $Z_\mathrm{update}$ consisting of $K_n=\Theta(|S_n|^{\dummyasyn})$ states and $z_\mathrm{e}$ where $|K_n| < |S_n|$. 
For each state $z$ in $Z_\mathrm{update}$, the procedure ${\tt UpdateS}$ as shown in Algorithm~\ref{algorithm:updateS_2} implements the following Bellman update:
$$
J_{n} (z,1) = \min_{v \in U_n(z) } \{ G_n(z,v) + \alpha^{\Delta t_n(z)} \EE_{P_n}[ J_{n-1}(y) | z,v]\}.
$$
The details of the implementation are as follows. A set of $U_n$ controls is constructed using the procedure ${\tt ConstructControlsS}$ where $|U_n|=\Theta(\log(|S_n|))$ at Line~\ref{line:Controls_2}. 
For each $v \in U_n$, we construct the support $Z_\mathrm{near}$ and compute the transition probability $P_n(\cdot \given z,v)$ consistently over $Z_\mathrm{near}$ from the procedure ${\tt ComputeTranProb}$ (Line \ref{line:pn_2}). 
The cost values for the state $z$ and controls in $U_n$ are computed at Lines~\ref{line:update_J_2}. 
We finally choose the best control in $U_n$ that yields the smallest updated cost value (Line~\ref{line:improve:end_2}). Correspondingly, we improve the min-failure probability $\gamma_n$ and its induced min-failure cost value $J_n^{\gamma}$ in Lines~\ref{line:update_gamma_2}-\ref{line:update_gamma_end_2}.
%
%

Similarly, in Algorithm~\ref{algorithm:processInterior_2}, we carry out the sampling and extending process in the augmented state space $\overline{S}$ to refine the MDP sequence $\overline{\mathcal{M}}_{n}$ (Lines~\ref{line:processinterior:sample_2}-\ref{line:processinterior:extend_2}). In this procedure, if an extended node has a martingale state that is below the corresponding min-failure probability, we initialize the cost value for extended node with a very large constant $\mathcal{C}$ representing $+\infty$ (see Lines~\ref{line:main:checkminprob_1}-\ref{line:main:checkminprob_2}). Otherwise, we initialize the extended node as seen in Lines~\ref{line:updateextended_1}-\ref{line:updateextended_2}. We then execute $\overline{L}_n$ rounds (Lines~\ref{line:processinterior:updatecost_start_2}-\ref{line:processinterior:updatecost_end_2}) to update the cost-to-go $J_n$ for \textit{states in the interior $D^{o}$ of $D$} using the procedure ${\tt UpdateSM}$ as shown in Algorithm~\ref{algorithm:updateSM_2}. When a state $\overline{z} \in \overline{S}_n$ is updated in ${\tt UpdateSM}$, we perform the following Bellman update:
\begin{align*}
\resizebox{1.0\hsize}{!}{$\displaystyle J_{n} (\overline{z}) = \min_{(v,c) \in \overline{U}_n(z) } \{ \overline{G}_n(z,v) + \alpha^{\overline{\Delta t}_n(z)} \EE_{\overline{P}_n}[ J_{n-1}(\overline{y}) | \overline{z},(v,c)]\},$}
\end{align*}
where the control set $\overline{U}_n$ is constructed by the procedure $\tt ConstructControlsSM$, and the transition probability $\overline{P}_n(\cdot|\overline{z},(v,c))$ consistently approximates the augmented dynamics in Eq.~\eqref{eqn:extendedsystem}. To implement the above Bellman update at Line~\ref{line:update_J_SM_2} in Algorithm~\ref{algorithm:updateSM_2}, we make use of the characteristics presented in Section~\ref{subsection:characterization} where the notation $1_{A}$ is $1$ if the event $A$ occurs and $0$ otherwise. That is, when the martingale state $s$ of a state $\overline{y}=(y,s)$ in the support $\overline{Z}_{near}$ is at least $\Upsilon_n(y)$, we substitute $J_n(\overline{y})$ with $J_n(y,1)$. Similarly, when the martingale state $s$ is equal to $\gamma_n(y)$, we substitute $J_n(\overline{y})$  with $J_n^{\gamma}(y)$.

\subsection{Feedback Control} \label{section:algorithm:control}

\begin{figure*}[t]
\begin{center}
  \subfigure{
  \label{figFeedback}
  \includegraphics[width=85mm]{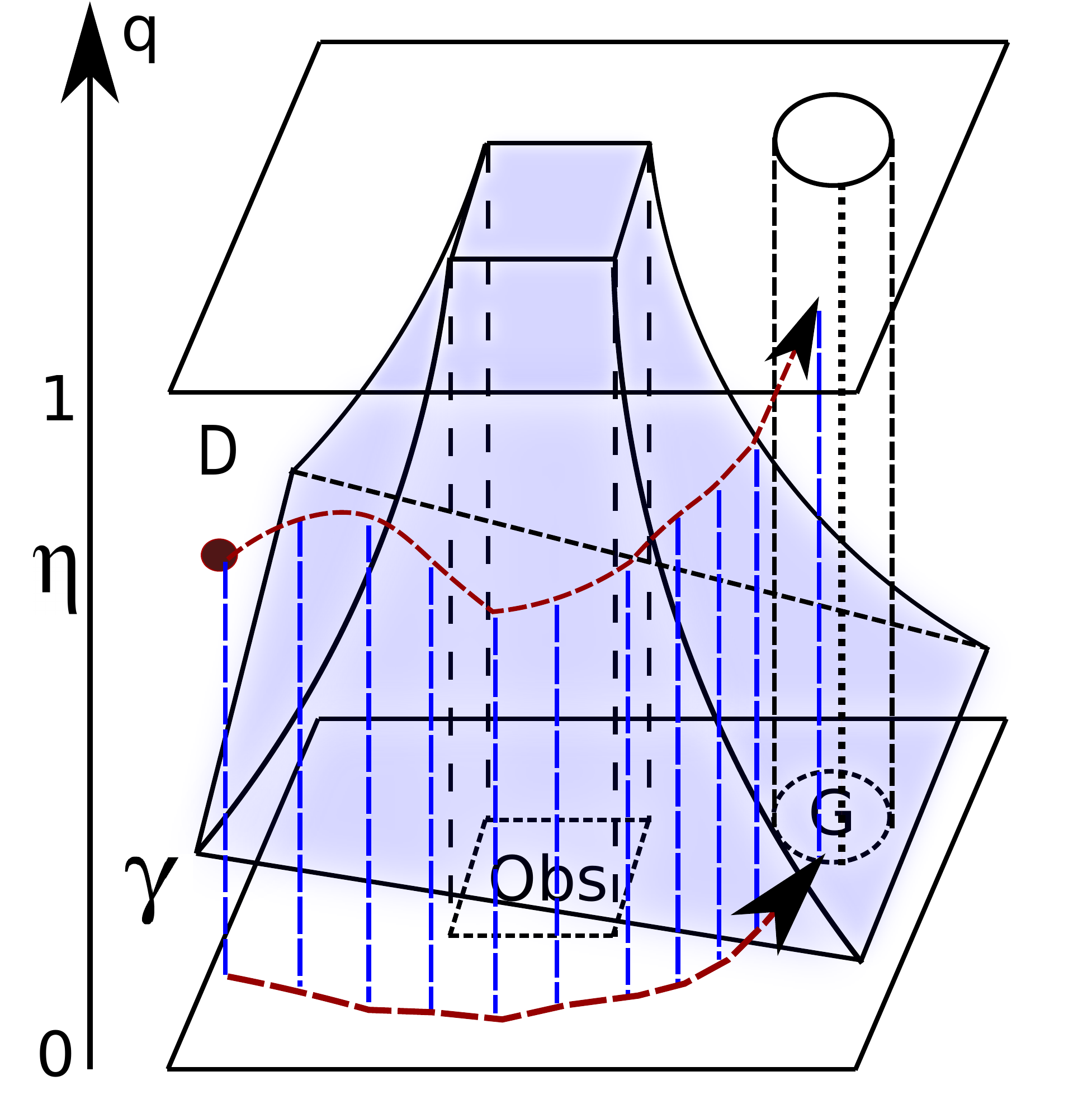}  
  }
  \hfill
  \subfigure{
  \label{figFeedback_Extra}
  \includegraphics[width=85mm]{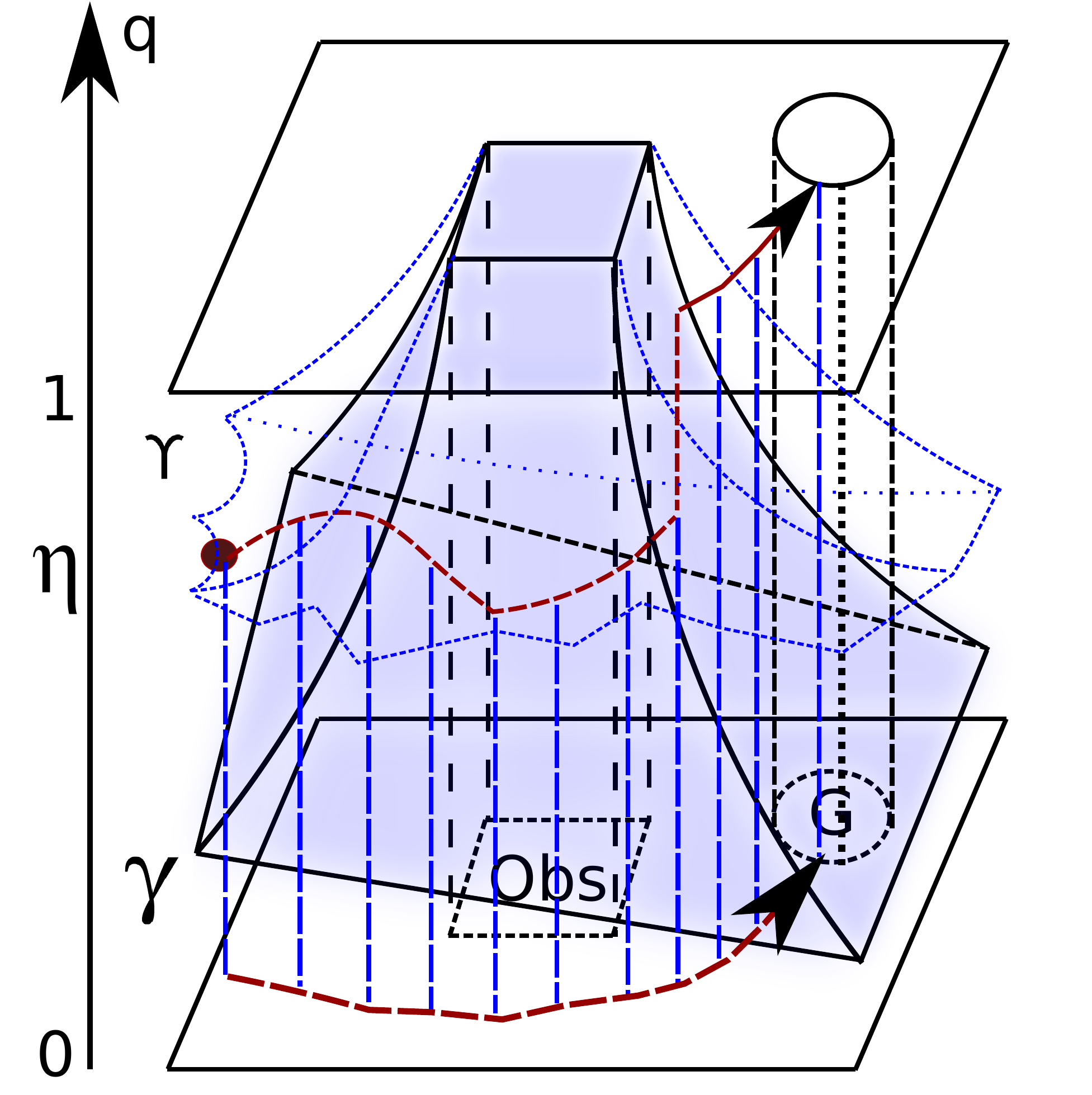}  
  }
  \caption{In Fig.~\ref{figFeedback}, we show a feedback-controlled trajectory of $\mathcal{OPT}1$ and $\mathcal{OPT}2$. In the augmented state space $\overline{S}$, a feedback control policy is a deterministic Markov policy as a function of an augmented state $(x,q)$. As the system actually evolves  in  the original state space $S$, and the martingale state $q$ can be seen as a random parameter at each state $x$, the feedback control policy is a randomized policy. In Fig.~\ref{figFeedback_Extra}, we show a modified feedback-controlled trajectory. We continue the illustration in Fig.~\ref{figFeedback}. When the martingale state along the trajectory is at least the corresponding value provided by $\Upsilon$, the system starts following a deterministic control policy $\mu_n(\cdot,1)$ of the unconstrained problem.} 
\end{center}
\end{figure*}

\begin{algorithm}[t]
\begin{footnotesize}
$z_\mathrm{nearest} \leftarrow {\tt Nearest}(z, S_n,1)$\;
\eIf{$q \geq \gamma_n(z_\mathrm{nearest})$}
{
	\vspace{0.05in}
	\tcp{Switch to a deterministic control policy}
	\Return{$ \big (\varphi(\overline{z})=( \mu_n(z_\mathrm{nearest}), 0), \Delta t_n(z_\mathrm{nearest}) \big)$} \;
	\label{line:deterpolicy}
}
{
	\vspace{0.05in}
	\tcp{Perform a Bellman update}
	$(J_{min},v_{min},c_{min}) \leftarrow (+\infty,\emptyset,\emptyset)$ \;	 
	\label{line:randompolicy_1}
	$\overline{\tau} \leftarrow {\tt ComputeHoldingTime}(\overline{z},|\overline{S}_{n}|,d_x+1)$\;
	\vspace{0.03in}
	\tcp{Construct $\overline{M}_n=\Theta(\log(|\overline{S}_n|))$ controls}
	$\overline{U}_n \leftarrow {\tt ConstructControlsSM}(\overline{M}_n,\overline{z},\overline{S}_n,\overline{\tau})$\;
	\For{$\overline{v}=(v,c) \in \overline{U}_n$}
	{
			$(\overline{Z}_\mathrm{near},\overline{P}_n) \leftarrow {\tt ComputeTranProb}(\overline{z},\overline{v},\overline{\tau},\overline{S}_n,\overline{f},\overline{F})$\;
		$J \leftarrow \overline{\tau} g(z,v) + \alpha^{\overline{\tau}}\sum_{\overline{y}=(y,s) \in \overline{Z}_{\mathrm{near}}}\overline{P}_n(\overline{y})\big [1_{s=\gamma_n(y)}J_n^{\gamma}(y) + 1_{\gamma_n(y) < s < \Upsilon_n(y)}J_n(\overline{y}) +  1_{s \geq \Upsilon_n(y)}J_n(y,1) \big ]$\; 
		\vspace{0.03in}
		\tcp{Improved cost}
		\If{$J < J_{min}$}
		{ 	
		     $(J_{min},v_{min},c_{min}) \leftarrow (J,v,c)$ \;	 
	 	}
	}
	\Return{$\big( \varphi(\overline{z})=( v_{min}, c_{min}), \overline{\tau} \big )$} \; 	\label{line:randompolicy_2}
}
\end{footnotesize}
\caption{\footnotesize ${\tt \text{ Risk Constrained Policy} } ( \overline{z}=(z,q) \in \overline{S},n)$}
\label{algorithm:Policy_2}
\end{algorithm}

At the $n^{th}$ iteration, given a state $x \in S$ and a martingale component $q$, to find a policy control $(v,c)$, we perform a Bellman update based on the approximated cost-to-go $J_n$ for the augmented state $(x,q)$. During the holding time $\overline{\Delta t}_n$, the original system takes the control $v$ and evolves in the original state space $S$ while we simulate the dynamics of the martingale component under the martingale control $c$. After this holding time period, the augmented system has a new state $(x',q')$, and we repeat the above process.

Figure~\ref{figFeedback} visualizes how feedback policies look in the original and augmented state spaces. In the augmented state space $\overline{S}$, a feedback control policy is a deterministic Markov policy as a function of an augmented state $(x,q)$. As the system actually evolves  in  the original state space $S$, and the martingale state $q$ can be seen as a random parameter at each state $x$, the feedback control policy is a randomized policy.

Using the characteristics presented in Section~\ref{subsection:characterization}, we infer that when a certain condition meets, the system can start following a deterministic control policy. More precisely, we recall that for all $\eta \in [\Upsilon(z), 1]$, we have $J^*(z,\eta)=J^*(z,1)$. Thus, starting from any augmented state $(z,\eta)$ where $\eta > \Upsilon(z)$, we can solve the problem as if the failure probability were $1.0$ and use optimal control policies of the unconstrained problem from the state $z$. We illustrate this idea in Fig.~\ref{figFeedback_Extra}. As we can see, when the martingale state along the trajectory is at least the corresponding value provided by $\Upsilon$, the system starts following a deterministic control policy $\mu_n(\cdot,1)$ of the unconstrained problem.

Algorithm~\ref{algorithm:Policy_2} implements the above feedback policy. As shown in this algorithm, Line~\ref{line:deterpolicy} returns a deterministic policy of the unconstrained problem if the martingale state is large enough, and Lines~\ref{line:randompolicy_1}-\ref{line:randompolicy_2} perform a Bellman update to find the best augmented control if otherwise. When the system starts using deterministic policies of the unconstrained problem, we can set the martingale state to $1.0$ and set the optimal martingale control to $0$ in the following control period.

\subsection{Complexity} 
The time complexity per iteration of Algorithms \ref{algorithm:main_2}-\ref{algorithm:updateSM_2} is $O \big ( |\overline{S}_n|^{\dummyasyn}(\log{|\overline{S}_n|})^2 \big)$. The space complexity of the iMDP algorithm is $O(|\overline{S}_n|)$ where $|\overline{S}_n|=\Theta(n)$ due to our sampling strategy. 

\section{Analysis} \label{section:analysis}
In this section, we present main results on the performance of the extended iMDP algorithm with brief explanation. More detailed proofs can be found in~\cite{Huynh2012.ArXiV}. 

We first review the following key results of the approximating Markov chain method when no additional risk constraints are considered~\cite{Kushner2000}. Local consistency implies the convergence of continuous-time interpolations of the trajectories of the controlled Markov chain to the trajectories of the stochastic dynamical system described by Eq.~\eqref{eqn:system}. In particular, previous results in~\cite{huynh.karaman.ea:icra12} show that $J_n(\cdot,1)$ returned from the iMDP algorithm converges uniformly to $J^*(\cdot,1)$ in probability. That is, we are able to compute $J^*(\cdot,1)$ in an incremental manner without directly computing $J^*_n(\cdot,1)$. As a consequence, it follows that $\Upsilon_n$ converges to $\Upsilon$ uniformly in probability. Using the same proof, we conclude that $\gamma_n(\cdot)$ and $J_n^\gamma(\cdot)$ converges uniformly to $\gamma(\cdot)$ and $J^*(\cdot,\gamma)$ in probability respectively. Therefore, we have incrementally constructed the boundary values on $\partial D$ of the equivalent stochastic target problem presented in Eqs.~\eqref{eqn:stpcost}-\eqref{eqn:asconstraint}. These results are established based on the approximation of the dynamics in Eq.~\eqref{eqn:system} using the MDP sequence $\{ \mathcal{M}_n \}_{n=0}^{\infty}$. 

Similarly, the uniform convergence of $J_n(\cdot,\cdot)$ to $J^*(\cdot,\cdot)$ in probability on the interior of $D$ is followed from the approximation of the dynamics in Eq.~\eqref{eqn:extendedsystem} using the MDP sequence $\{ \overline{\mathcal{M}}_n \}_{n=0}^{\infty}$. In the following theorem, we formally summarize the key convergence results of the extended iMDP algorithm.
\begin{theorem}
\label{theoremProbSoundnessUniformConvergence}
Let $\mathcal{M}_n$ and $\overline{\mathcal{M}}_n$ be two MDPs with discrete states constructed in $S$ and $\overline{S}$ respectively, and let $J_{n}:\overline{S}_n \rightarrow \overline{\reals}$ be the cost-to-go function returned by the extended iMDP algorithm at the $n^{th}$ iteration. Let us define $||b||_{X}=\sup_{z\in X}b(z)$ as the sup-norm over a set $X$ of a function $b$ with a domain containing $X$. We have the following random variables converge in probability:
\begin{enumerate}
\item $\plim_{n \rightarrow \infty} ||J_{n}(\cdot,1)-J^*(\cdot,1)||_{S_n}=0,$
\item $\plim_{n \rightarrow \infty} ||\Upsilon_{n}-\Upsilon||_{S_n}=0,$
\item $\plim_{n \rightarrow \infty} ||\gamma_{n}-\gamma||_{S_n}=0,$
\item $\plim_{n \rightarrow \infty} ||J^{\gamma}_{n}-J^{\gamma}||_{S_n}=0,$
\item $\plim_{n \rightarrow \infty} ||J_n-J^*||_{\overline{S}_n}=0.$
\end{enumerate}
The first four events construct the boundary values on $\partial D$ in probability, which leads to the probabilistically sound property of the extended iMDP algorithm. The last event asserts the asymptotically optimal property through the convergence of the approximating cost-to-go function $J_n$ to the optimal cost-to-go function $J^*$ on the augmented state space $\overline{S}$.
\end{theorem}

\section{Experiments} \label{section:experiments}

\begin{figure*}[t]
\begin{center}
  \subfigure[Policy on $\mathcal{M}_{500}$.]{
  \label{figsubConstruct1}
  \includegraphics[width=47mm,bb= 78 210 538 582]{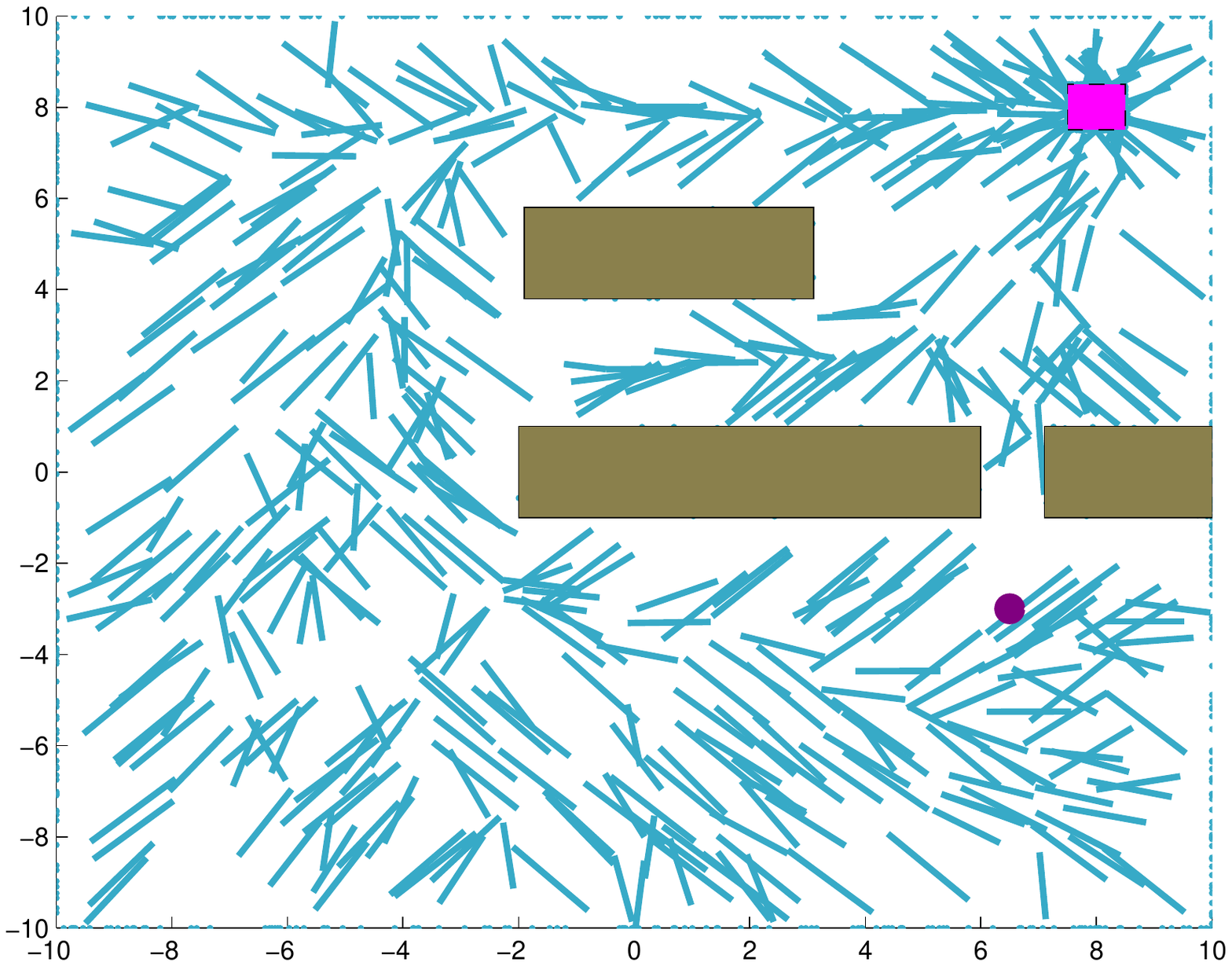}
  }
  \hfill
  \subfigure[Policy on $\mathcal{M}_{1000}$.]{
  \label{figsubConstruct2}
  \includegraphics[width=47mm,bb= 78 210 538 582]{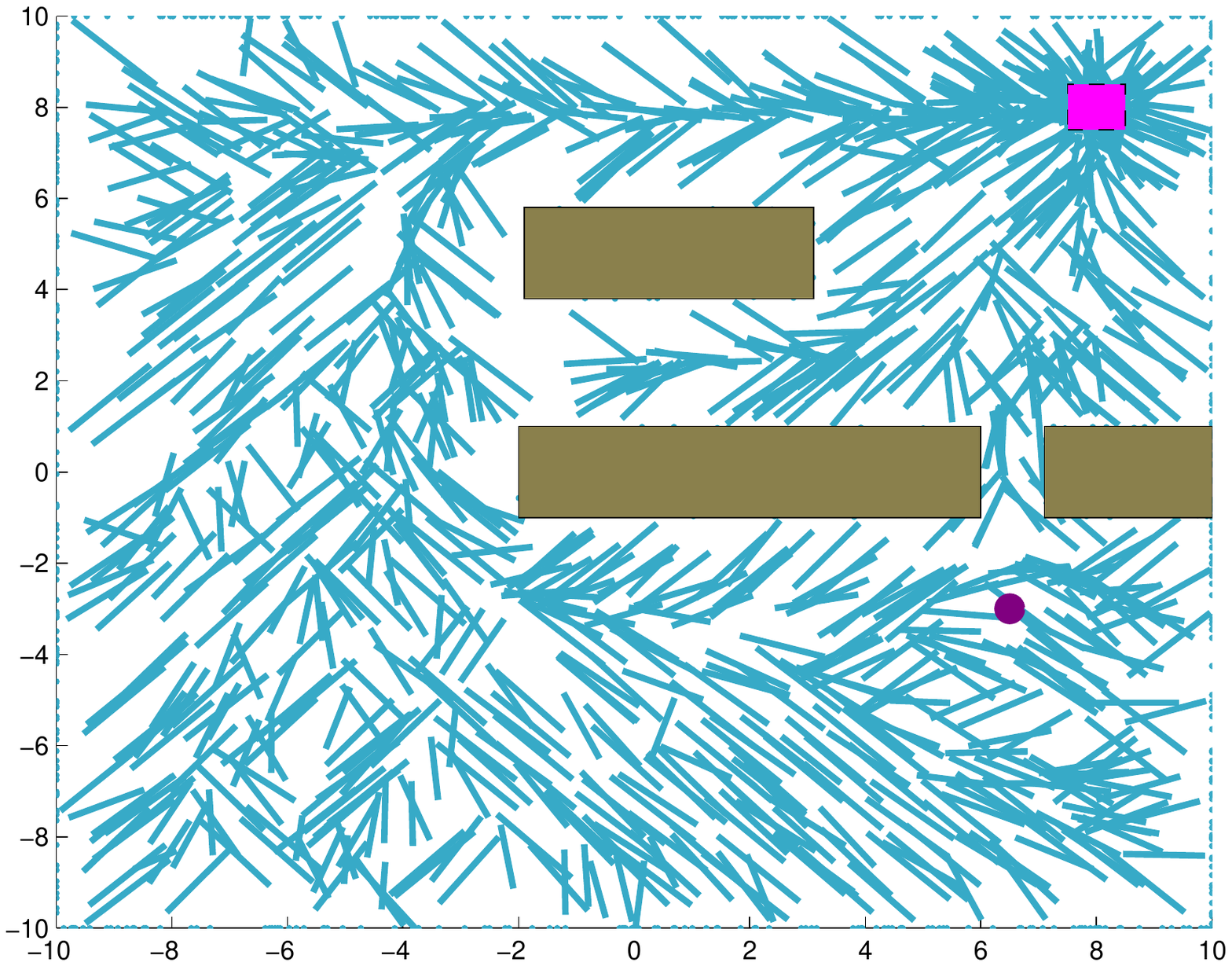}
  }
  \hfill
  \subfigure[Policy on $\mathcal{M}_{3000}$.]{
  \label{figsubConstruct3}
  \includegraphics[width=47mm,bb= 78 210 538 582]{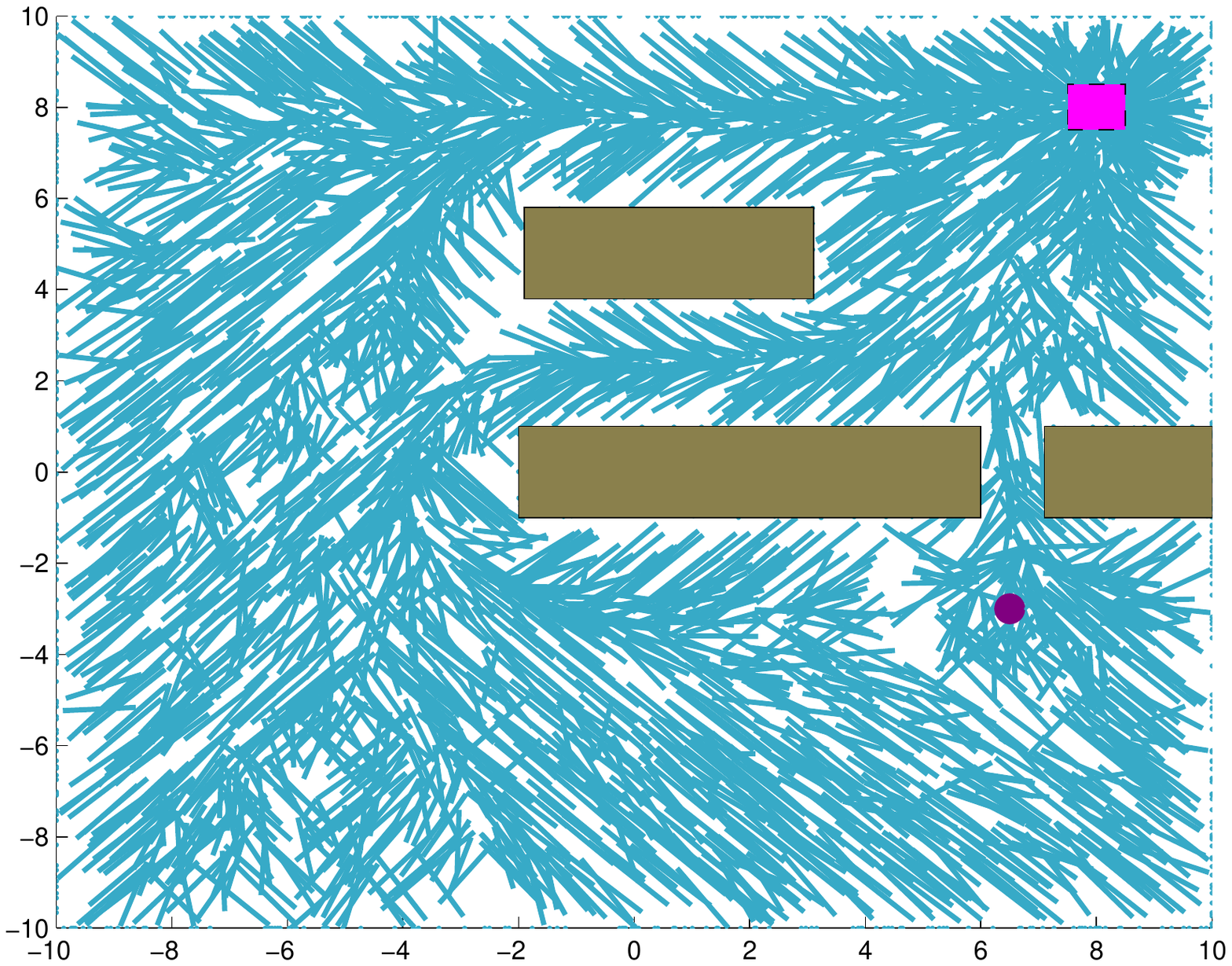}
  }
  \subfigure[Markov chain implied by $\mathcal{M}_{200}$.]{
  \label{figsubConstruct4}
  \includegraphics[width=47mm,bb= 78 210 538 582]{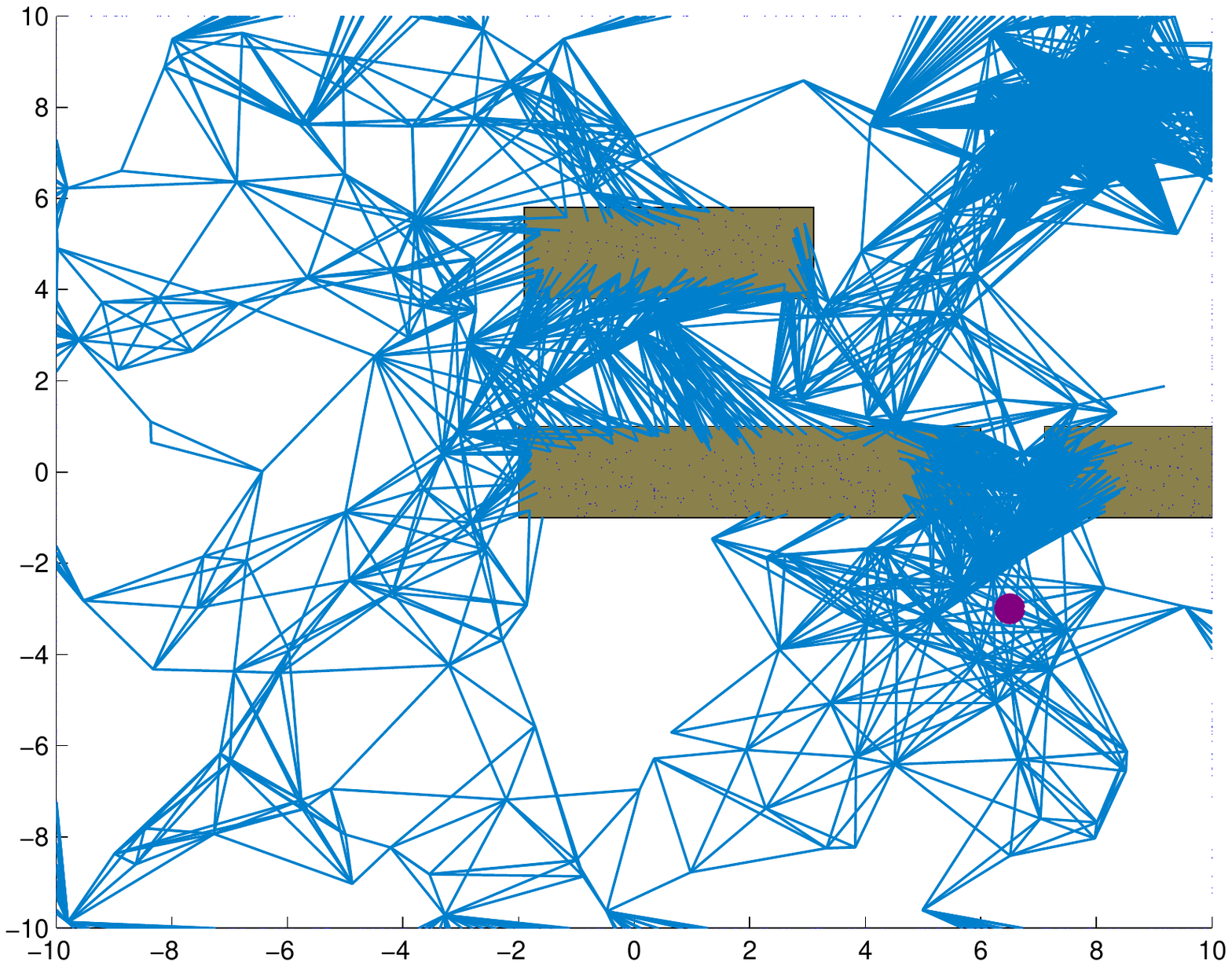}
  }
  \hfill
  \subfigure[Markov chain implied by $\mathcal{M}_{500}$.]{
  \label{figsubConstruct5}
  \includegraphics[width=47mm,bb= 78 210 538 582]{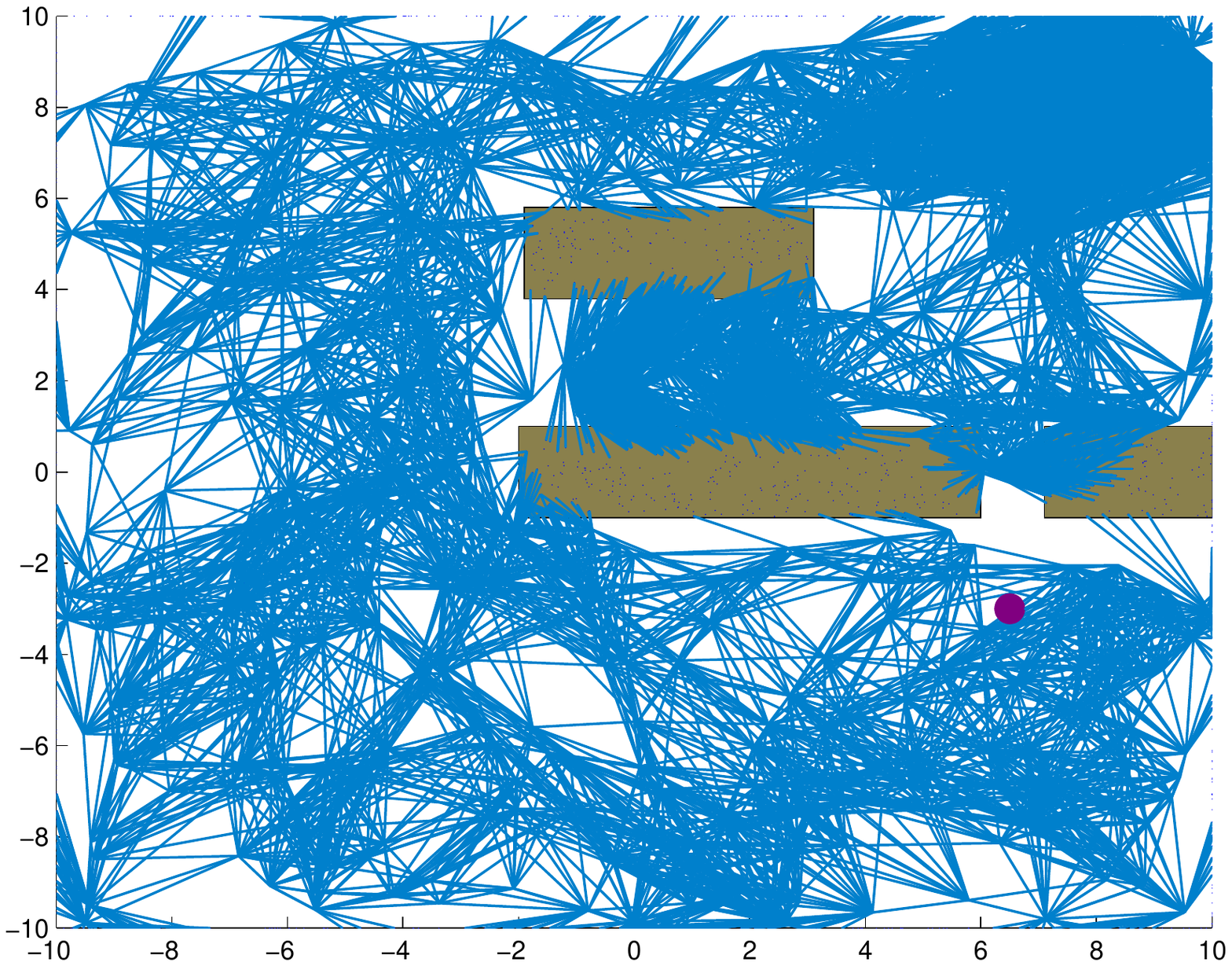}
  }
  \hfill
  \subfigure[Markov chain implied by $\mathcal{M}_{1000}$.]{
  \label{figsubConstruct6}
  \includegraphics[width=47mm,bb=  140 60 1100 860]{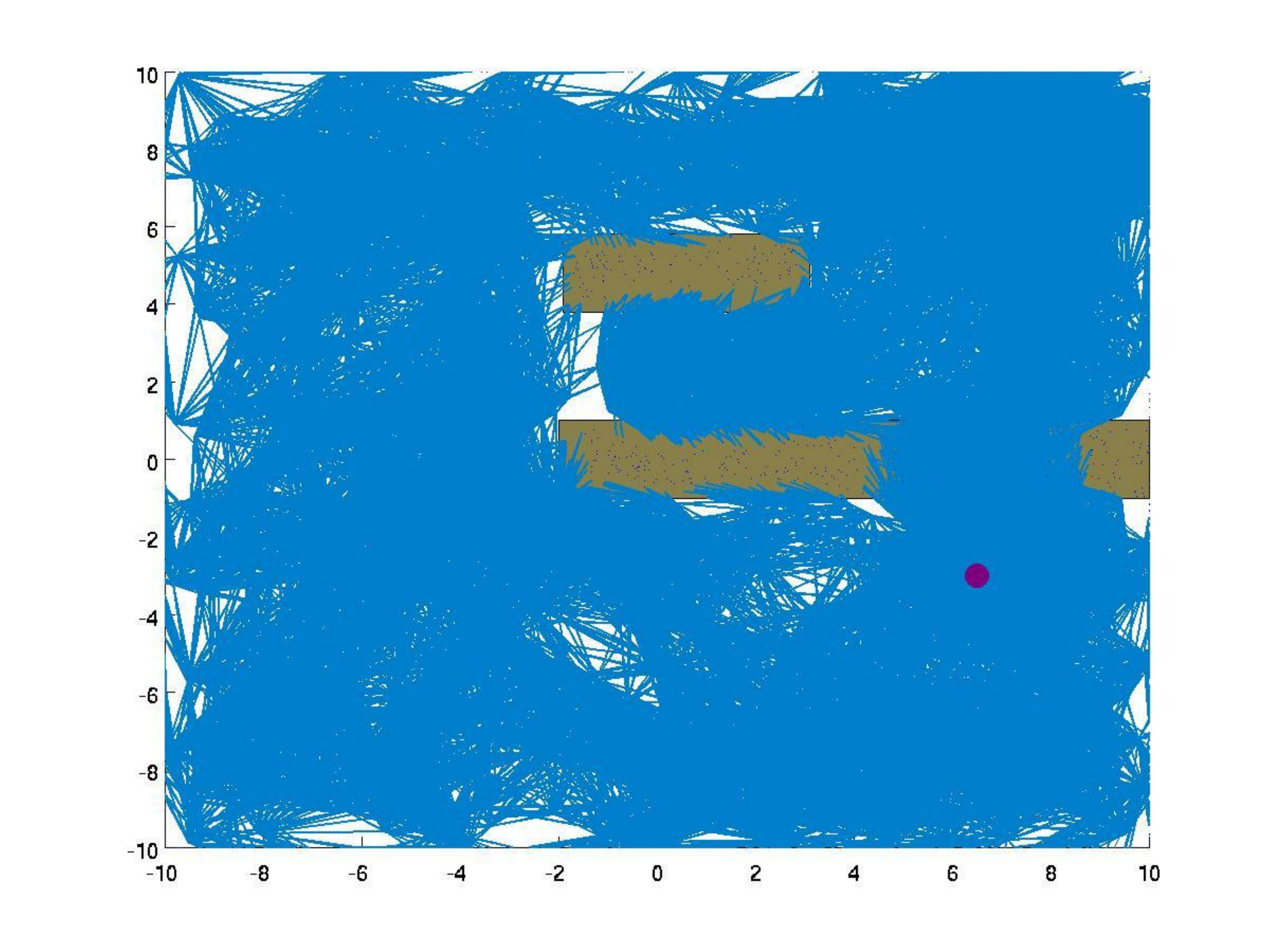}
  }
  \caption{A system with stochastic single integrator dynamics in a cluttered environment. The standard deviation of noise in $x$ and $y$ directions is $0.5$. The cost function is the sum of total energy spent to reach the goal, which is measured as the integral of square of control magnitude, and a terminal cost, which is $-1000$ for the goal region ($G$) and $10$ for the obstacle region ($\Gamma$), with a discount factor $\alpha=0.9$. Figures~\ref{figsubConstruct1}-\ref{figsubConstruct3} depict anytime policies on the boundary $S \times 1.0$ over iterations. Figures~\ref{figsubConstruct4}-\ref{figsubConstruct6} show the Markov chains created by anytime policies on $\mathcal{M}_n$ over iterations.}
  \label{figConstruct}
\end{center}
\end{figure*}

\begin{figure*}[t]
\begin{center}
  \subfigure[Policy on $\mathcal{M}_{4000}$.]{
  \label{figBinfo1}
  \includegraphics[width=47mm,bb= 78 210 538 582]{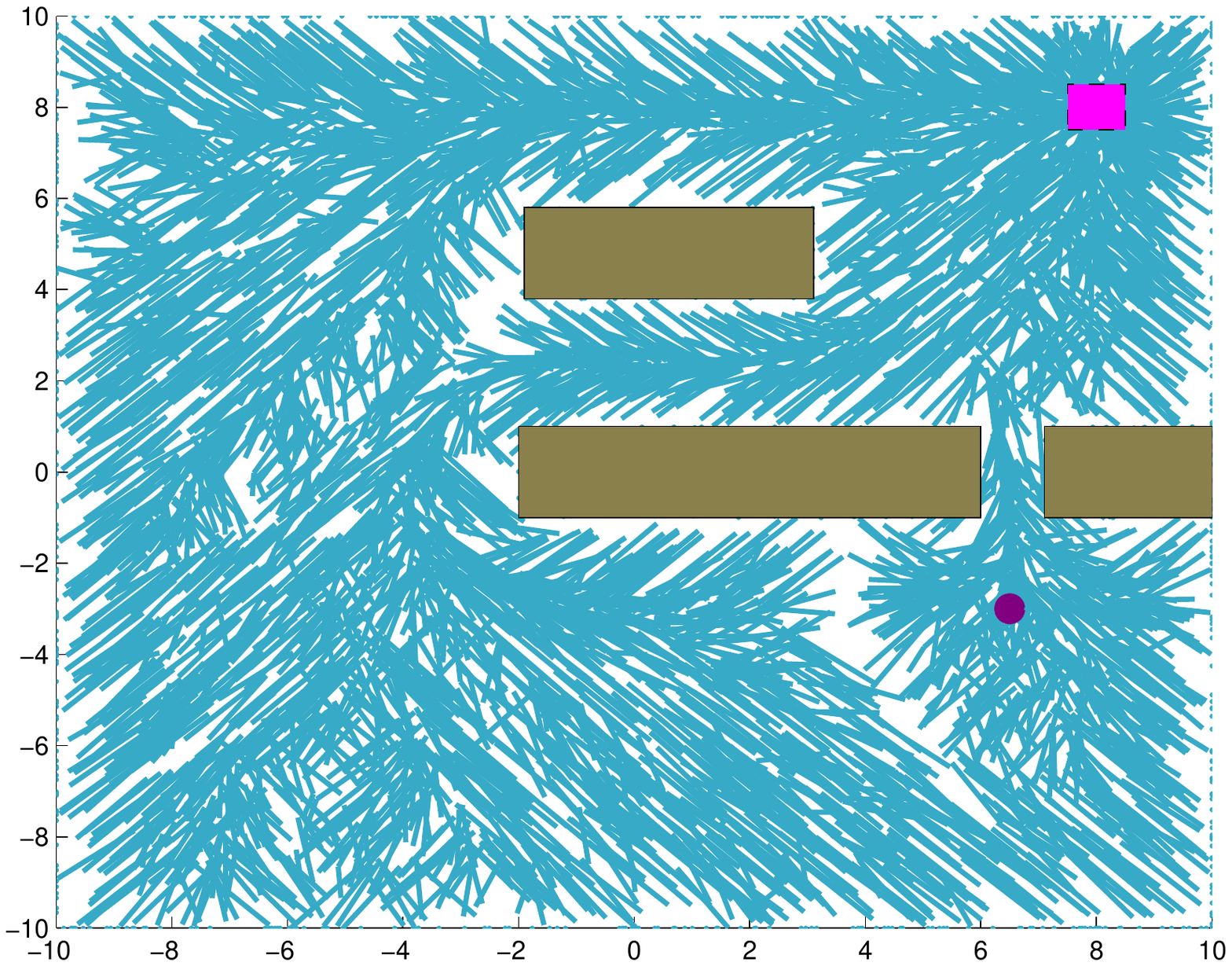}
  }
  \hfill
  \subfigure[Value function $J_{4000,1.0}$.]{
  \label{figBinfo2}
  \includegraphics[width=47mm,bb= 78 210 538 582]{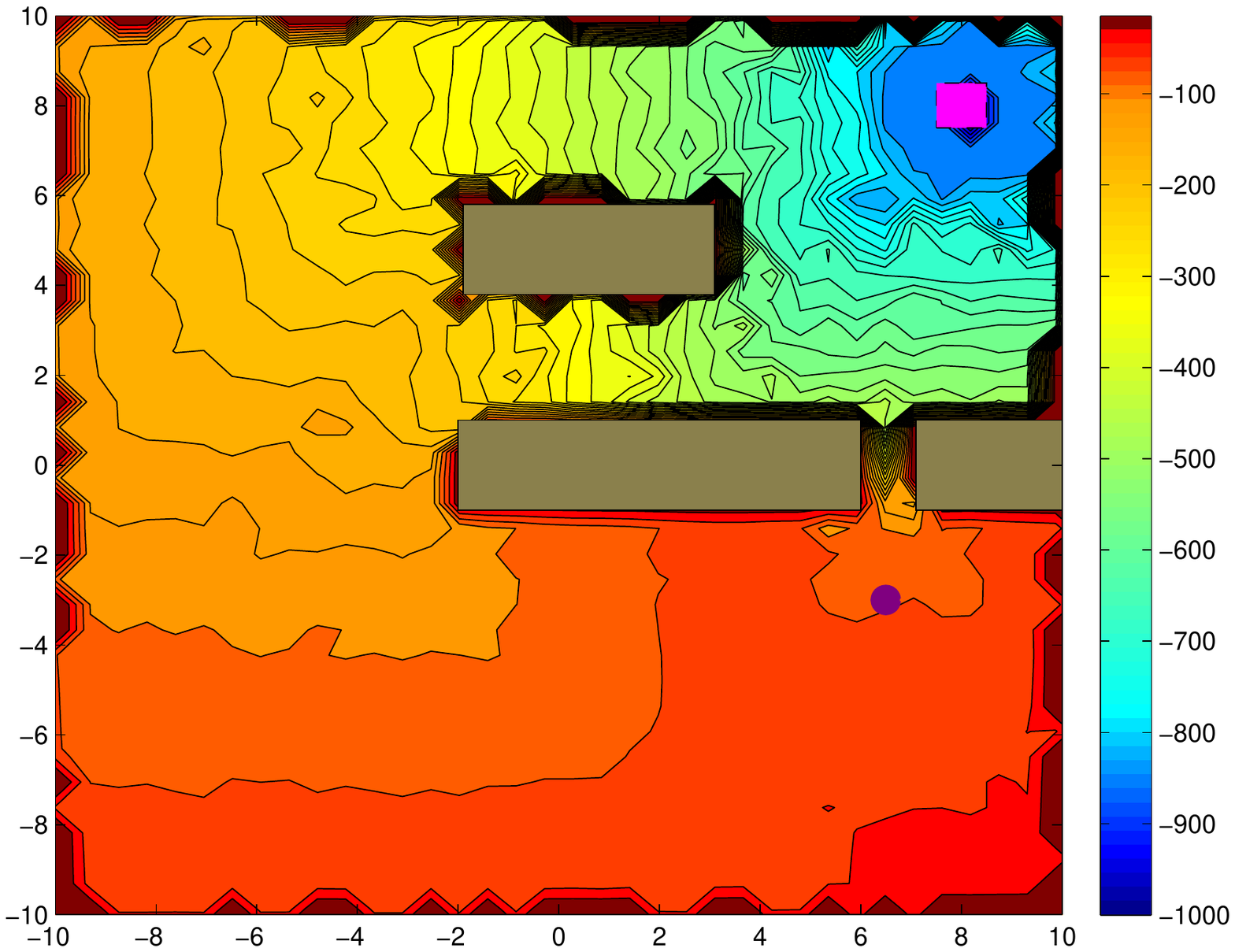}
  }
  \hfill
  \subfigure[Collision probability $\Upsilon_{4000}$.]{
  \label{figBinfo3}
  \includegraphics[width=47mm,bb= 78 210 538 582]{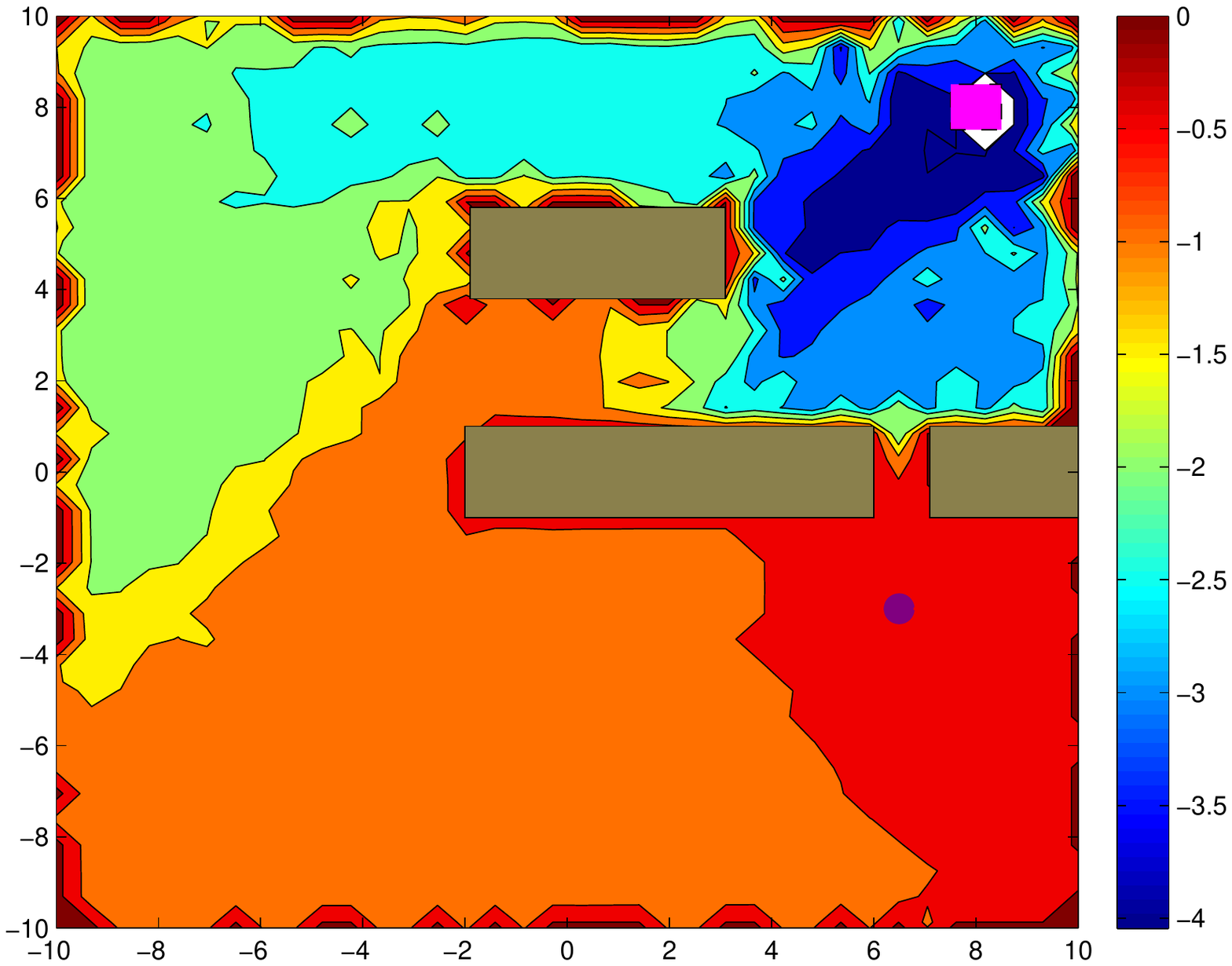}
  }
  \hfill
  \subfigure[Policy map induced by $\gamma_{4000}$.]{
  \label{figBinfo4}
  \includegraphics[width=47mm,bb= 78 210 538 582]{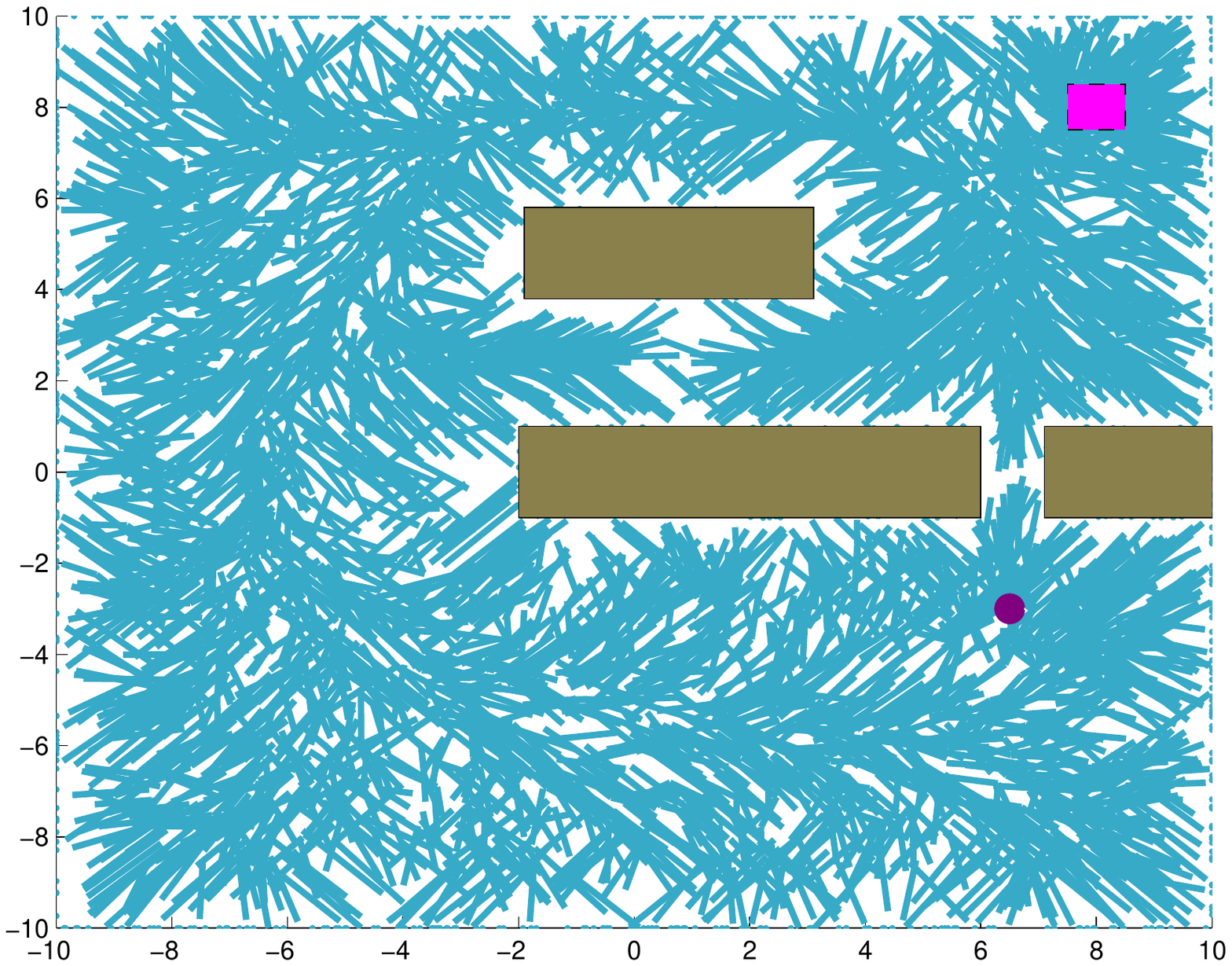}
  }
  \hfill
  \subfigure[Value function $J^{\gamma}_{4000}$.]{
  \label{figBinfo5}
  \includegraphics[width=47mm,bb= 78 210 538 582]{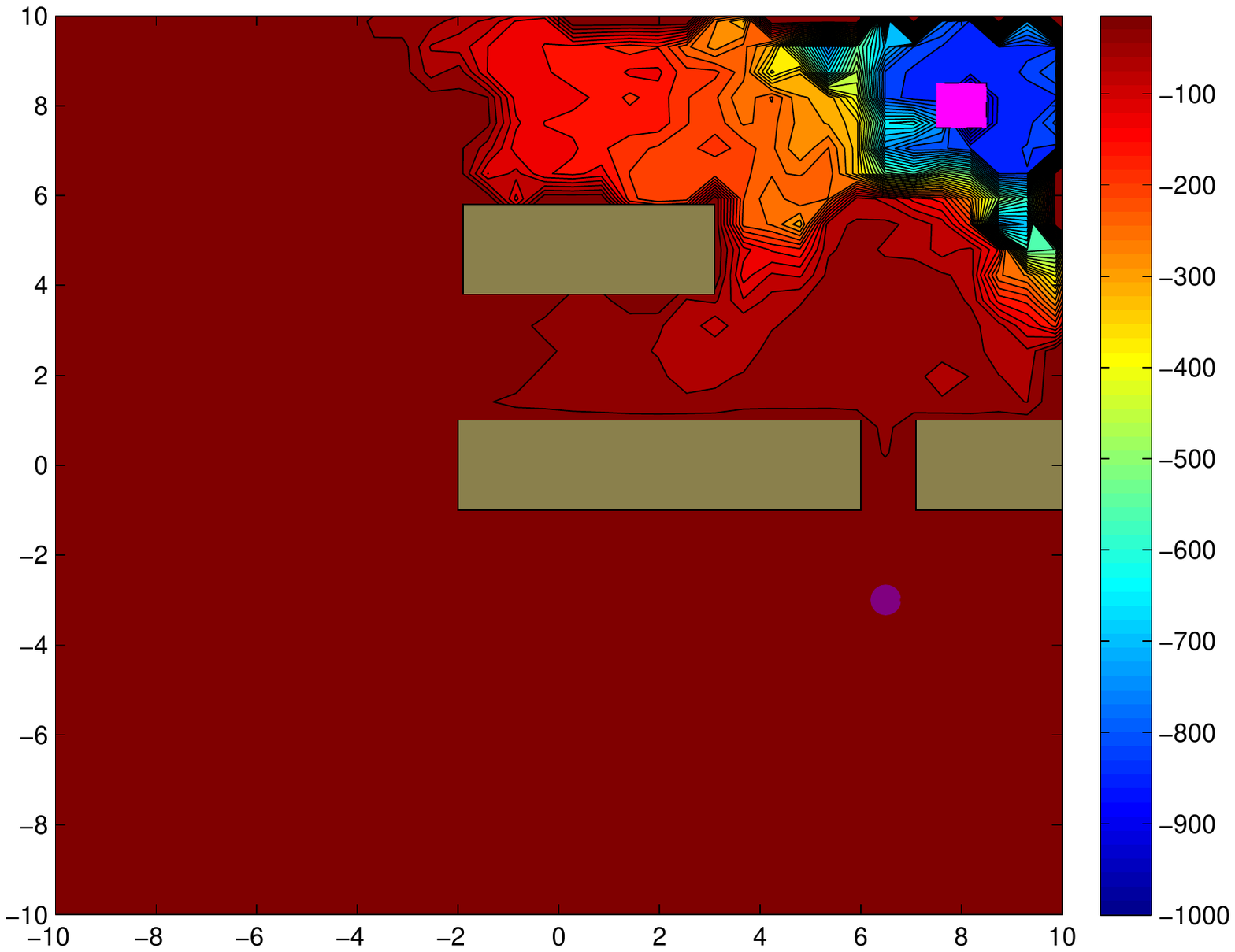}
  }
  \hfill
  \subfigure[Min-collision prob. $\gamma_{4000}$.]{
  \label{figBinfo6}
  \includegraphics[width=47mm,bb= 78 210 538 582]{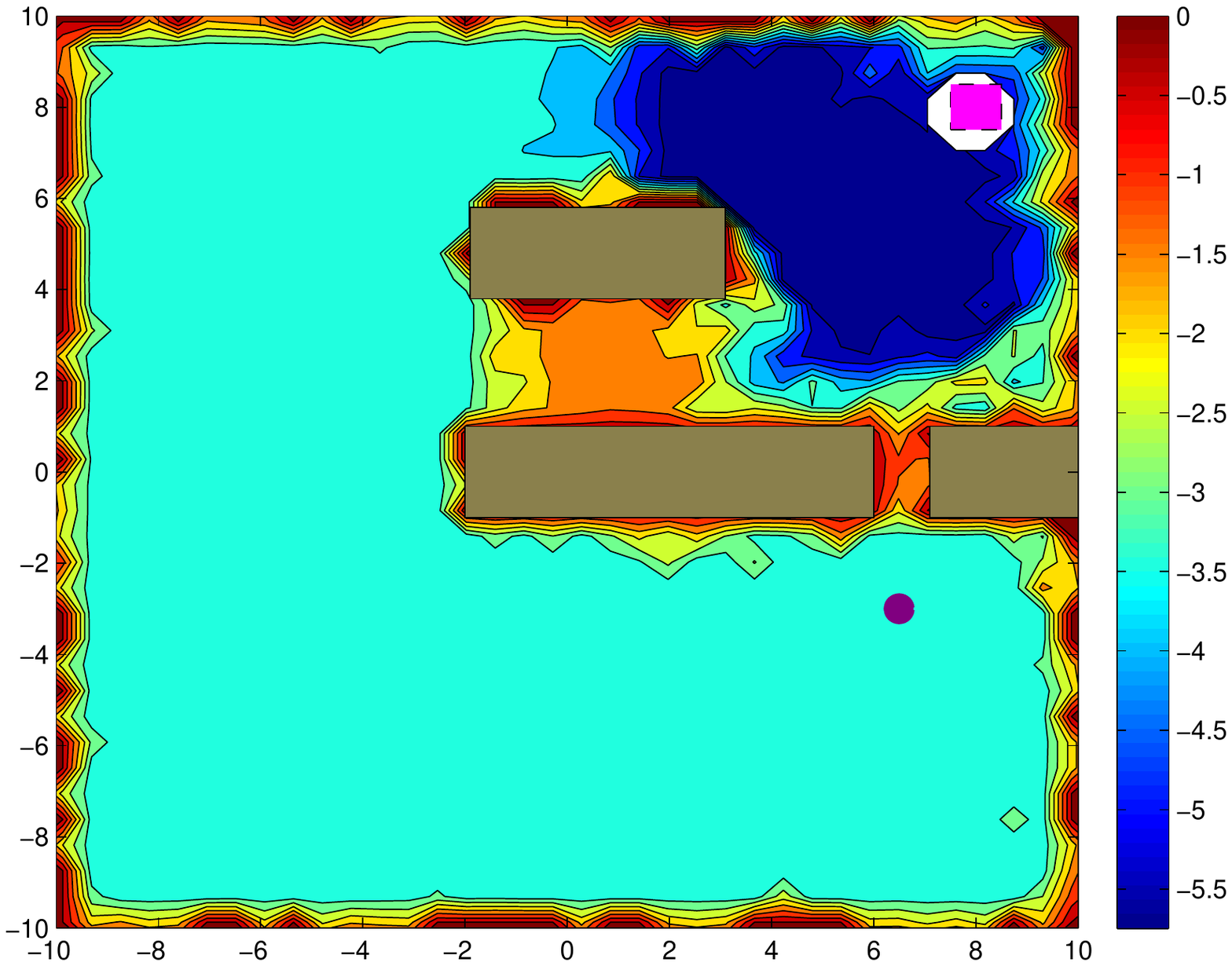}
  }
  \caption{ Figures~\ref{figBinfo1}-\ref{figBinfo3} shows a policy map, cost value  function and the associated collision probability function for the unconstrained problem after 4000 iterations. Similar, Figures~\ref{figBinfo4}-\ref{figBinfo6} show a policy map, the associated value function, and the min-collision probability function after 4000 iterations. These values provide the boundary values for the stochastic target problem. For the unconstrained problem, the policy map encourages the system to go through the narrow corridors with low cost-to-go values and high probabilities of collision. In contrast, the policy map from the min-collision probability problem encourages the system to detour around the obstacles with high cost-to-go values and low probabilities of collision.}
  \label{figBoundaryInfo}
\end{center}
\end{figure*}

\begin{figure*}[t]
\begin{center}
  \subfigure[Policy on $\overline{\mathcal{M}}_{200}$]{
  \label{figsubConstruct7}
  \includegraphics[width=47mm,bb= 78 210 538 582]{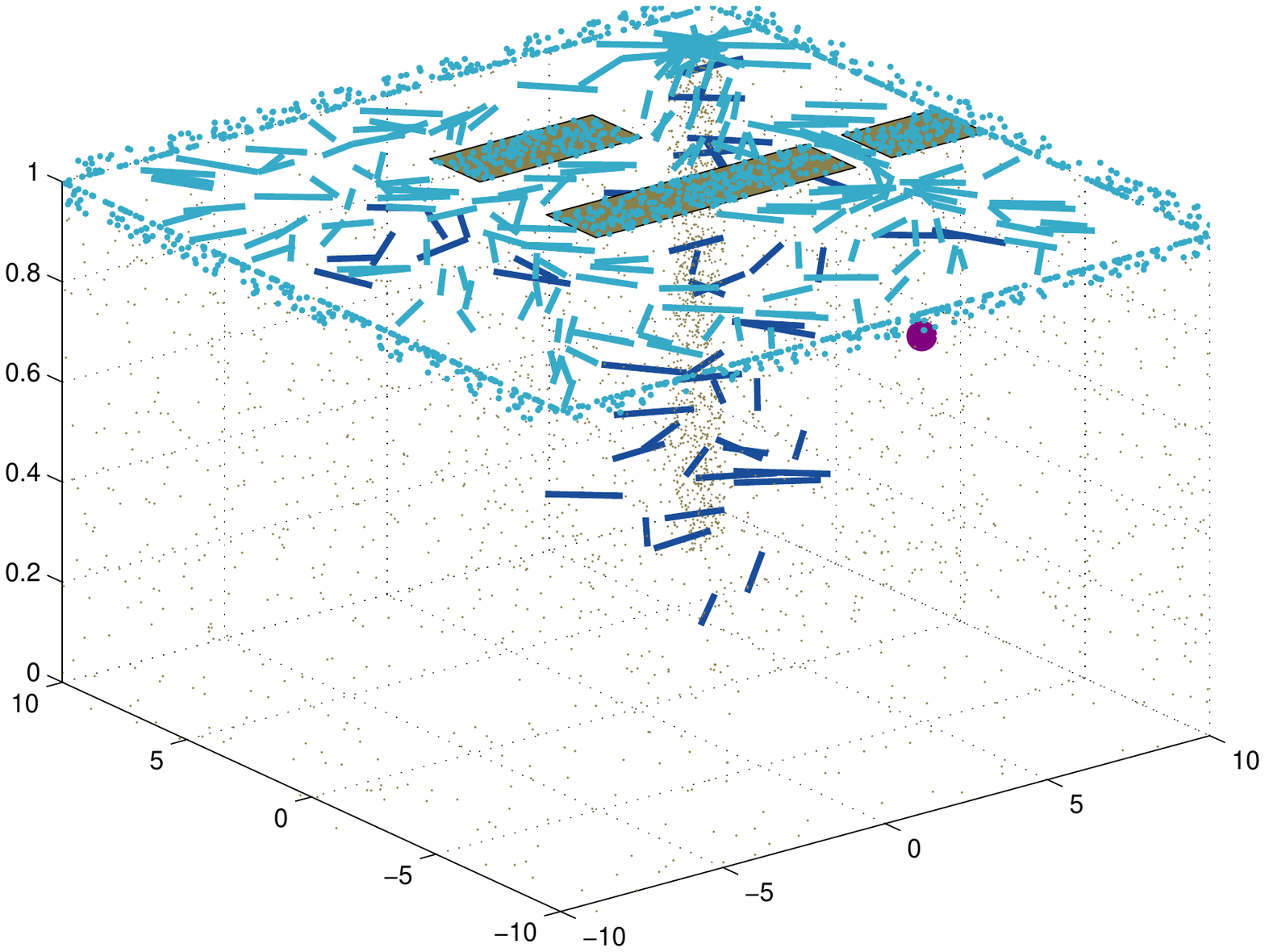}
  }
  \hfill
  \subfigure[Policy on $\overline{\mathcal{M}}_{3000}$]{
  \label{figsubConstruct8}
  \includegraphics[width=47mm,bb= 78 210 538 582]{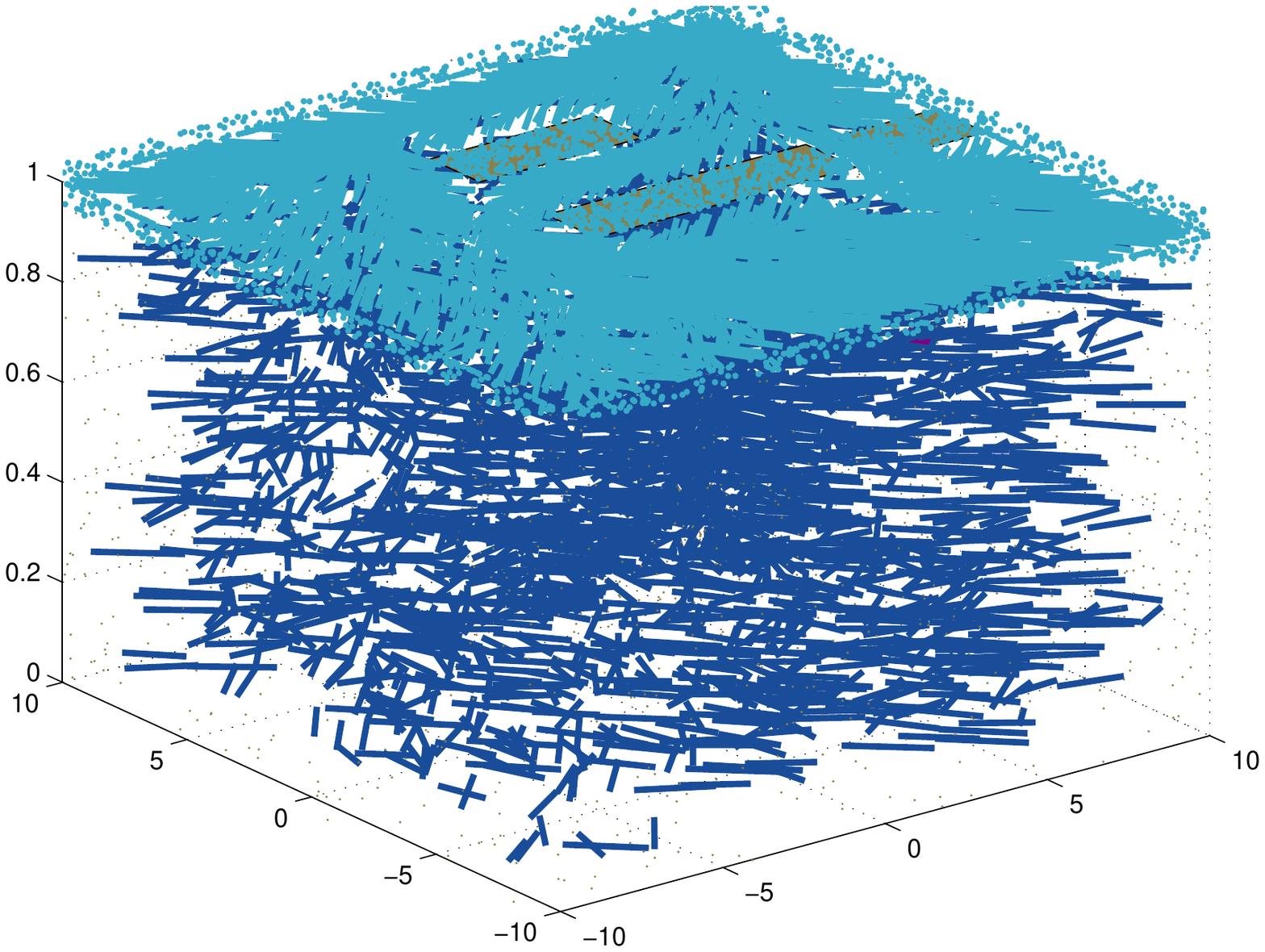}
  }
  \hfill
  \subfigure[Policy on $\overline{\mathcal{M}}_{3000} \backslash \mathcal{M}_{3000}$: Top-down view]{
  \label{figsubConstruct9}
  \includegraphics[width=47mm,bb= 78 210 538 582]{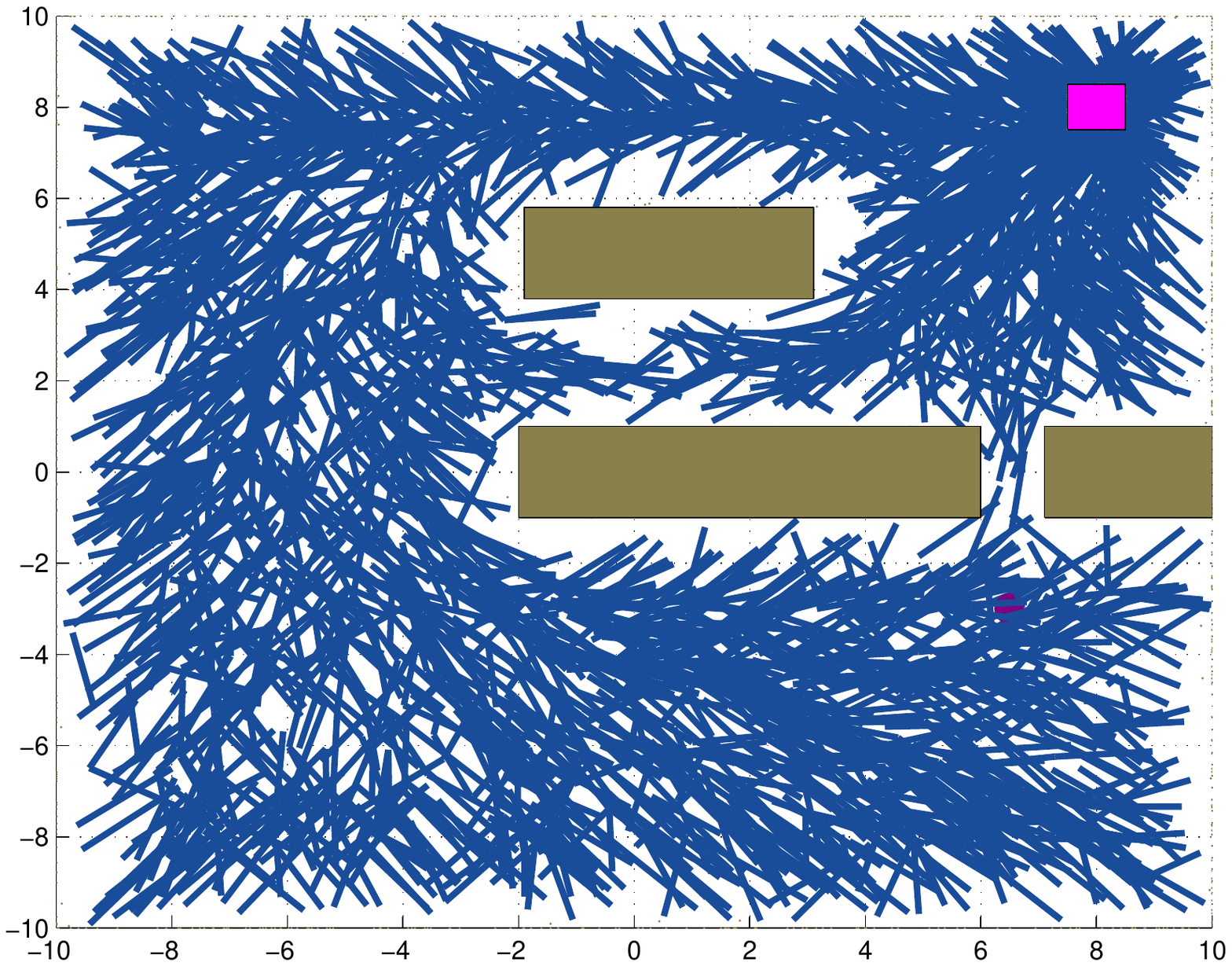}
  }
  \hfill
  \subfigure[Markov chain implied by $\overline{\mathcal{M}}_{200}$.]{
  \label{figsubConstruct10}
  \includegraphics[width=47mm, bb= 140 60 1100 860]{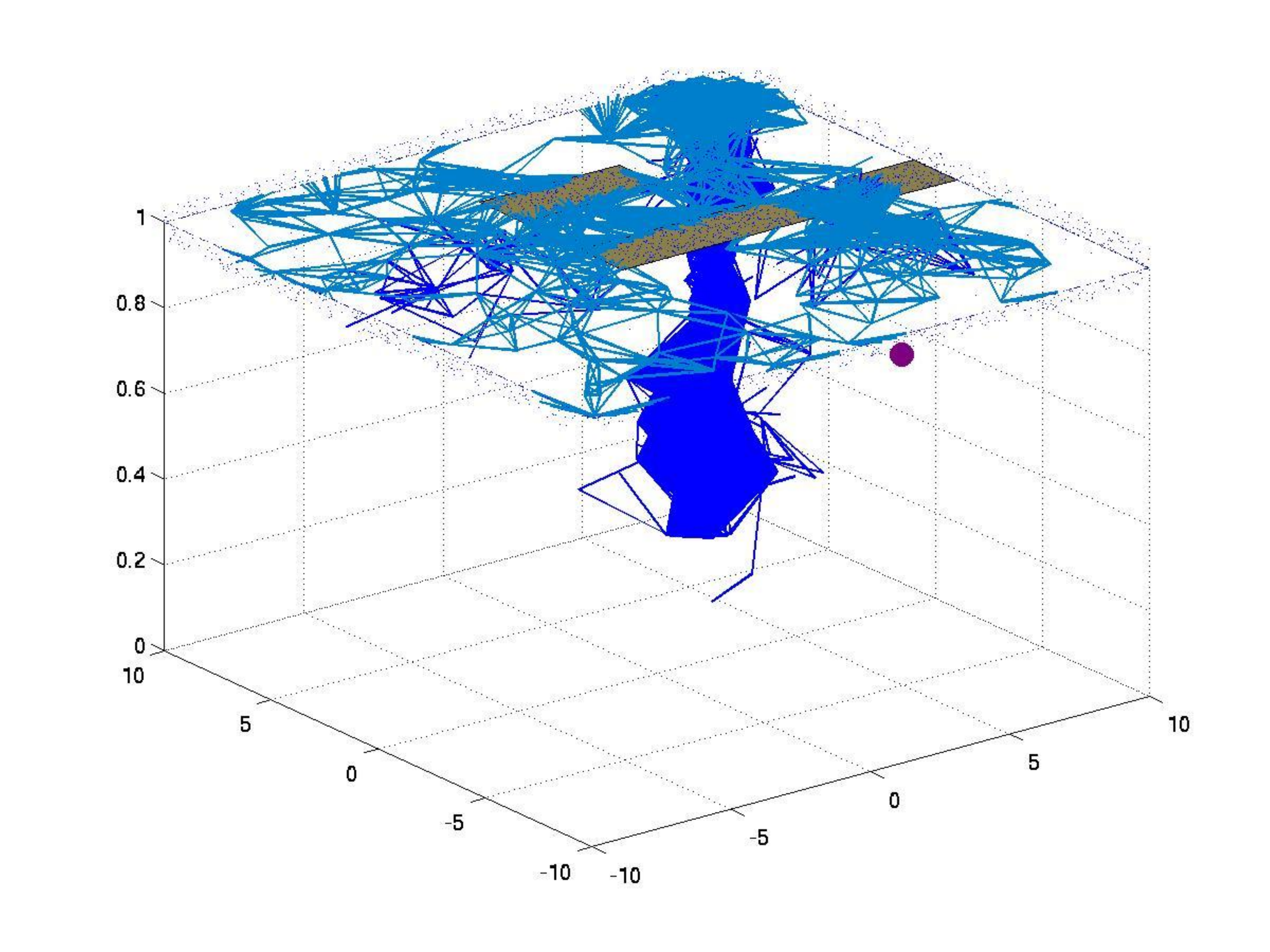}
  }
  \hfill
  \subfigure[Markov chain implied by $\overline{\mathcal{M}}_{500}$.]{
  \label{figsubConstruct11}
  \includegraphics[width=47mm,bb= 140 60 1100 860]{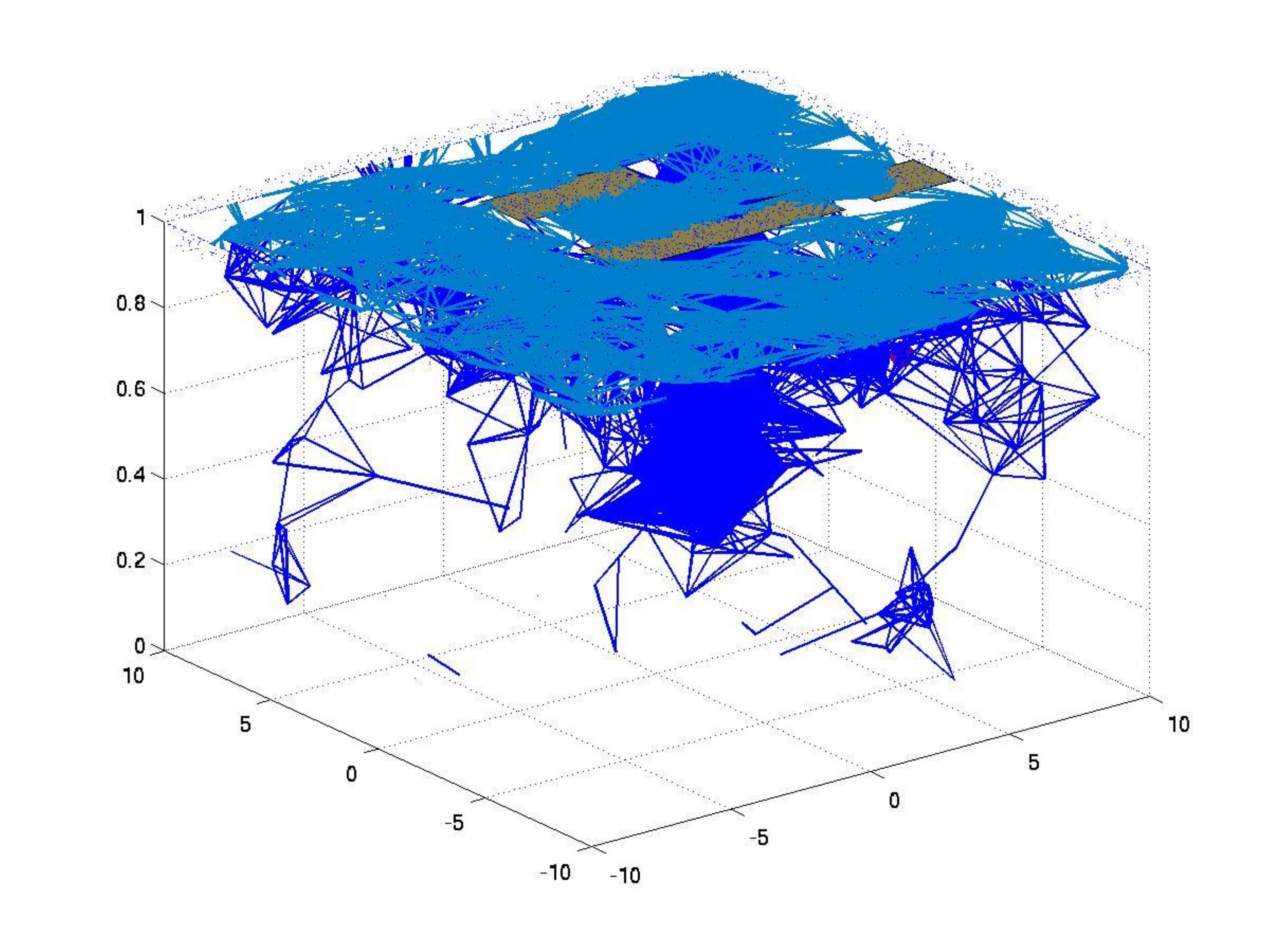}
  }
  \hfill
  \subfigure[Markov chain implied by $\overline{\mathcal{M}}_{1000}$.]{
  \label{figsubConstruct12}
  \includegraphics[width=47mm,bb= 140 60 1100 860]{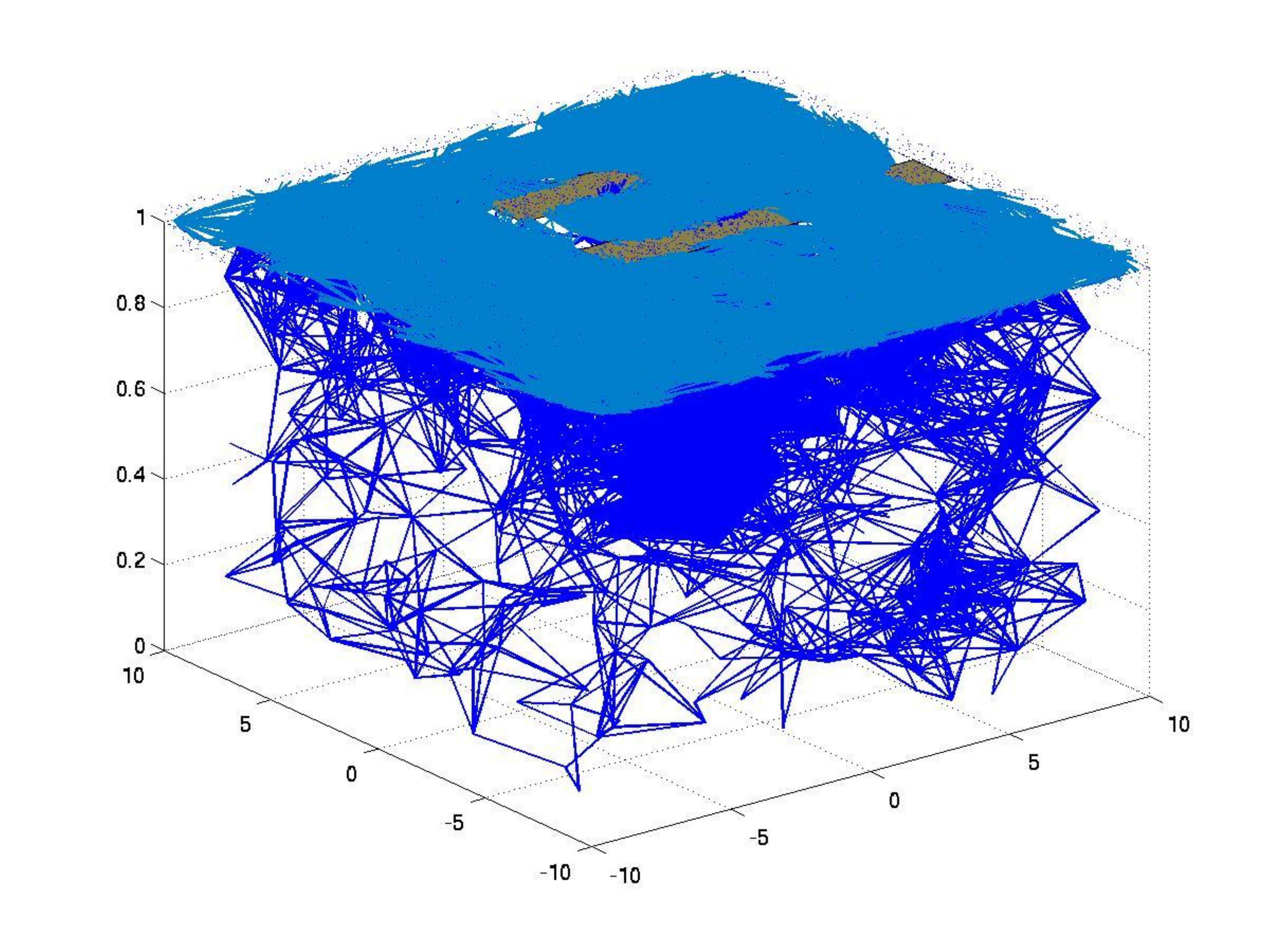}
  }
  \caption{Figures~\ref{figsubConstruct7}-\ref{figsubConstruct9} and Figures~\ref{figsubConstruct10}-\ref{figsubConstruct12} show the corresponding anytime policies and the associated Markov chains on $\overline{\mathcal{M}}_{n}$ respectively. In Fig.~\ref{figsubConstruct9}, we show the top-down view of a policy for states in $\overline{\mathcal{M}}_{3000} \backslash \mathcal{M}_{3000}$. We observe that the system will try to avoid the narrow corridors when the risk tolerance is low. We can also observe that the structures of the Markov chains quickly cover the state spaces $S$ and $\overline{S}$ with connected random graphs.}
  \label{figConstruct1}
\end{center}
\end{figure*}

\begin{figure*}[!]
\begin{center}
  \subfigure[Value function $J_{200,1.0}$.]{
  \label{figsubValue1}
  \includegraphics[width=47mm,bb= 78 210 538 582]{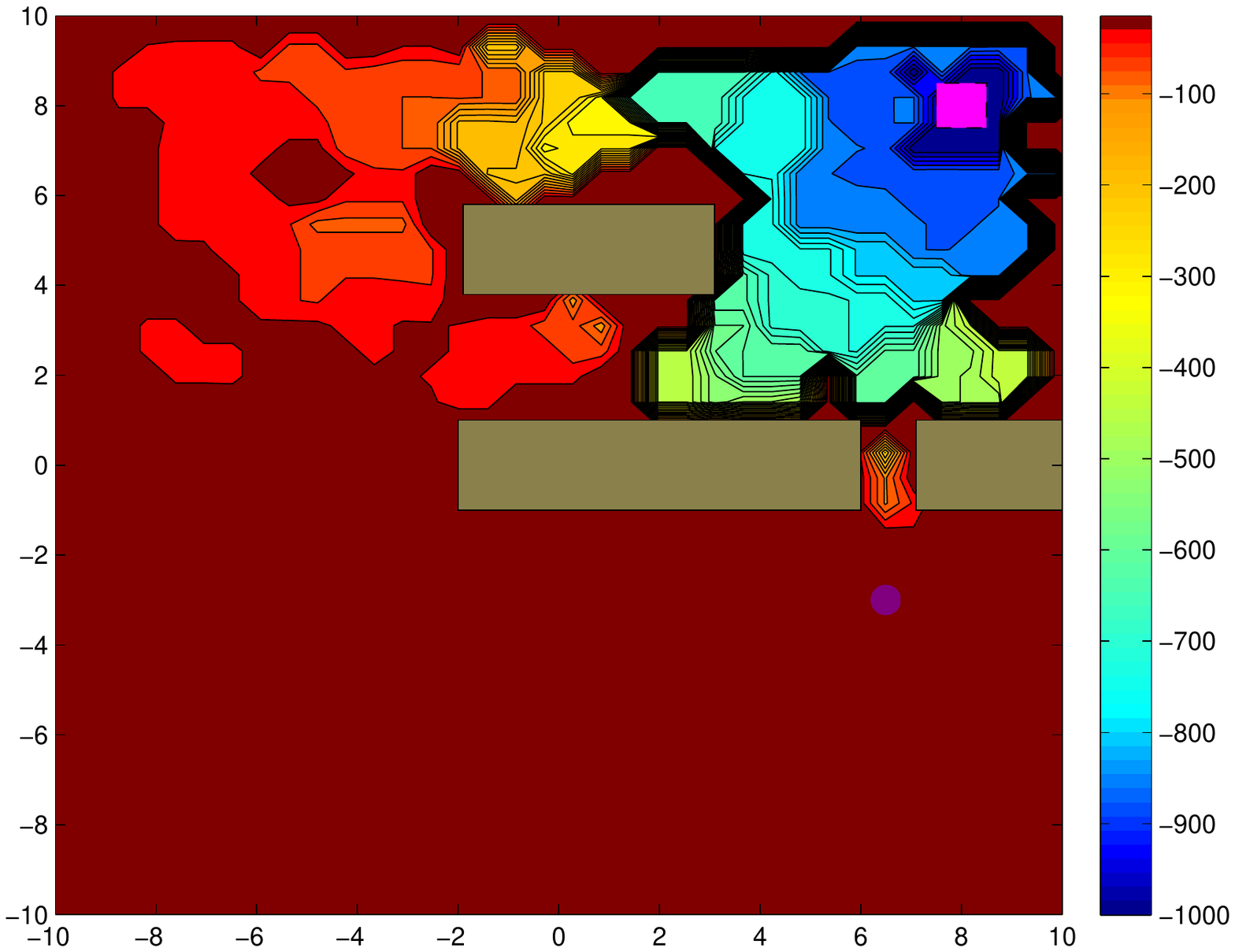}
  }
  \hfill
  \subfigure[Value function $J_{2000,1.0}$.]{
  \label{figsubValue2}
  \includegraphics[width=47mm,bb= 78 210 538 582]{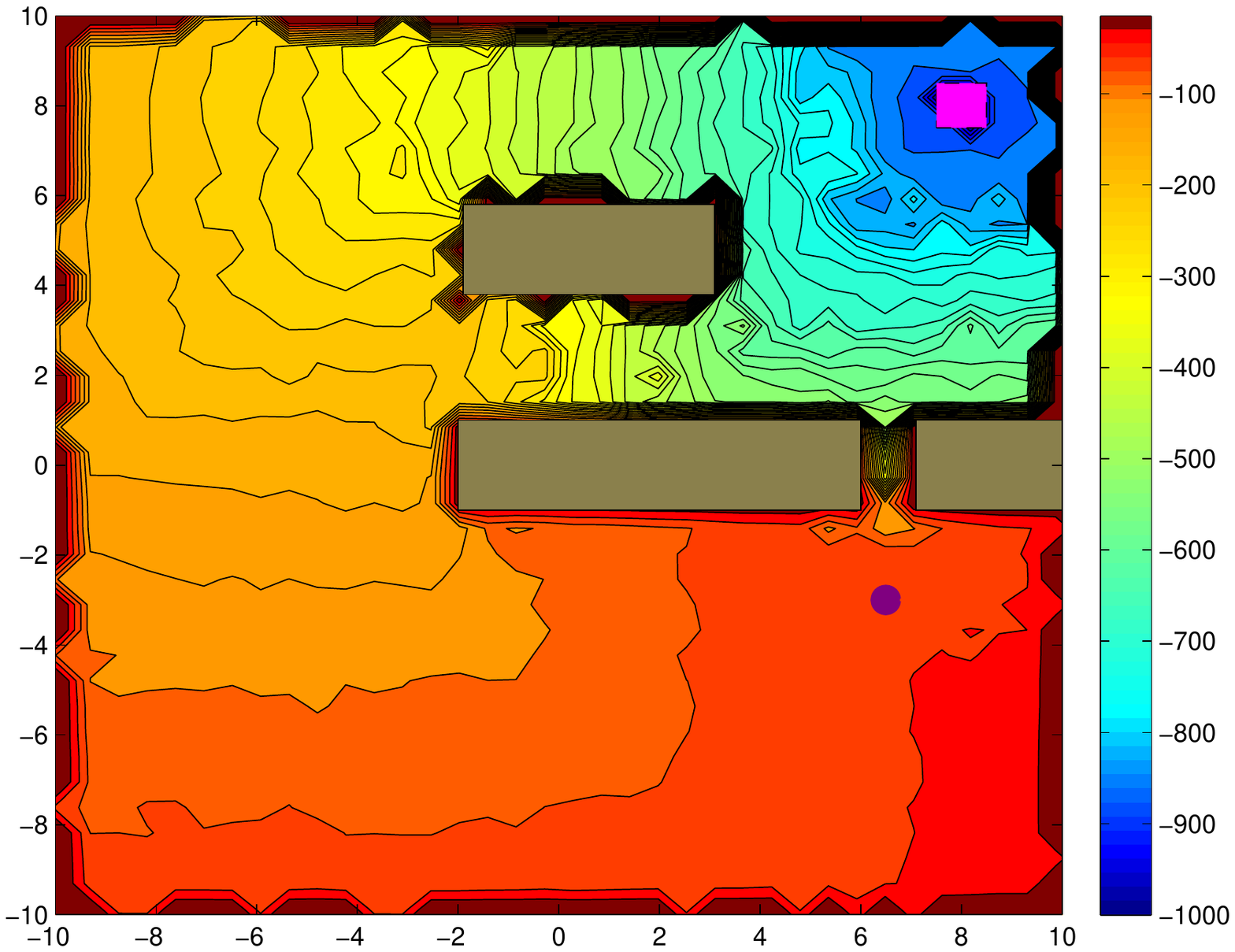}
  }
  \hfill
  \subfigure[Value function $J_{4000,1.0}$.]{
  \label{figsubValue3}
  \includegraphics[width=47mm,bb= 78 210 538 582]{value_0003999.pdf}
  }
  \hfill
  \subfigure[Value function $J_{4000,0.1}$]{
  \label{figsubValue7}
  \includegraphics[width=47mm,bb= 78 210 538 582]{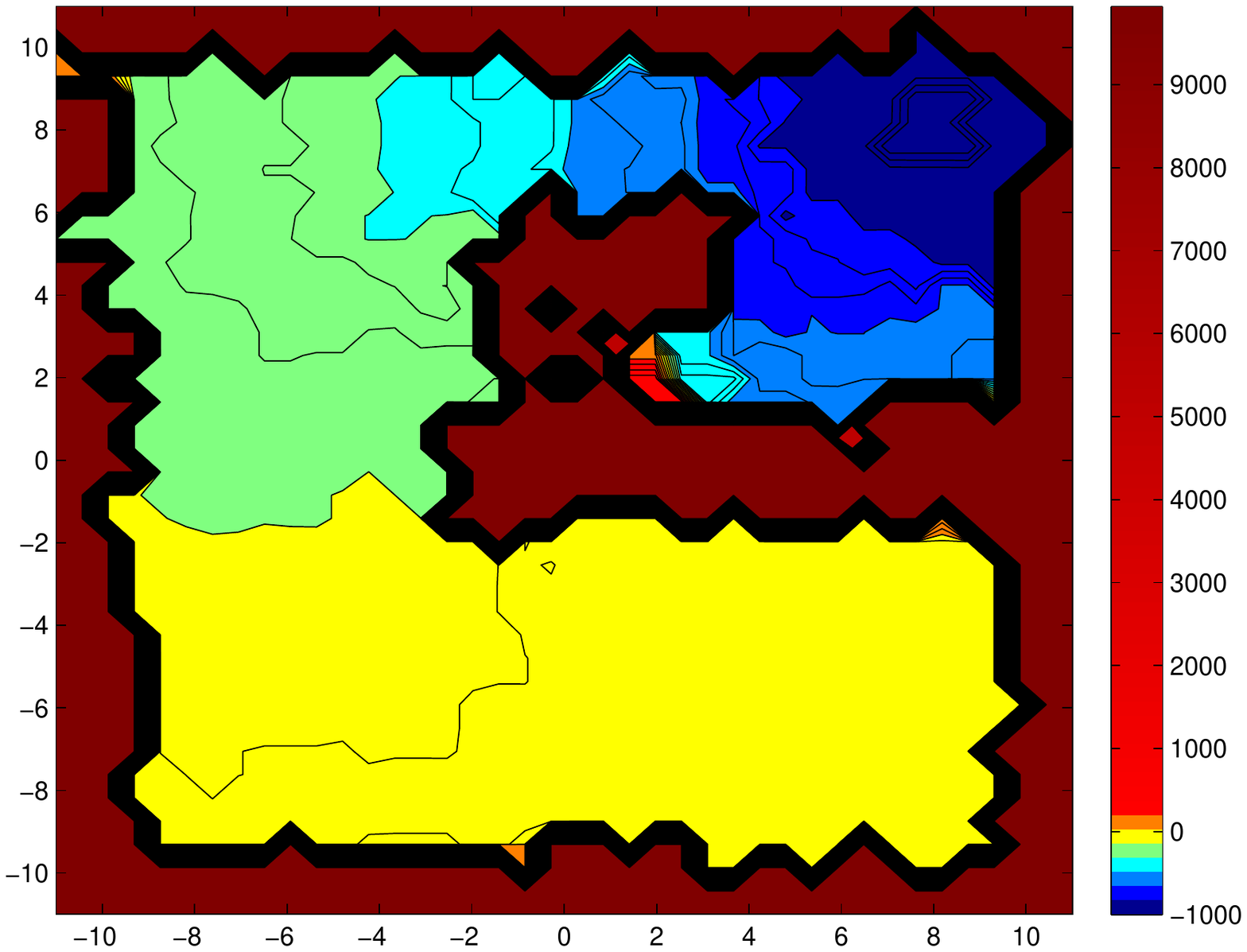}
  }
  \hfill
  \subfigure[Value function $J_{4000,0.5}$]{
  \label{figsubValue8}
  \includegraphics[width=47mm,bb= 78 210 538 582]{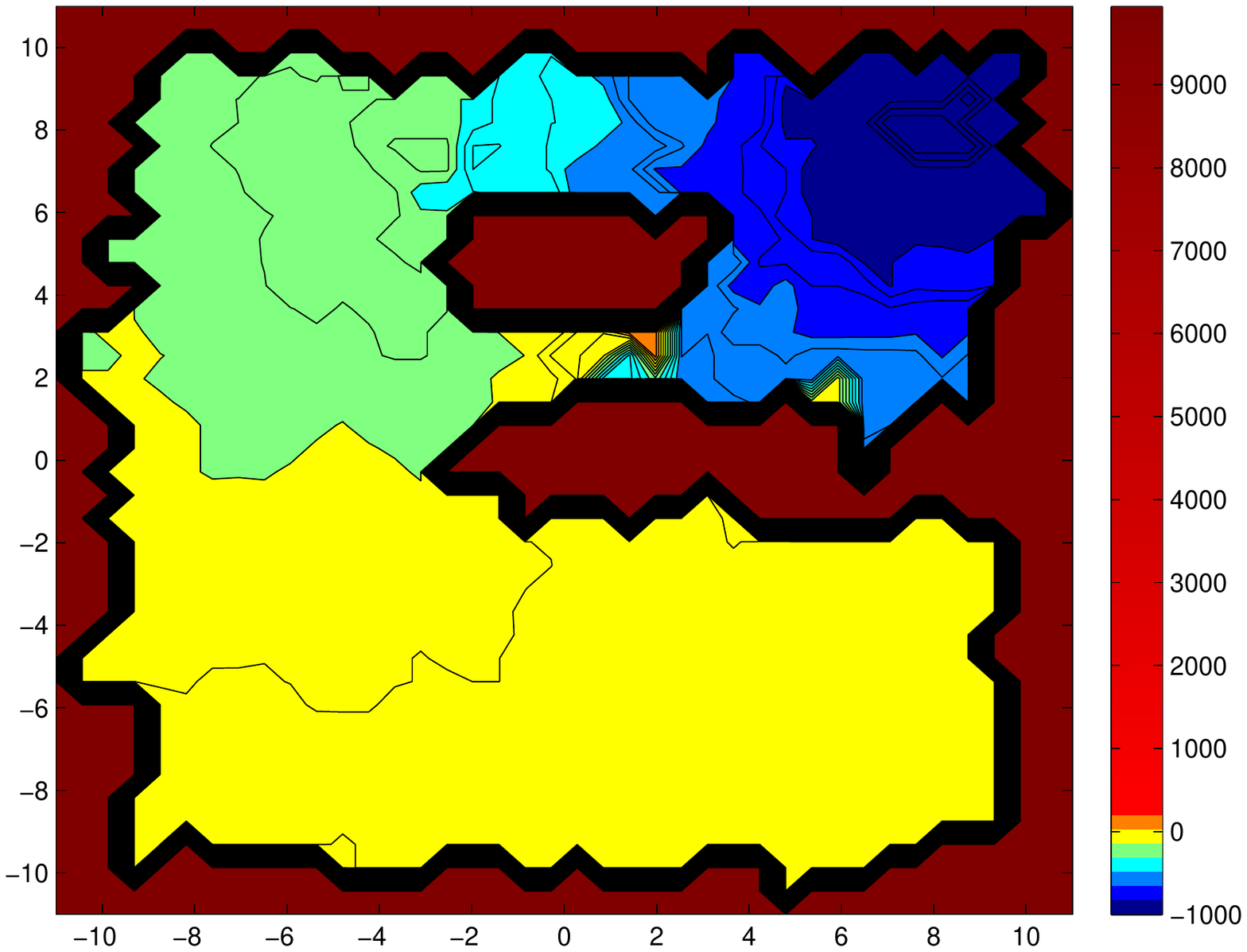}
  }
  \hfill
  \subfigure[Value function $J_{4000,0.9}$]{
  \label{figsubValue9}
  \includegraphics[width=47mm,bb= 78 210 538 582]{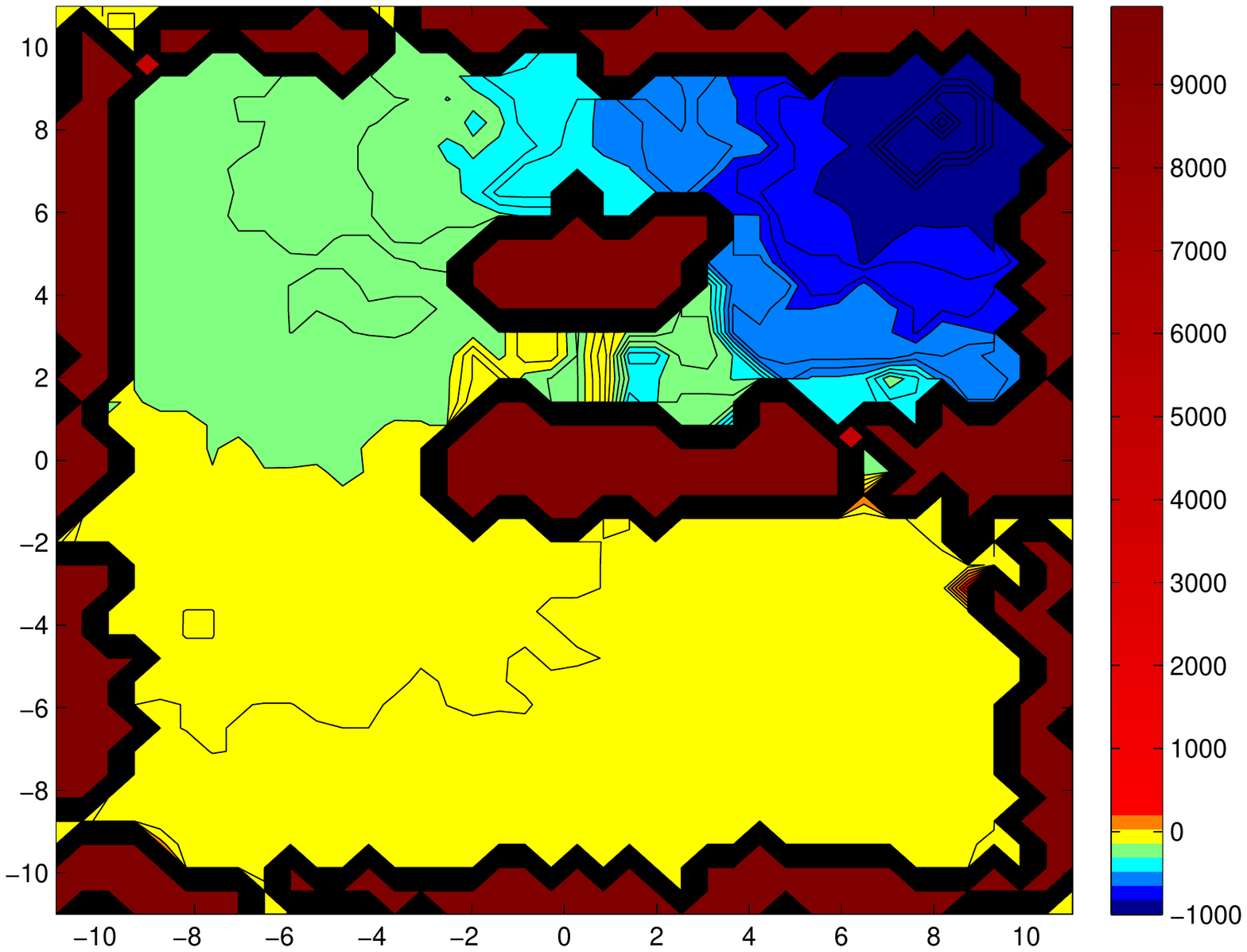}
  }
  \caption{Examples of incremental value functions over iterations. Figure~\ref{figsubValue1}-\ref{figsubValue3} show the approximate cost-to-go functions $J_n$ when the probability threshold $\eta_0$ is 1.0 for $n=200$, $2000$ and $4000$. Figures~\ref{figsubValue7}-\ref{figsubValue9} present the approximate cost-to-go function $J_{4000}$ in $\overline{\mathcal{M}}_{4000}$ for augmented states where their martingale components are $0.1$, $0.5$ and $0.9$ respectively. The plot shows that the lower the martingale state is, the higher the cost value is -- which is consistent with the characteristics in Section~\ref{subsection:characterization}.
} 
  \label{figValue}  
\end{center}
\end{figure*}

\begin{figure*}[t]
\begin{center}
  \subfigure[Unconstrained problem trajectories: simulated collision probability $49.27\%$, average cost $-125.20$.]{
  \label{figsubTrajNormal}
  \includegraphics[width=70mm,bb= 78 210 538 582]{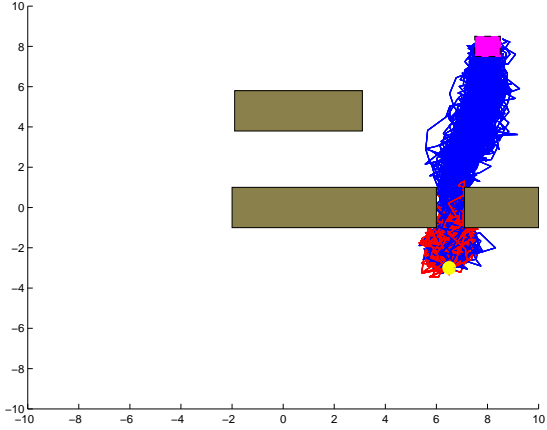}
  }
  \ \ \ \ \ \ \ \ \ \ \ \ \ \ \ \ \ \ \	
  \subfigure[Min-collision trajectories: simulated collision probability $0\%$, average cost $-17.85$.]{
  \label{figsubTrajMinProb}
  \includegraphics[width=70mm,bb= 78 210 538 582]{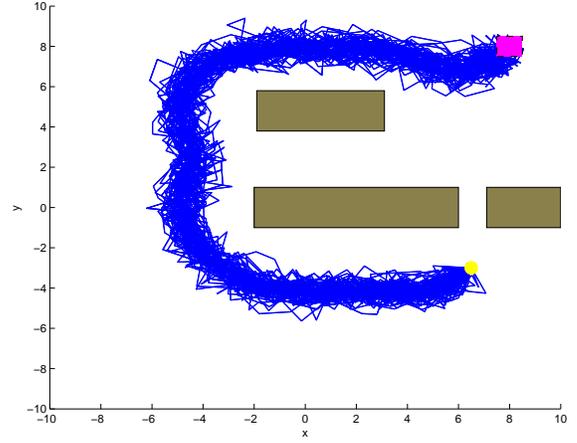}
  }
  \caption{Examples of trajectories from policies of the unconstrained problem (Fig.~\ref{figsubTrajNormal}) and the min-collision probability problem (Fig.~\ref{figsubTrajMinProb}). In the unconstrained problem, the system takes risk to go through one of the narrow corridors to reach the goal. In contrast, in the min-collision probability problem, the system detours around the obstacles to reach the goal. While there are about $49.27\%$ of $2000$ trajectories (plotted in red) that collide with the obstacles for the former, we observe no collision out of $2000$ trajectories for the latter.}
  \label{figUnMinTrajs}
\end{center}
\end{figure*}

\begin{figure*}[t]
  \begin{center}
  \subfigure[An example of controlled trajectories using boundary values.]{
  \label{figImerse}
  \includegraphics[width=85mm,bb= 115 240 500 530]{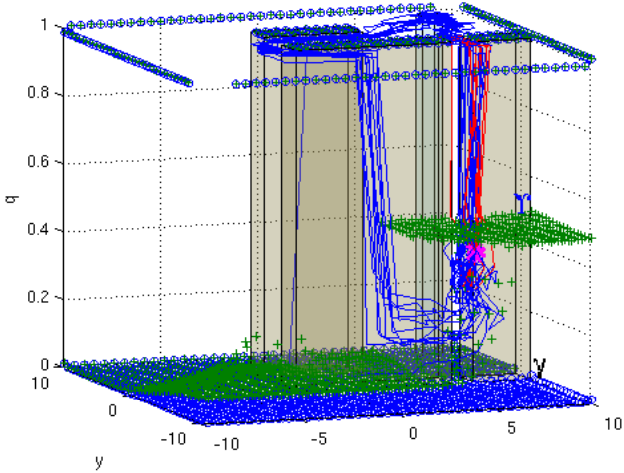}
  }
  \hfill
  \subfigure[Failure ratios for the first $N$ trajectories ($N \leq 5000$) with different $\eta$.]{
  \label{figFailureRatio}
  \includegraphics[width=85mm,bb= 105 248 500 492]{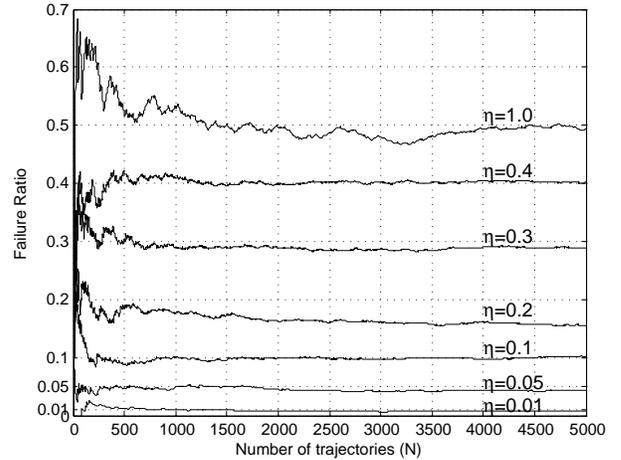}
  }
  \caption{In Fig.~\ref{figImerse}, we show an example of controlled trajectories using boundary values. The system starts from $(6.5,-3)$ with the failure-probability threshold $\eta=0.4$. The martingale state varies along controlled trajectories as a random parameter in a randomized control policy. When the martingale state is above $\Upsilon$, the system follows a deterministic control policy obtained from the unconstrained problem. In Fig.~\ref{figFailureRatio},we show failure ratios for the first $N$ trajectories ($1\leq N \leq 5000$) starting from $(6.5,-3)$ with different values of $\eta$. As seen in Fig.~\ref{figFailureRatio}, the algorithm is able to keep the failure ratio in 5000 executions around $0.40$ as dictated by the choice of $\eta=0.40$ at time $0$. Other failure ratios follow very closely the values of $\eta$, which indicates that the iMDP algorithm is able to provide solutions that are probabilistically sound.}
  \end{center}
\end{figure*}

\begin{figure*}[t]
\begin{center}
  \subfigure[Threshold $\eta=0.01$.]{
  \label{figsubTraj1}
  \includegraphics[width=47mm,bb= 140 42 1150 860]{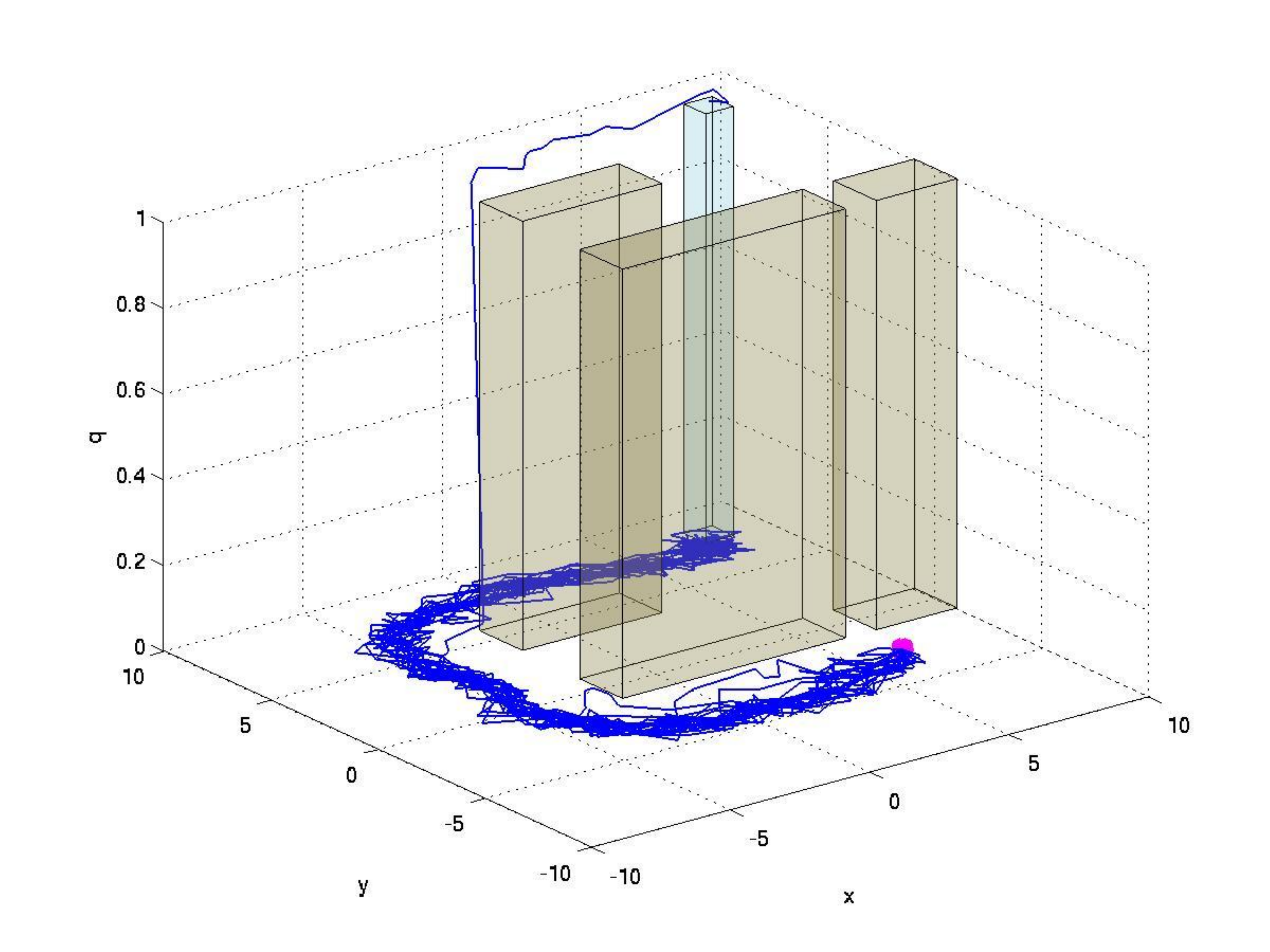}
  }
  \hfill
  \subfigure[Threshold $\eta=0.05$.]{
  \label{figsubTraj2}
  \includegraphics[width=47mm,bb= 140 42 1150 860]{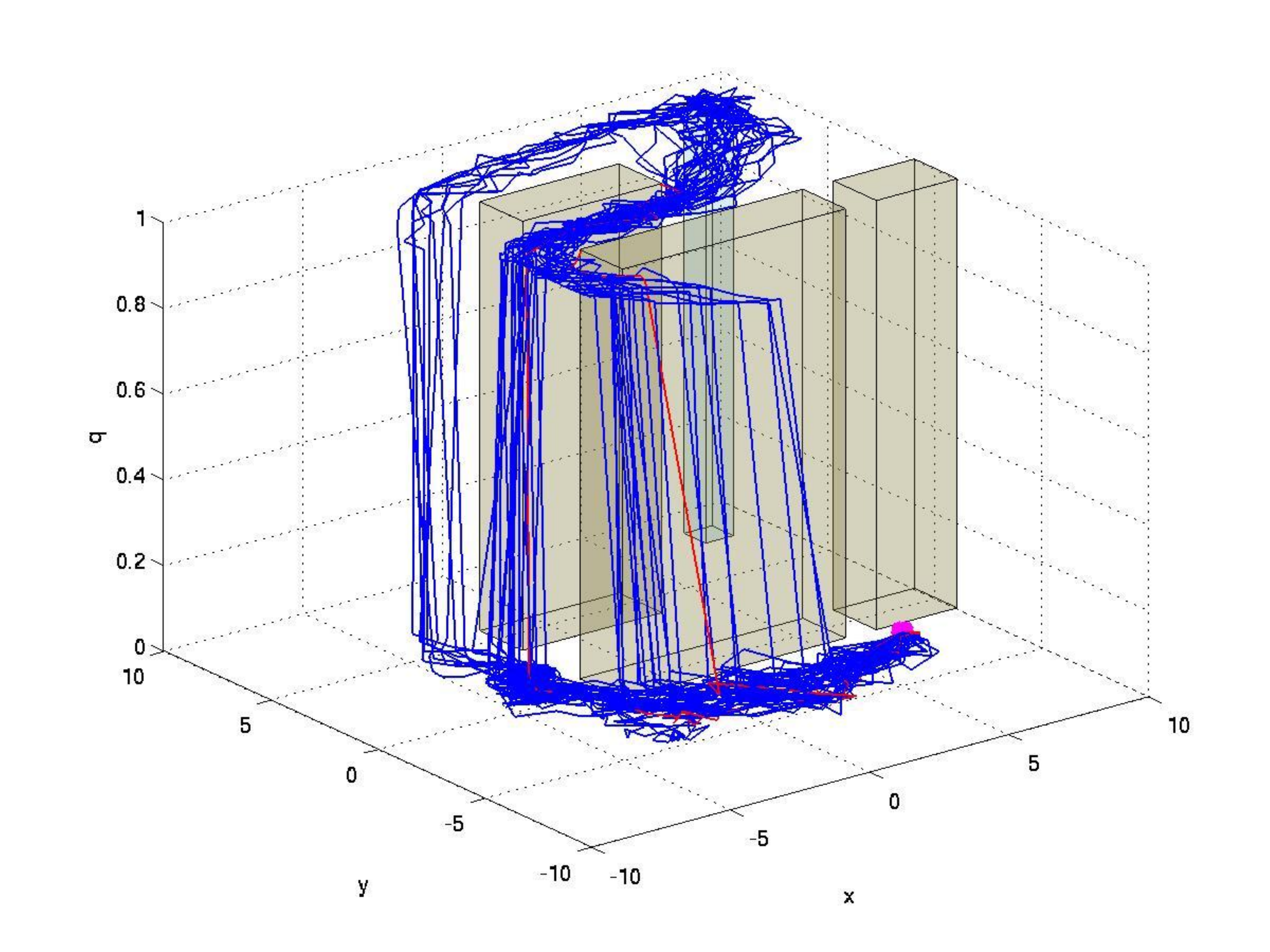}
  }
  \hfill
  \subfigure[Threshold $\eta=0.10$.]{
  \label{figsubTraj3}
  \includegraphics[width=47mm,bb= 140 42 1150 860]{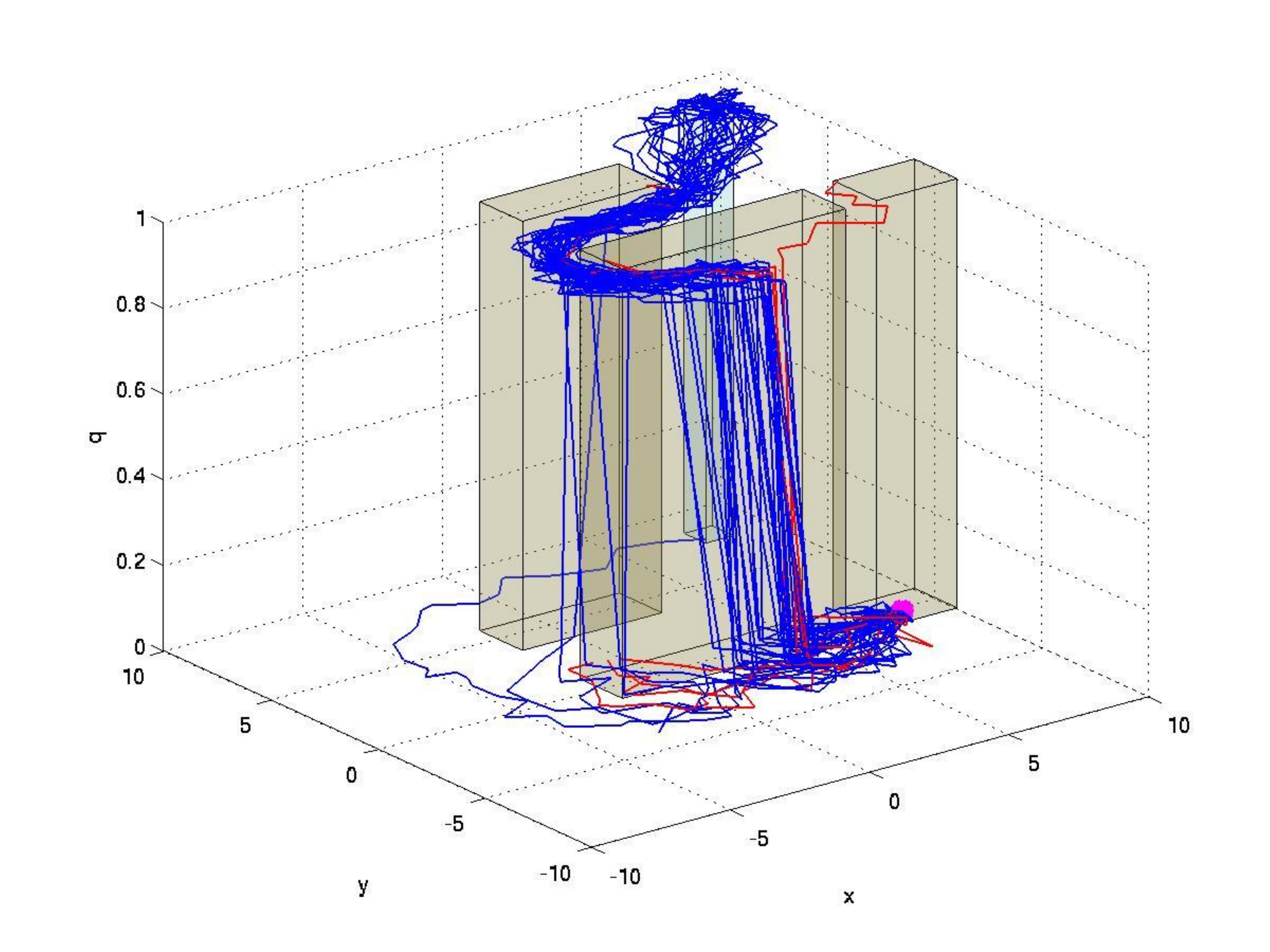}
  }
  \subfigure[$\eta=0.01$: $0.8\%$, $-19.42$.]{
  \label{figsubTraj4}
  \includegraphics[width=47mm,bb=  140 40 1150 860]{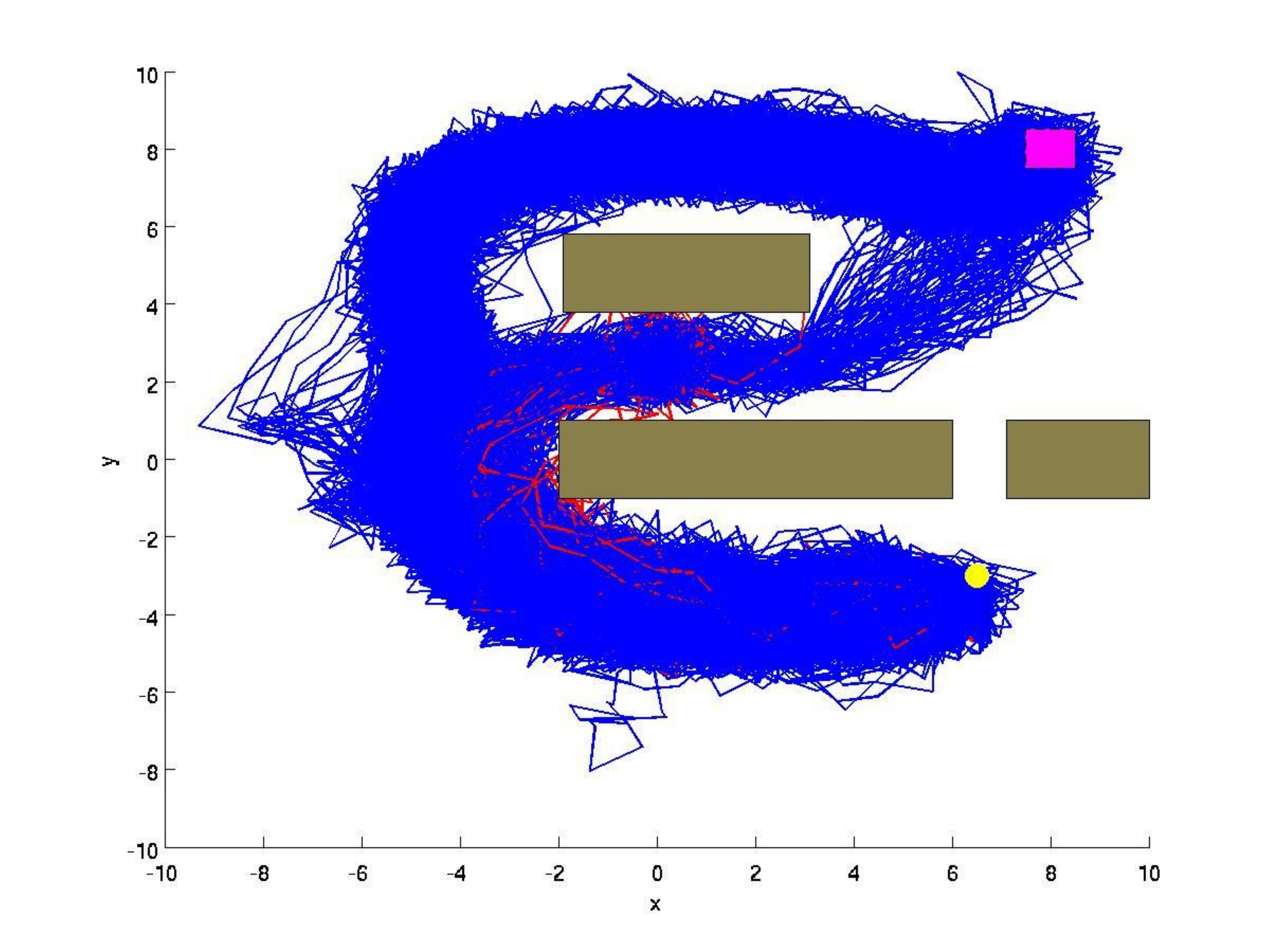}
  }
  \hfill
  \subfigure[$\eta=0.05$: $4.2\%$, $-42.53$.]{
  \label{figsubTraj5}
  \includegraphics[width=47mm,bb=  140 40 1150 860]{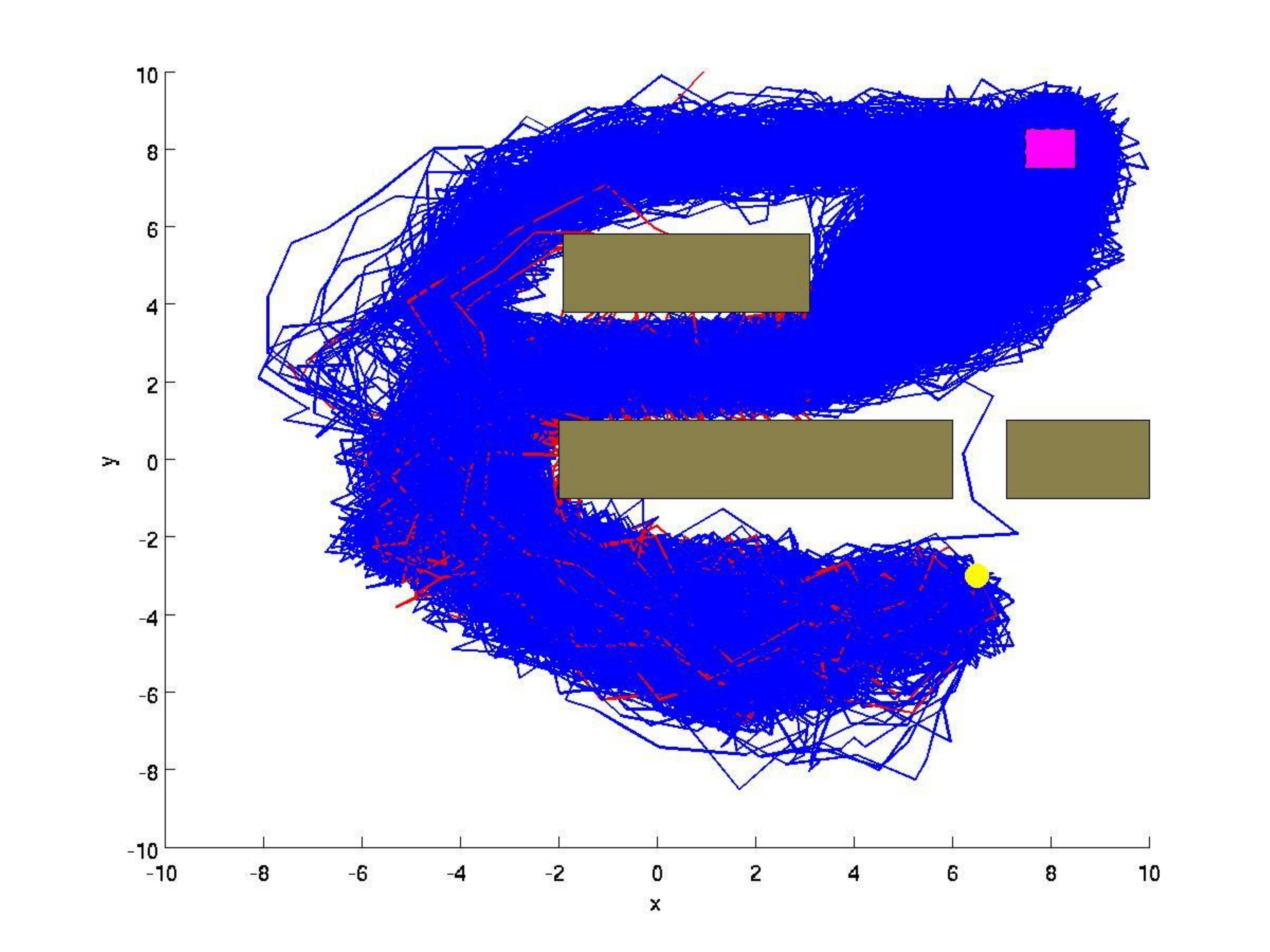}
  }
  \hfill
  \subfigure[$\eta=0.10$: $10\%$, $-58.00$]{
  \label{figsubTraj6}
  \includegraphics[width=47mm,bb=  140 40 1150 860]{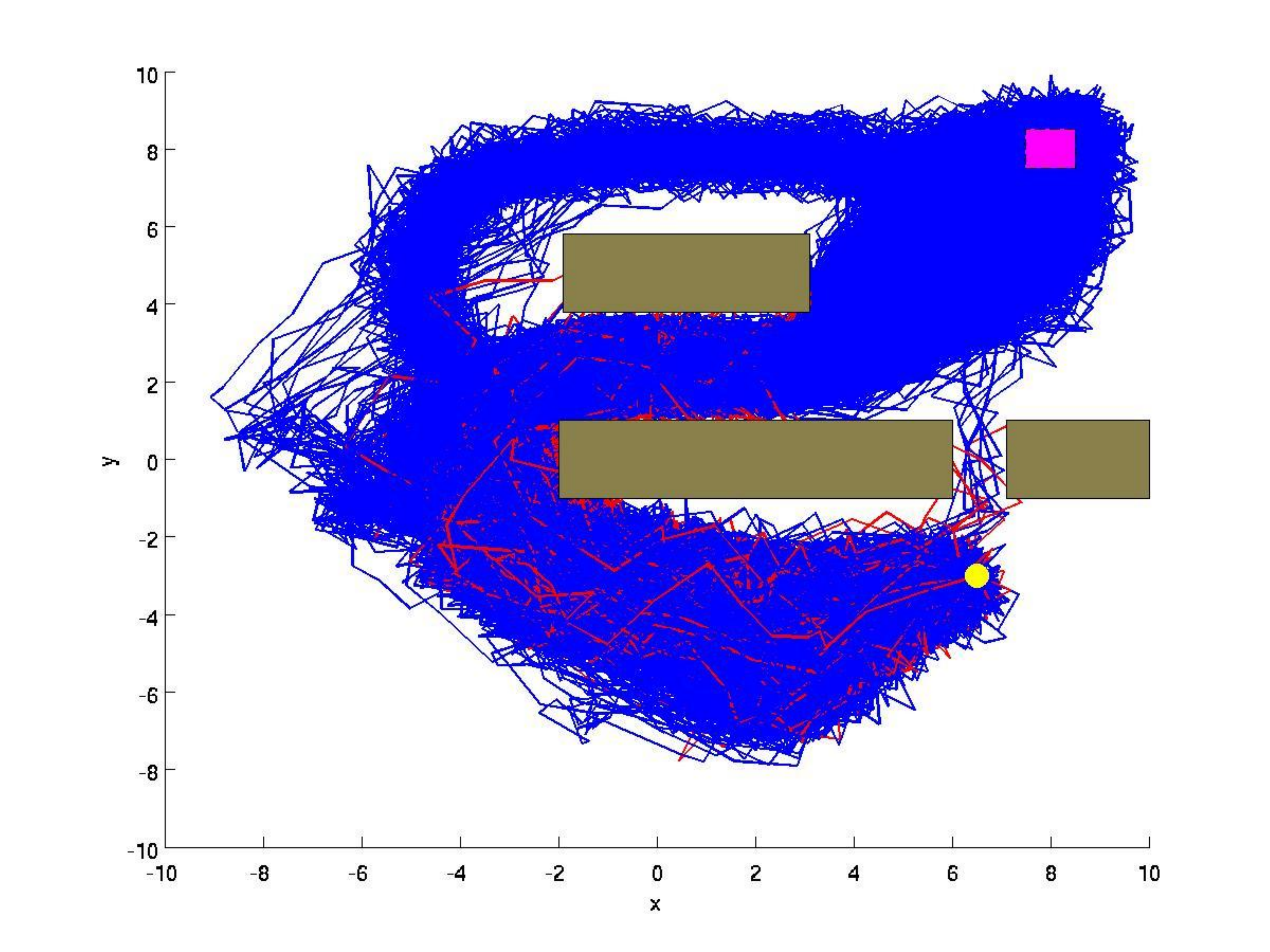}
  }
  \hfill
    \subfigure[Threshold $\eta=0.2$.]{
  \label{figsubTraj7}
  \includegraphics[width=47mm,bb= 140 42 1150 860]{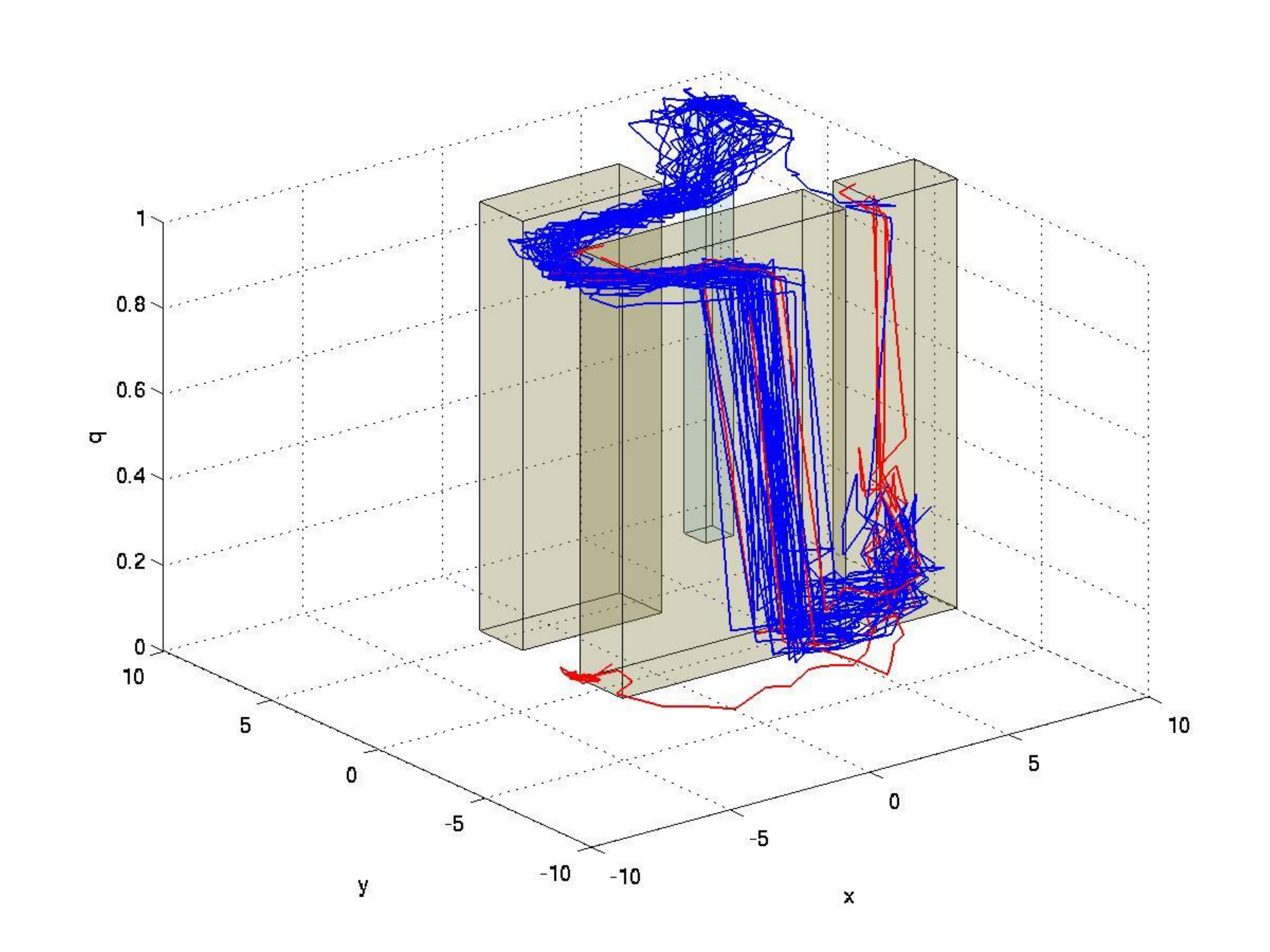}
  }
  \hfill
  \subfigure[Threshold $\eta=0.3$.]{
  \label{figsubTraj8}
  \includegraphics[width=47mm,bb= 140 42 1150 860]{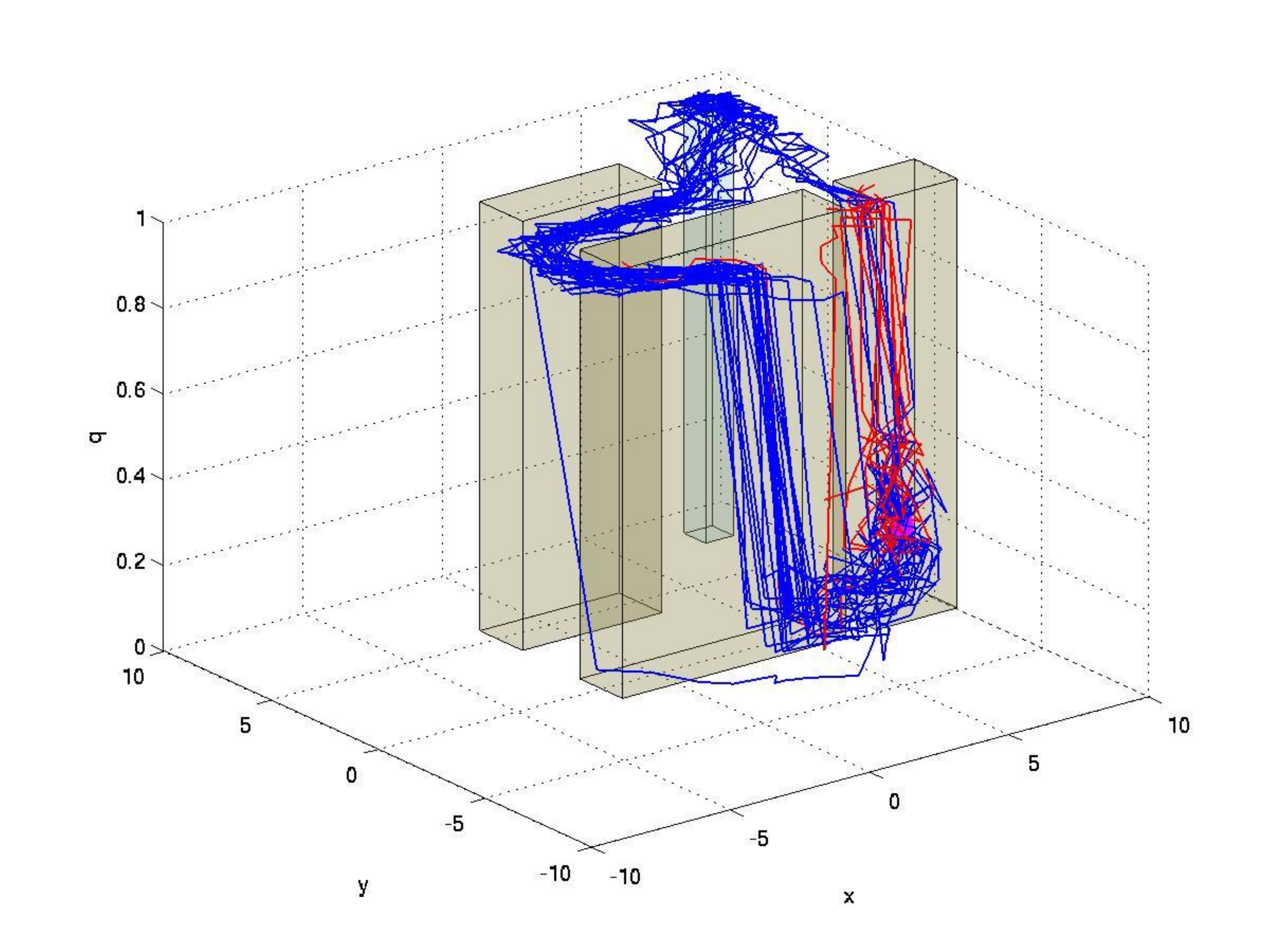}
  }
  \hfill
  \subfigure[Threshold $\eta=0.4$.]{
  \label{figsubTraj9}
  \includegraphics[width=47mm,bb= 140 42 1150 860]{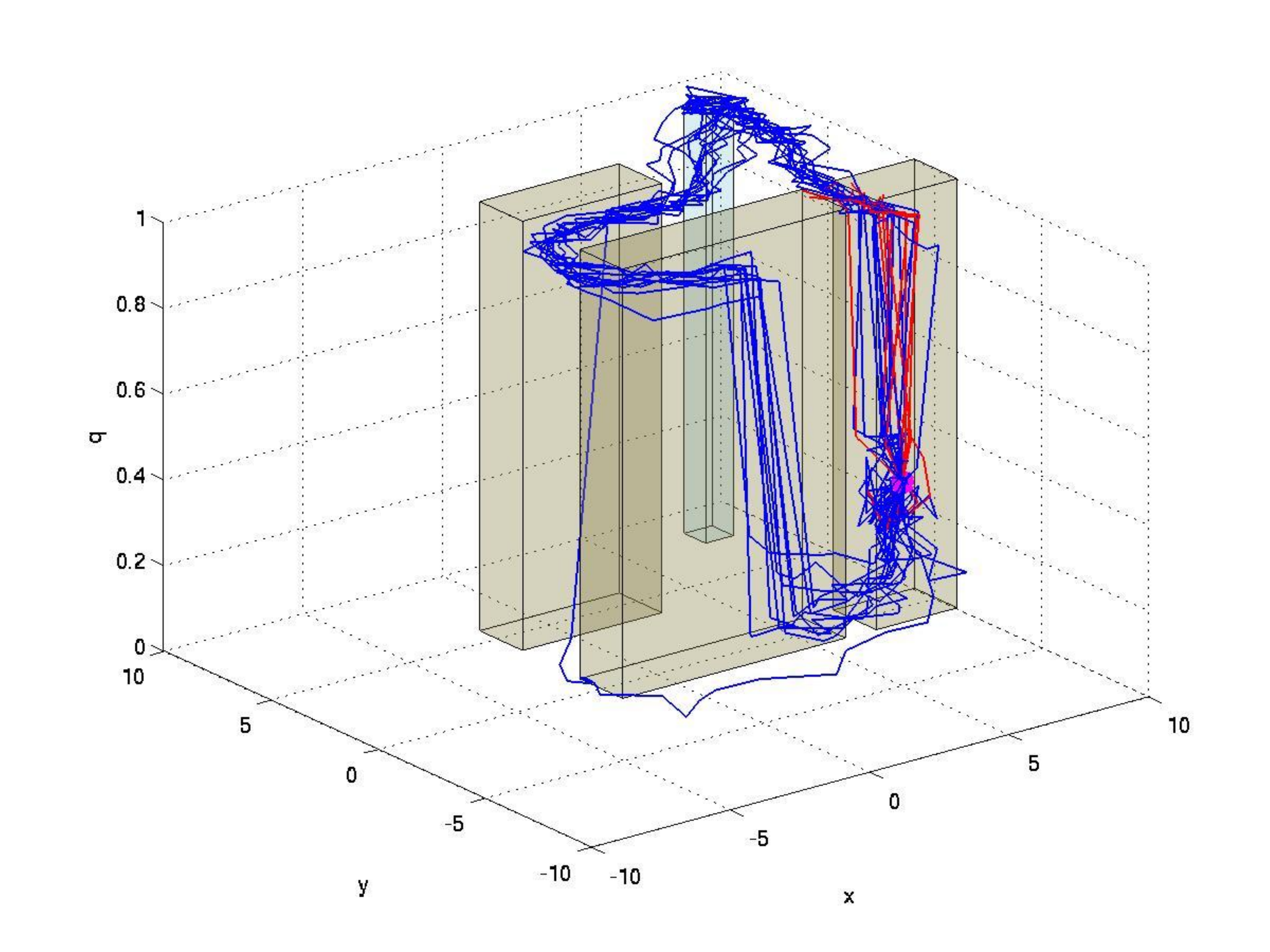}
  }
  \subfigure[$\eta=0.2$: $15.6\%$, $-65.81$.]{
  \label{figsubTraj10}
  \includegraphics[width=47mm,bb=  140 40 1150 860]{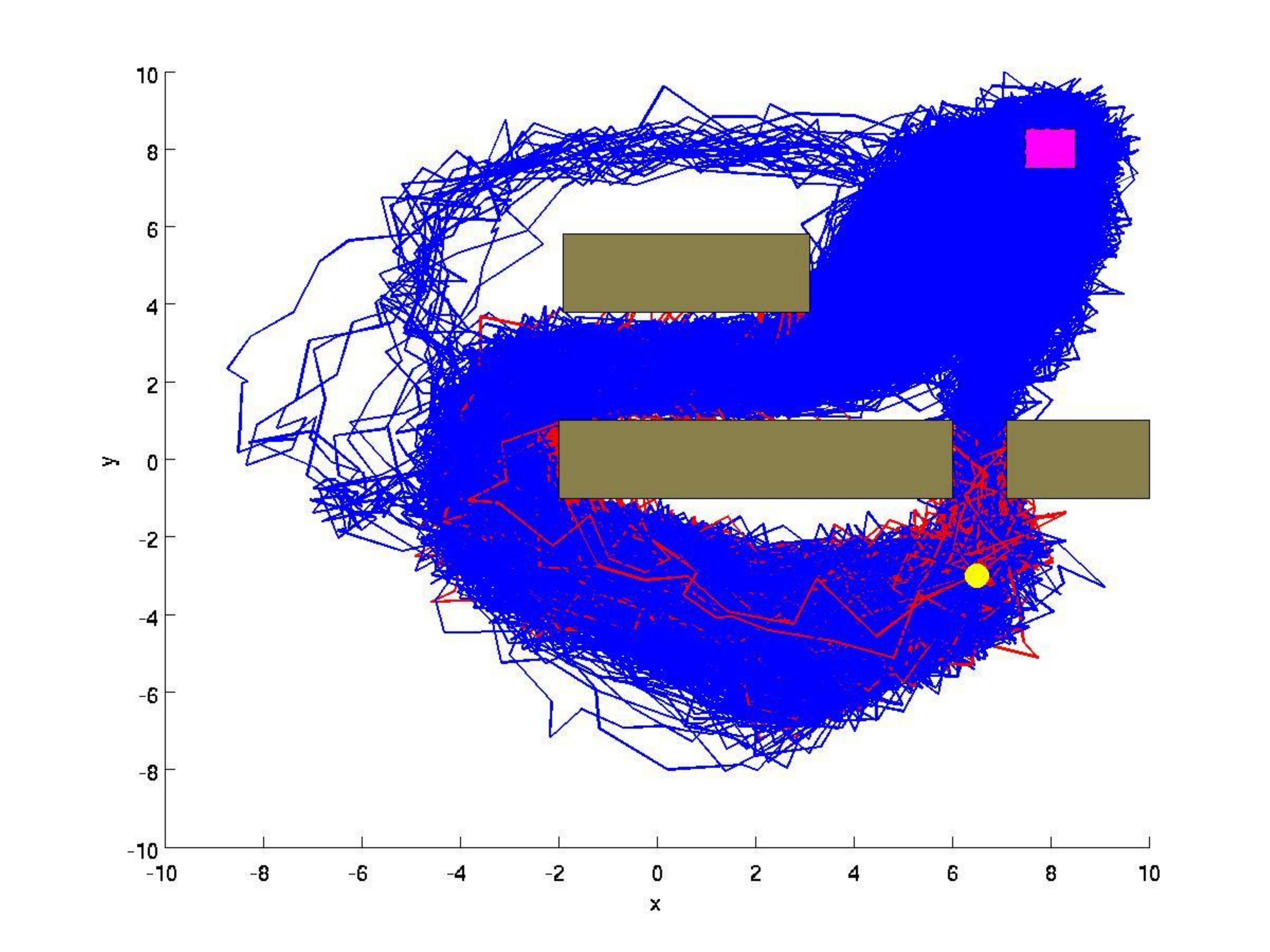}
  }
  \hfill
  \subfigure[$\eta=0.3$: $28.19\%$, $-76.80$.]{
  \label{figsubTraj11}
  \includegraphics[width=47mm,bb=  140 40 1150 860]{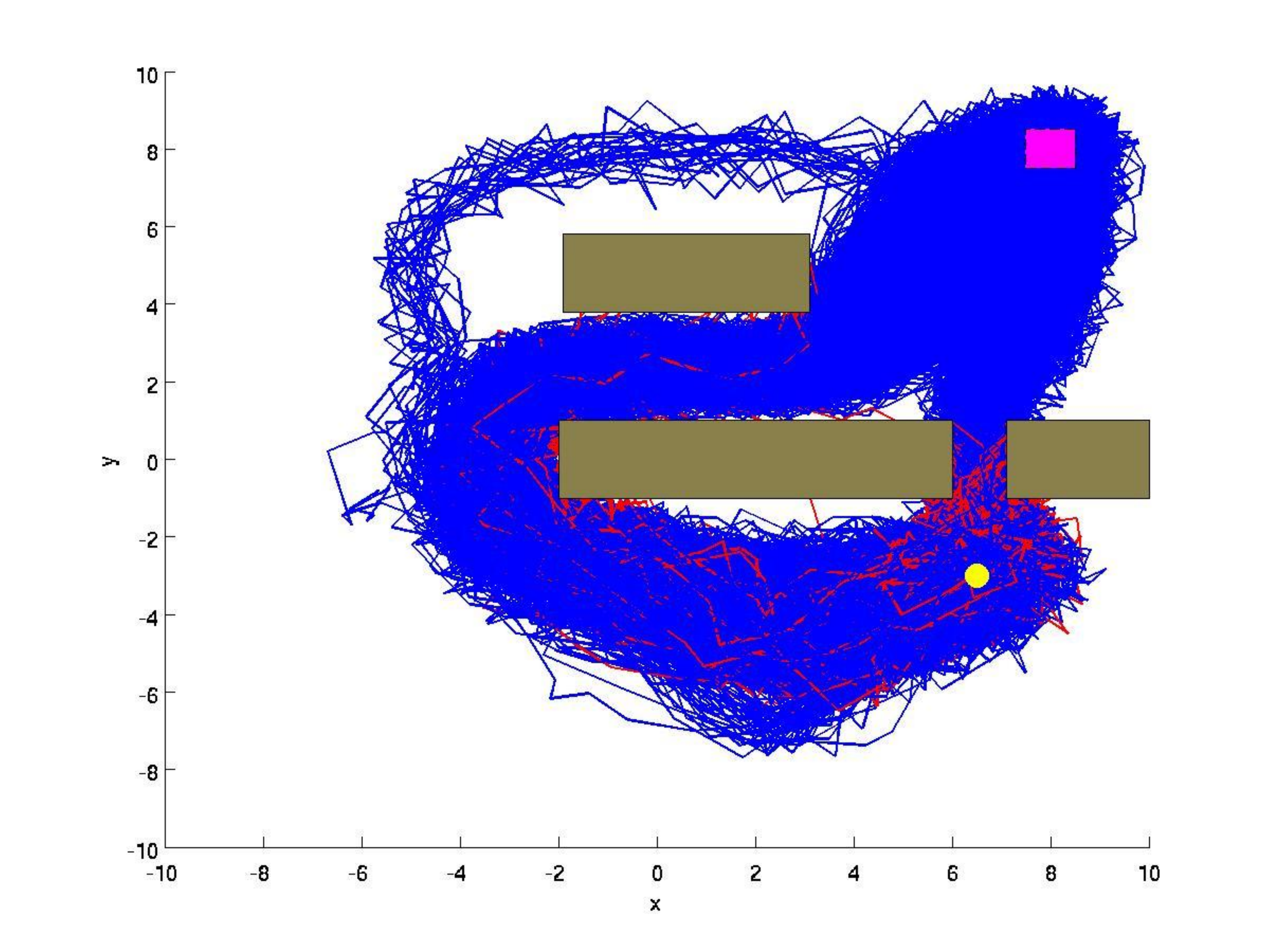}
  }
  \hfill
  \subfigure[$\eta=0.4$: $40\%$, $-115.59$.]{
  \label{figsubTraj12}
  \includegraphics[width=47mm,bb=  140 40 1150 860]{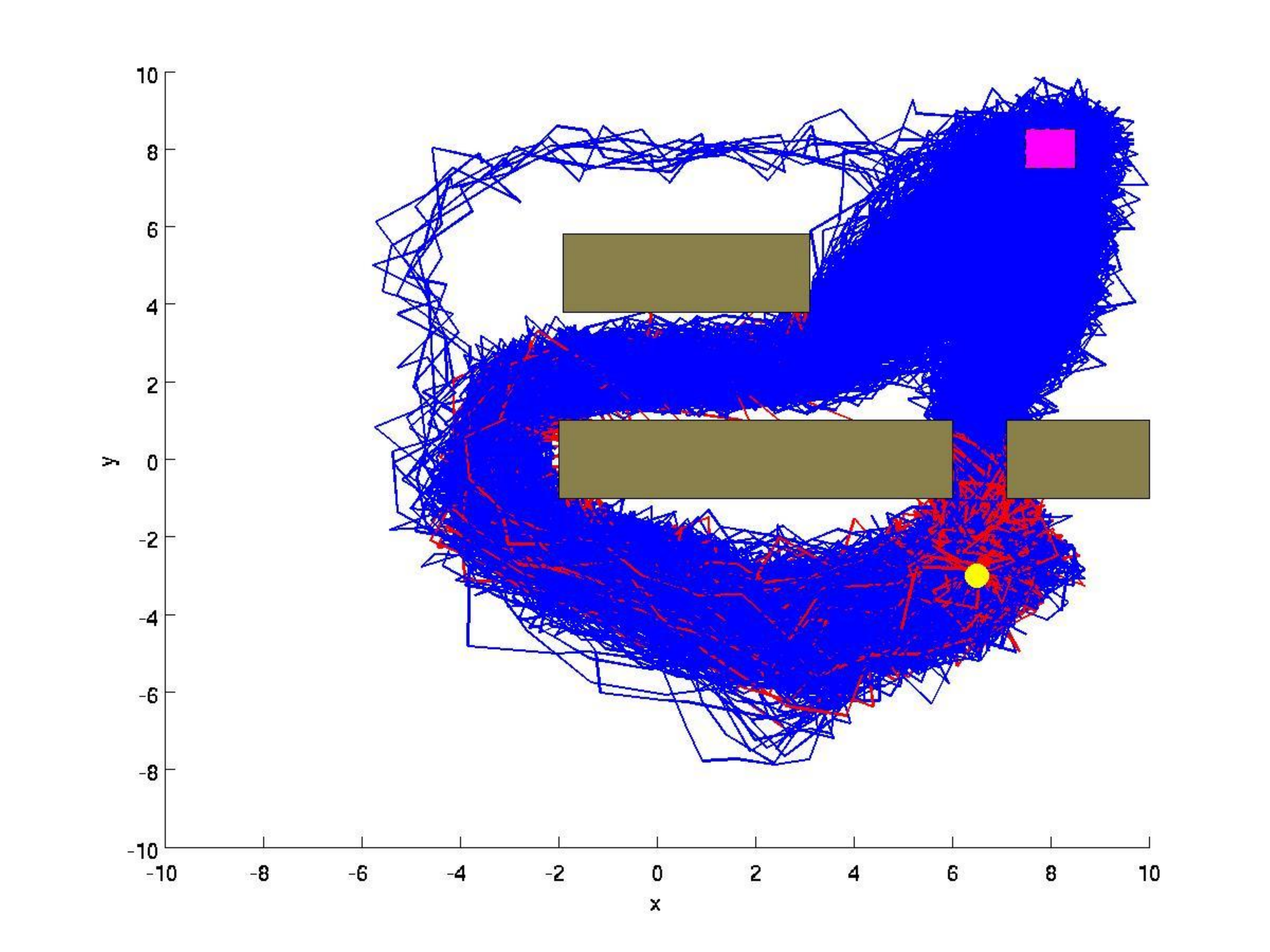}
  }
  \caption[Trajectories after 5000 iterations starting from $(6.5,-3)$.]{Trajectories after 5000 iterations starting from $(6.5,-3)$. In Figs.~\ref{figsubTraj1}-\ref{figsubTraj3} and Figs.~\ref{figsubTraj7}-\ref{figsubTraj9}, we show 50 trajectories resulting from a policy induced by $J_{4000}$ with different collision-probability thresholds ($\eta= 0.01, 0.05, 0.10, 0.20, 0.30, 0.40$). In Figs.~\ref{figsubTraj4}-\ref{figsubTraj6} and Figs.~\ref{figsubTraj10}-\ref{figsubTraj12}, we show $5000$ corresponding trajectories in the original state space $S$ with \textit{simulated collision probabilities} and \textit{average costs} in their captions. Trajectories that reach the goal region are plotted in blue, and trajectories that hit obstacles are plotted in red.}
  \label{figTraj}
\end{center}
\end{figure*}

In the following experiments, we used a computer with a 2.0-GHz Intel Core 2 Duo T6400 processor and $4$ GB of RAM. We controlled a system with stochastic single integrator dynamics to a goal region with free ending time in a cluttered environment. The dynamics is given by $dx(t)= u(t)dt + Fdw(t)$ where $x(t) \in \reals^{2}$, $u(t) \in \reals^2$, and $F=\left[ {\begin{array}{cc}
   0.5 & 0 \\
   0 & 0.5\\
  \end{array} } \right]$. The system stops when it collides with obstacles or reach the goal region. The cost function is the weighted sum of total energy spent to reach the goal $G$ at $(8,8)$, which is measured as the integral of square of control magnitude, and a terminal cost, which is $-1000$ for the goal region $G$ and $10$ for the obstacle region $\Gamma$, with a discount factor $\alpha=0.9$. The maximum velocity of the system in the x and y directions is one. At the beginning, the system starts from $(6.5,-3)$. Failure is defined as collisions with obstacles, and thus we use \textit{failure probability} and \textit{collision probability} interchangeably.

We first show how the extended iMDP algorithm constructs the sequence of approximating MDPs on $S$ over iterations in Fig.~\ref{figConstruct}. In particular, Figs.~\ref{figsubConstruct1}-\ref{figsubConstruct3} depict anytime policies on the boundary $S \times 1.0$ after 500, 1000, and 3000 iterations. Figures~\ref{figsubConstruct4}-\ref{figsubConstruct6} show the Markov chains created by anytime policies found by the algorithm on $\mathcal{M}_n$ after 200, 500 and 1000 iterations. We observe that the structures of these Markov chains are indeed random graphs that are (asymptotically almost-surely) connected to cover the state space $S$. As in the original version of iMDP, it is worth noting that the structures of these Markov chains can be constructed on-demand during the execution of the algorithm. 

The sequence of approximating MDPs on $S$ provides boundary values for the stochastic target problem as shown in Fig.~\ref{figBoundaryInfo}. In particular, Figs.~\ref{figBinfo1}-\ref{figBinfo3} shows a policy map, cost value function $J_{4000,1.0}$ and the associated collision probability function $\Upsilon_{4000}$ for the unconstrained problem after 4000 iterations. Similarly, Figs.~\ref{figBinfo4}-\ref{figBinfo6} show a policy map, the associated value function $J^{\gamma}_{4000}$, and the min-collision probability function $\gamma_{4000}$ after 4000 iterations. As we can see, for the unconstrained problem, the policy map encourages the system to go through the narrow corridors with low cost-to-go values and high probabilities of collision. In contrast, the policy map from the min-collision probability problem encourages the system to detour around the obstacles with high cost-to-go values and low probabilities of collision.

We now show how the extended iMDP algorithm constructs the sequence of approximating MDPs on the augmented state space $\overline{S}$. Figures~\ref{figsubConstruct7}-\ref{figsubConstruct9} show the corresponding anytime policies in $\overline{S}$ over iterations. In Fig.~\ref{figsubConstruct9}, we show the top-down view of a policy for states in $\overline{\mathcal{M}}_{3000} \backslash \mathcal{M}_{3000}$. Compared to Fig~\ref{figsubConstruct3}, we observe that the system will try to avoid the narrow corridors when the risk tolerance is low. In Figs.~\ref{figsubConstruct10}-\ref{figsubConstruct12}, we show the Markov chains that are created by anytime policies in the augmented state space. As we can see again, the structures of these Markov chains quickly cover $\overline{S}$ with (asymptotically almost-surely) connected random graphs. 

We then examine how the algorithm computes the value functions for the interior $D^o$ of the reformulated stochastic target problem in comparison with the value function of the unconstrained problem in Fig.~\ref{figValue}. Figure~\ref{figsubValue1}-\ref{figsubValue3} show approximate cost-to-go $J_n$ when the probability threshold $\eta_0$ is 1.0 for $n=200$, $2000$ and $4000$. We recall that the value functions in these figures form the boundary conditions on $S \times 1$, which is a subset of $\partial D$. In the interior $D^o$, Figs.~\ref{figsubValue7}-\ref{figsubValue9} present the approximate cost-to-go $J_{4000}$ for augmented states where their martingale components are $0.1$, $0.5$ and $0.9$. As we can see, the lower the martingale state is, the higher the cost value is -- which is consistent with the characteristics in Section~\ref{subsection:characterization}.

Lastly, we tested the performance of obtained anytime policies after 4000 iterations with different initial collision probability thresholds $\eta$. To do this, we first show how the policies of the unconstrained problem and the min-collision probability problem perform in Fig.~\ref{figUnMinTrajs}. As we can see, in the unconstrained problem, the system takes risk to go through one of the narrow corridors to reach the goal. In contrast, in the min-collision probability problem, the system detour around the obstacles to reach the goal. While there are about $49.27\%$ of $2000$ trajectories (plotted in red) that collide with the obstacles for the former, we observe no collision out of $2000$ trajectories for the latter. From the characteristics presented in Section~\ref{subsection:characterization} and illustrated in Fig.~\ref{figFeedback_Extra}, from the starting state $(6.5,-3)$, for any initial collision probability threshold $\eta$ above $0.4927$, we can execute the deterministic policy of the unconstrained problem.

In Fig.~\ref{figImerse}, we provide an example of controlled trajectories that are illustrated in Fig.~\ref{figFeedback_Extra} when the system starts from $(6.5,-3)$ with the failure probability threshold $\eta=0.4$. In this figure, the min-collision probability function $\gamma_{4000}$ is plotted in blue, and the collision probability function $\Upsilon_{4000}$ is plotted in green. Starting from the augmented state $(6.5,-3,0.40)$, the martingale state varies along controlled trajectories as a random parameter in a randomized control policy. When the martingale state is above $\Upsilon_{4000}$, the system follows a deterministic control policy obtained from the unconstrained problem.

Similarly, in Fig.~\ref{figTraj}, we show controlled trajectories for different values of $\eta$ ($0.01, 0.05, 0.10, 0.20, 0.30, 0.40$). In Figs.~\ref{figsubTraj1}-\ref{figsubTraj3} and Figs.~\ref{figsubTraj7}-\ref{figsubTraj9}, we show 50 trajectories resulting from a policy induced by $J_{4000}$ with different initial collision probability thresholds. In Figs.~\ref{figsubTraj4}-\ref{figsubTraj6} and Figs.~\ref{figsubTraj10}-\ref{figsubTraj12}, we show $5000$ corresponding trajectories in the original state space $S$ with reported \textit{simulated collision probabilities} and \textit{average costs} in their captions. Trajectories that reach the goal region are plotted in blue, and trajectories that hit obstacles are plotted in red. These simulated collision probabilities and average costs are shown in Table~\ref{table:statistics}. As we can see, the lower the threshold is, the higher the average cost is as we expect. When $\eta=0.01$, the average cost is $-19.42$ and when $\eta=1.0$, the average cost is $-125.20$. 

More importantly, the simulated collision probabilities follow very closely the values of $\eta$ chosen at time $0$. In Fig.~\ref{figFailureRatio}, we plot these simulated probabilities for the first $N$ trajectories where $N \in [1,5000]$ to show that the algorithm fully respects the bounded failure probability. Thus, this observation indicates that the extended iMDP algorithm is able to manage the risk tolerance along trajectories in different executions to minimize the expected costs using feasible and time-consistent anytime policies.

\begin{table}
	\centering
	\caption{Failure ratios and average costs for Fig.~\ref{figFailureRatio}.}   \label{table:statistics}
	\vskip-0.15cm
	\begin{tabular}{c|c|c}
		 {\color{black}{\textbf{ $\eta$ } }} & {\color{black}{ \textbf{ Failure Ratio}}} & {\color{black}{\centering \textbf{Average Cost}}} \\
		\hline
		1.00 & 0.4927 & -125.20 \\
		\hline
		0.40 & 0.4014 & -115.49 \\
		\hline
		0.30 & 0.2819 & -76.80 \\
		\hline
		0.20 & 0.1560 & -65.81 \\
		\hline
		0.10 & 0.1024 & -58.00 \\
		\hline
		0.05 & 0.0420 & -42.53 \\
		\hline
		0.01 & 0.0084 & -19.42 \\
		\hline
		0.001 & 0.0000 & -18.86 \\
		\hline
	\end{tabular}
	\vskip-0.5cm
\end{table}

\section{Conclusions} \label{section:conclusions}
We have introduced and analyzed the extension of the incremental Markov Decision Process (iMDP) algorithm for stochastic optimal control subject to bounded failure probabilities for initial states. We present here the martingale approach that diffuses the probability constraint into a martingale. The martingale stands for the level of risk tolerance that is contingent on available information over time. The approach transforms the probability-constrained problem into an equivalent stochastic target problem with the augmented state and control spaces. The boundary conditions for the transformed problem is, however, unspecified. The extended iMDP algorithm incrementally computes the boundary values and any-time feedback control policies for the transformed problem using asynchronous value iterations. The returned policies can be considered as randomized policies in the original state space. Effectively, the extended iMDP algorithm provides probabilistically-sound and asymptotically-optimal control policies for the class of stochastic control problems with bounded failure-probability constraints.

The future extension of the work is broad. We intend incorporate logical rules expressed as temporal logic constraints to achieve high degree of autonomy for systems to operate safely in uncertain and highly dynamic environments with complex mission specifications. We also plan to implement the algorithm outlined in this paper on robotic platforms for practical demonstration. 

\section*{ACKNOWLEDGMENTS}
This work was partially supported by the National Science Foundation grant CNS-1016213 and the Army Research Office MURI grant W911NF-11-1-0046.

\bibliographystyle{IEEEtran}
\bibliography{IEEEabrv,mybib}

\end{document}